\documentclass[a4]{elsarticle}

\usepackage{bussproofs}
\usepackage{proof}
\usepackage{amsmath}
\usepackage{amsthm}
 
\usepackage{amsfonts}
\usepackage{amssymb}
\usepackage{latexsym}
\usepackage{epsf}
\usepackage{bbold}

\usepackage{pstricks}
\usepackage{amscd}

\usepackage{url} 
\usepackage{hyperref} 

\newtheorem{lemma}{Lemma}
\newtheorem{theorem}{Theorem}

\theoremstyle{definition}
\newtheorem{example}{Example}

\newcommand{\union}{\cup}

\newcommand{\NN}{\mathbf{N}}
\newcommand{\ZZ}{\mathbf{Z}}
\newcommand{\QQ}{\mathbf{Q}}

\newcommand{\II}{\mathrm{I\!I}}

\newcommand{\A}{\mathbf{C}}
\newcommand{\B}{\mathbf{B}}
\newcommand{\C}{\mathbf{S}}
\newcommand{\D}{\mathbf{D}}
\newcommand{\G}{\mathbf{G}}

\newcommand{\AP}{\mathbf{AP}}
\newcommand{\AI}{\mathbf{AI}}

\newcommand{\AIB}{\mathbf{AIB}}

\newcommand{\RG}{\mathbf{R}}
\newcommand{\LG}{\mathbf{L}}

\newcommand{\con}{\mathsf{C}}
\newcommand{\cona}{\con_{\alpha,\rho}}
\newcommand{\conad}{\con_{\alpha,\rho'}}

\newcommand{\tail}{\mathsf{tail}}
\newcommand{\head}{\mathsf{head}}
\newcommand{\minus}{\mathsf{minus}}
\newcommand{\identity}{\mathsf{id}}
\newcommand{\sgh}{\mathsf{sgh}}
\newcommand{\sgt}{\mathsf{sgt}}
\newcommand{\stog}{\mathsf{stog}}

\newcommand{\nh}{\mathsf{nh}}

\newcommand{\inv}{\mathsf{inv}}
\newcommand{\stepone}{\mathsf{step1}}
\newcommand{\steptwo}{\mathsf{step2}}
\newcommand{\stepthree}{\mathsf{step3}}
\newcommand{\stepfour}{\mathsf{step4}}

\newcommand{\Fun}{\mathbf{Fun}}
\newcommand{\FV}{\mathrm{FV}} 
\newcommand{\Left}{\mathbf{Left}}

\newcommand{\Nil}{\mathbf{Nil}}
\newcommand{\Pair}{\mathbf{Pair}} 

\newcommand{\Right}{\mathbf{Right}}

\newcommand{\SD}{\mathbf{SD}}

\newcommand{\False}{\mathbf{False}}

\newcommand{\tent}{\mathbf{t}}
\newcommand{\case}{\mathbf{case}}

\newcommand{\of}{\mathbf{of}}

\newcommand{\botexp}{\mathbf{\bot}}
\newcommand{\lett}{\mathbf{let}}
\newcommand{\inn}{\mathbf{in}}

\newcommand{\re}{\mathbf{r}}

\newcommand{\eqdef}{\stackrel{\mathrm{Def}}{=}}
\newcommand{\eqmu}{\stackrel{\mu}{=}}
\newcommand{\eqnu}{\stackrel{\nu}{=}}
\newcommand{\eqrec}{\stackrel{\mathrm{rec}}{=}}

\newcommand{\ire}[2]{#1\,\mathbf{r}\,#2}

\newcommand{\inl}[1]{\Left(#1)}
\newcommand{\inr}[1]{\Right(#1)}

\newcommand{\projl}{\mathbf{\pi_{Left}}}
\newcommand{\projr}{\mathbf{\pi_{Right}}}
\newcommand{\all}[1]{\forall #1\,}
\newcommand{\ex}[1]{\exists #1\,}

\newcommand{\IFP}{\mathrm{IFP}}  
\newcommand{\RIFP}{\mathrm{RIFP}}  

\newcommand{\rec}{\mathbf{rec}}

\newcommand{\tfix}[2]{\mathbf{fix}\,#1\,.\,#2}

\newcommand{\ssp}{\rightsquigarrow}
\newcommand{\newprintp}{\overset{\mathrm{p}}{\ssp}}
\newcommand{\newprintptr}{\mathbin{\stackrel{\mathrm{p}}{\ssp}\kern-.25em{}^*}}

 \newcommand{\muprint}{\print^{\!\!\!\!\!\!\!\!\!\!\!\ {}^{\scriptstyle \mu} \,\,\,\,}}
 \newcommand{\mubprint}{\print^{\!\!\!\!\!\!\!\!\!\!\!\!\!\ {}^{\scriptstyle \mu\bot} \,}}
 \newcommand{\nuprint}{\print^{\!\!\!\!\!\!\!\!\!\!\!\ {}^{\scriptstyle\nu} \,\,\,\,}}
 \newcommand{\nubprint}{\print^{\!\!\!\!\!\!\!\!\!\!\!\!\!\ {}^{\scriptstyle\nu\bot} \,}}

\newcommand{\print}{\Longrightarrow}

\newcommand{\valu}[2]{[\![#1]\!]#2}
\newcommand{\val}[1]{[\![#1]\!]}
\newcommand{\rk}{\mathbf{rk}}
\newcommand{\cl}[1]{\mathrm{Pr}(#1)}

\newcommand{\tdata}{E_\mathrm{t}}
\newcommand{\fdata}{E_\mathrm{f}}
\newcommand{\ftdata}{E_{\mathrm{ft}}}
\newcommand{\data}{E}

\newcommand{\rea}{\mathbf{R}}
\newcommand{\reah}{\mathbf{H}}

\newcommand{\dle}{\sqsubseteq}
\newcommand{\dlee}{\dle_{\data}} 

\newcommand{\reali}[1]{\tilde{#1}}

\newcommand{\idty}{\mathbf{id}}

\newcommand{\acc}{\mathbf{Acc}}

\newcommand{\wfi}{\mathbf{WfI}}
\newcommand{\prog}{\mathbf{Prog}}

\newcommand{\CL}{\mathbf{CL}}
\newcommand{\IND}{\mathbf{IND}}

\newcommand{\COCL}{\mathbf{COCL}}
\newcommand{\COIND}{\mathbf{COIND}}

\newcommand{\sci}{\mathbf{SCI}}
\newcommand{\hsci}{\mathbf{HSCI}}

\newcommand{\si}{\mathbf{SI}}
\newcommand{\hsi}{\mathbf{HSI}}

\newcommand{\mon}{\mathsf{mon}}

\newcommand{\mycomment}[1]{}

\newcommand{\bs}{\Downarrow}
\newcommand{\bigstep}{\Downarrow}

\newcommand{\allex}{\Diamond}
\newcommand{\munu}{\Box}

\newcommand{\dsort}{\delta}

\newcommand{\funsum}[2]{[#1+#2]}
\newcommand{\funpair}[2]{\langle #1,#2\rangle}

\newcommand{\pcv}[1]{\hat{#1}}

\newcommand{\valg}[1]{\val{#1}_{\G}}

\newcommand{\valtg}[1]{\valg{#1}}

\newcommand{\GC}{\mathbf{GC}}

\newcommand{\BT}{\mathbf{BT}}
\newcommand{\BTnc}{\mathbf{BT_{nc}}}

\newcommand{\pat}{\mathbf{Path}}
\newcommand{\less}{\prec}

\newcommand{\one}{\mathbf{1}}

\newcommand{\lang}{\mathcal{L}}
\newcommand{\ax}{\mathcal{A}}

\newcommand{\odata}{\Phi_{\bot}}
\newcommand{\otdata}{\Phi}

\newcommand{\oless}{\Phi^2_{\bot}}
\newcommand{\oeq}{\Phi^2_{\bot,\bot}}
\newcommand{\oteq}{\Phi^2}

\newcommand{\appr}[2]{\mathrm{appr}(#1,#2)}

\newcommand{\peq}[2]{\mathrm{eq}(#1,#2)}
\newcommand{\teq}[2]{\mathrm{teq}(#1,#2)}

\newcommand{\ocl}{\Phi^{\mathrm{op}}}

\newcommand{\isfun}{\mathbf{IsFun}}

\newcommand{\dom}{\mathrm{dom}}

\newcommand{\andi}[2]{\land^+(#1,#2)}
\newcommand{\andel}[1]{\land_l^-(#1)}
\newcommand{\ander}[1]{\land_r^-(#1)}
\newcommand{\oril}[2]{\lor_{l,#2}^+(#1)}
\newcommand{\orir}[2]{\lor_{r,#2}^+(#1)}
\newcommand{\ore}[3]{\lor^-(#1,#2,#3)}
\newcommand{\impi}[2]{\to^+_{#1}(#2)}
\newcommand{\impe}[2]{\to^{-}(#1,#2)}
\newcommand{\alli}[2]{\forall^{+}_{#1}(#2)}
\newcommand{\alle}[2]{\forall^{-}_{#2}(#1)}
\newcommand{\exi}[3]{\exists_{#3,#1}^+(#2)}
\newcommand{\exe}[2]{\exists^-(#1,#2)}
\newcommand{\clos}[1]{\mathbf{Cl}_{#1}}
\newcommand{\induct}[2]{\mathbf{Ind}_{#1}(#2)}
\newcommand{\indu}[1]{\mathbf{Ind}(#1)}
\newcommand{\induprime}[1]{\mathbf{Ind'}(#1)}
\newcommand{\cocl}[1]{\mathbf{CoCl}_{#1}}
\newcommand{\coinduct}[2]{\mathbf{CoInd}_{#1}(#2)}
\newcommand{\coind}[1]{\mathbf{CoInd}(#1)}
\newcommand{\coindprime}[1]{\mathbf{CoInd'}(#1)}
\newcommand{\congr}[3]{\mathbf{Cong}_{#3}(#1,#2)}
\newcommand{\refl}[1]{\mathbf{Refl}_{#1}}

\newcommand{\hsindu}[1]{\mathbf{HSInd}(#1)}
\newcommand{\scoind}[1]{\mathbf{SCoInd}(#1)}
\newcommand{\hscoind}[1]{\mathbf{HSCoInd}(#1)}
\newcommand{\sinduct}[2]{\mathbf{SInd}_{#1}(#2)}
\newcommand{\hsinduct}[2]{\mathbf{HSInd}_{#1}(#2)}
\newcommand{\scoinduct}[2]{\mathbf{SCoInd}_{#1}(#2)}
\newcommand{\hscoinduct}[2]{\mathbf{HSCoInd}_{#1}(#2)}

\newcommand{\subd}{\Delta}
\newcommand{\subdom}[1]{#1 \lhd D}
\newcommand{\subdoms}{\lhd D}

\newcommand{\tval}[2]{D^{#2}_{#1}}
\newcommand{\ftyp}[2]{#1 \Rightarrow #2}

\newcommand{\ep}[1]{\mathbf{ep}(#1)}
\newcommand{\epp}[1]{\mathbf{ep'}(#1)}
\newcommand{\epph}[1]{\mathbf{eph'}(#1)} 
\newcommand{\pt}[1]{\mathbf{pt}(#1)}  
\newcommand{\monproof}[1]{\mathbf{Mon}_{#1}}

\newcommand{\monprop}[2]{\mathrm{Mon}_{#1}(#2)}

\newcommand{\timesd}[1]{\nabla(#1)} 

\newcommand{\tri}{\mathbf{3}}
\newcommand{\bool}{\mathbf{2}}
\newcommand{\nat}{\mathbf{nat}}
\newcommand{\rat}{\mathbf{rat}}

\newcommand{\subrank}[2]{#1\upharpoonright #2}
\newcommand{\depth}[2]{\mathbf{depth}_{#1}(#2)}

\newcommand{\exone}{\exists_{1}}
\newcommand{\exor}{\oplus}
\newcommand{\bigexor}{\bigoplus}

\newcommand{\stream}[1]{#1^\omega}

\newcommand{\halt}{\mathrm{Halt}}

\newcommand{\Haskell}[1]{{\mathsf H}(#1)}

\newcommand{\too}{\ \to \ }
\newcommand{\toot}{\ \leftrightarrow \ }

\newcommand{\adummy}[1]{\Delta(#1)}  

\begin{document}

\title{Intuitionistic Fixed Point Logic\tnoteref{t1}} 
\tnotetext[t1]{
{\protect\includegraphics[scale = 0.05]{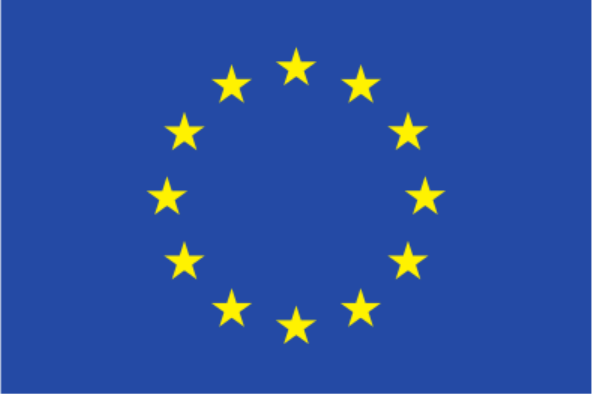}}
This work was supported by the International Research Staff
  Exchange Scheme (IRSES) No.~612638 CORCON and No.~294962 COMPUTAL of
  the European Commission, the JSPS Core-to-Core Program, A. Advanced
  research Networks and JSPS KAKENHI Grant Number 15K00015 as well as 
  the European Union’s Horizon 2020 research and innovation programme 
  under the Marie 
Sklodowska-Curie 
grant agreement No.~731143.}

 \author[1]{Ulrich Berger\corref{cor1}}
 \ead{u.berger@swansea.ac.uk}

\author[2]{Hideki Tsuiki} \ead{tsuiki.hideki.8e@kyoto-u.ac.jp}
\cortext[cor1]{Corresponding author}

\address[1]{Department of Computer Science, Swansea University, Swansea, United Kingdom}

\address[2]{Graduate School of Human and Environmental Studies, Kyoto University, Yoshida-Nihonmatsu, Kyoto, Japan}

\begin{frontmatter}


\begin{abstract}
  { 
We study the system $\IFP$ of 
intuitionistic fixed point logic, an extension of 
intuitionistic first-order logic by strictly positive 
inductive and coinductive definitions.
We define a realizability interpretation of $\IFP$ and use it to extract 
computational content from proofs about abstract structures specified by 
arbitrary classically true disjunction free formulas. 
The interpretation is shown to be sound with respect to a domain-theoretic 
denotational semantics and a corresponding lazy operational semantics
of a functional language for extracted programs.
We also show how extracted programs can be {translated} into Haskell.
As an application we extract a program converting the signed digit 
representation of real numbers to infinite Gray code from a proof
of inclusion of the corresponding coinductive predicates.
}
\end{abstract}

\begin{keyword}
Proof theory \sep
realizability \sep
program extraction \sep
induction \sep
coinduction \sep
exact real number computation
\end{keyword}

\end{frontmatter}

\tableofcontents

\section{Introduction}
\label{sec-introduction}

According to the Brouwer-Heyting-Kolmogorov interpretation of 
constructive logic, formulas correspond to data types and proofs 
to constructions of objects of these data types 
\cite{Troelstra73,MartinLoef84,Constable86,Troelstra88,TroelstraSchwichtenberg96,SchwichtenbergWainer12}. 
Moreover, by the Curry-Howard correspondence constructive
proofs can be directly represented in a typed $\lambda$-calculus such that 
proof normalization is modelled by $\beta$-reduction.
This tight connection between logic and computation
has led to a number of implementations of proof systems that support 
the extraction of programs from constructive proofs,  
e.g.~PX~\cite{Hayashi88},
Nuprl~\cite{Constable86},
Coq~\cite{CoqProofAssistant},
Minlog~\cite{SchwichtenbergMinlog06,BergerMiyamotoSchwichtenbergSeisenberger11}, 
Isabelle/HOL~\cite{Berghofer03},
Agda~\cite{Agda}.
In general, program extraction is
restricted to proofs about structures that are constructively given.
This can be considered a drawback since it excludes 
abstract mathematics done on a purely axiomatic basis. 
This paper introduces the formal system $\IFP$ of 
\emph{Intuitionistic Fixed Point Logic} as a basis for 
program extraction from proofs that does not suffer from this limitation.  
Preliminary versions of the system were presented in 
\cite{BergerCSL09,Berger10,Berger11,SeisenBerger12,BergerPetrovska18}. 

$\IFP$ is an extension of first-order logic by inductive and coinductive 
definitions, i.e., predicates defined as least and greatest fixed points
of strictly positive operators.
Program extraction is performed via a `uniform' 
realizability interpretation.
Uniformity concerns the 
interpretation of quantifiers: A formula $\forall x\,A(x)$ is realized
uniformly by one object $a$ that realizes $A(x)$ for all $x$, so $a$
may not depend on $x$. Dually, a formula $\exists x\,A(x)$ is realized
uniformly by one object $a$ that realizes $A(x)$ for some $x$, so $a$
does not contain a witness for $x$. The usual interpretations of 
quantifiers may be recovered by relativization, 
$\forall x\,(D(x) \to A(x))$ and $\exists x\,(D(x) \land A(x))$,
for a predicate $D$ that specifies that $x$ has some
  concrete representation.
The uniform interpretation of quantifiers makes $\IFP$ 
classically inconsistent with the scheme `realizability implies 
truth'~(see the remark after Lemma~\ref{lem-realizability} 
in Sect.~\ref{sec-realizability}).
The Minlog system~\cite{SchwichtenbergMinlog06}, which also
supports program extraction based on realizability, does permit a uniform
interpretation of quantifiers as well
but differs from $\IFP$ in other respects,
for example the treatment of inductive and coinductive definitions.

Besides the support of proofs about abstract structures
on an axiomatic basis, $\IFP$ has further features that distinguishes it
from other approaches to program extraction.
\emph{Classical logic\/}:
Although $\IFP$ is based on intuitionistic logic a fair amount of
classical logic is available. For example, 
soundness of realizability holds in the presence of 
any disjunction-free axioms that are classically true.
typical example is stability of equality, 
$\forall x,y\,(\neg\neg x=y \to x=y)$.
\emph{Partial computation\/}:
Like the majority of programming languages, $\IFP$'s language of extracted
programs admits general recursion and therefore partial, i.e., 
nonterminating computation. This makes it possible to extract
data representations that are inherently partial, such as 
infinite Gray code~\cite{Gianantonio99, Tsuiki02} 
(see also Sect.~\ref{sec-realnumbers}).
\emph{Infinite computation\/}:
Infinite data, as they naturally occur in exact real number computation,
can be represented by infinite computations. This is achieved by an
operational semantics where computations may continue forever outputting
arbitrary close approximations to the complete
(infinite) result at their finite stages (Sect.~\ref{sec-opsem}).
\emph{Haskell output\/}:
Extracted programs are typable and can be translated into executable Haskell code in a straightforward way.

\emph{Related work: Minlog.}
The motivation for this article mainly stems from recent developments in the
Minlog proof system~\cite{SchwichtenbergMinlog06}. Minlog implements a 
formal system which, from its very conception, is a constructive 
theory of computable objects and functionals with an effective 
domain-theoretic semantics~\cite{SchwichtenbergWainer12}.
In order to increase the expressiveness of the logic and the flexibility 
of program extraction this system has been extended by an elaborate 
`decoration' mechanism for the logical operations that allows for a 
fine control of computational content (this extension is also described 
in~\cite{SchwichtenbergWainer12}).
For example, an existential quantifier can be decorated as `computational' 
or `non-computational' which causes the extracted program to include the 
witnessing term or not. Since in the non-computational case no witness
is required, the range of the quantified variable no longer needs to be 
effectively (i.e. domain-theoretically) given but may be an abstract 
mathematical structure.
This new possibility of including abstract structures in Minlog
formalizations triggered the present article which studies the 
implications and the potential of a computationally meaningful
theory of abstract structures in isolation.
Minlog's `non-computational' decoration corresponds to the uniform
realizability interpretation of $\IFP$ mentioned earlier.
There are some differences between Minlog and $\IFP$ though. 
For example, regarding the logical system, in Minlog all logical operations 
except implication and universal quantification are defined in terms of clausal 
inductive definitions while in $\IFP$ they are primitive and inductive 
definitions are not in the format of clauses. Regarding computational content, 
Minlog's realizers are typed and realizability is defined in the style
of Kreisel's modified realizability~\cite{Kreisel59} whereas in $\IFP$ 
realizers are untyped and realizability is 
closer to Kleene~\cite{Kleene45} (albeit $\IFP$ realizers are not numbers 
but domain elements denoted by functional programs).

\emph{PX\/}. Another related system is PX~\cite{Hayashi88} which is based on Feferman's
system $T_0$ of explicit mathematics~\cite{Feferman79} and uses a version of realizability 
with truth to extract untyped programs from proofs. The main differences to $\IFP$ are
that PX has a fixed, constructively given, model similar to LISP expressions and treats
quantifiers in the usual `non-uniform' way. PX supports positive inductive definitions,
however, restricted to operators without computational content.

\emph{Further related work\/}.
Theories of inductive and coinductive definitions have been studied
extensively in the past. The proof-theoretic strength of classical iterated
inductive definitions has been determined 
in~\cite{BuFePoSi81}. 
A proof-theoretic analysis of a stronger system that is close to $\IFP$,
but based on classical logic, has been given 
in~\cite{Moellerfeld03}. 
In~\cite{Tupailo04} it was shown that
the proof-theoretic strength does not change if the base system is 
changed to intuitionistic logic.
Inductive definitions have also been studied in the context of 
constructive set theory~\cite{Aczel77,rathjen05a}, 
type theory~\cite{dybjersetzer:2003:indrekjour,nordvallforsbergSetzer2010inductiveinductive} and explicit mathematics~\cite{GlassRathjenSchlueter97}. 
In~\cite{AvigadTowsner09} and \cite{Zucker73}, Inductive definitions are
related to theories of finite type in the framework of G\"odel's
Functional Interpretation.
Propositional logics for inductive and coinductive definitions interpreted 
on (finite) labelled transition 
are known as 
\emph{modal $\mu$-calculi}~\cite{Kozen83,BradfieldStirling07}.
These systems are based on classical logic and are mainly concerned 
with determining the computational complexity of definable properties
aiming at applications in automatic program verification systems.
Computational aspects of induction and coinduction 
(coiteration and corecursion), in particular
questions regarding termination, are studied widely
in the context of inductive and coinductive types. 
The strongest and most far reaching normalization 
can be found in~\cite{Mendler91} and~\cite{Matthes01}.
A programming language for real numbers extending PCF has been studied 
in~\cite{Escardo96}. 
It has a small step operational semantics
that permits the incremental computation of digits, similar to our semantics
in Sect.~\ref{sec-opsem}.
Logical, computational, semantical and category-theoretical aspects 
of coinduction are studied in the context of 
coalgebra~\cite{JacobsRutten97,KupkeKurzPattinson04}. 
The representation of coinductive types in dependent type theories
and the associated problems are an intensive object of 
study~\cite{Geuvers92,Coquand94,HancockSetzer03,AbelPientkaSetzer13,BergerSetzer18}.
The computational complexity of corecursion has been studied 
in~\cite{LeivantRamyaa11}.
Realizability interpretation related to the one for $\IFP$
were also studied in~\cite{Hayashi88,Tatsuta98,Miranda-Perea05,BauerBlanck09}
(see the introduction
of~\cite{Berger10} for a discussion of similarities and differences).
In Constructive Analysis~\cite{BishopBridges86} and 
Computable Analysis~\cite{Weihrauch00} 
one works with represented structures 
and explicitly manipulates and reasons about these
representations. In contrast, in $\IFP$ representations remain 
implicit and are made explicit only through realizability.
Proof Mining~\cite{Kohlenbach08} treats real numbers as a represented space
but one can extract effective bounds from ineffective proofs
about abstract spaces without a constructive representation 
(see e.g.~\cite{Kohlenbach05,GerhardyKohlenbach08}).

\emph{Overview of the paper\/}.
\emph{Section~\ref{sec-ifp}} introduces the system $\IFP$. Among 
other things, the usual principle of wellfounded induction
is exhibited as an instance of strictly positive induction 
and shown to be strengthened by an abstract form of
Brouwer's Thesis.
The definitions are illustrated by 
an axiomatic specification of the 
real numbers and a definition of the natural numbers as an 
inductively defined subset of the reals.
Special attention is paid to a formulation of the Archimedean property as
an induction principle.

\emph{Section~\ref{sec-realizability}} begins with a definition of a Scott 
domain $D$ 
that serves as the semantic domain of simple untyped functional programming 
language with constructors and unrestricted recursion.
Then we introduce simple recursive types denoting sub domains of $D$ 
that serve as spaces of potential realizers of formulas
and show that the expected typing rules are valid.
We extend $\IFP$
to a system $\RIFP$ that contains new sorts
$\dsort$ and $\subd$ for elements and subdomains of $D$ as well as new terms, 
called programs and types, for denoting 
them. 
This is followed by a formal realizability interpretation
of $\IFP$ in $\RIFP$. The interpretation is optimized by
exploiting the fact that Harrop formulas, which
are formulas that do not contain a disjunction at a strictly positive
position, have trivial realizers (similar optimizations are available 
in the Minlog system).

In Section~\ref{sec-soundness} we prove the
Soundness Theorem~(Thm.~\ref{thm-soundness}) which shows that from an 
$\IFP$ proof of a formula $A$ from nc axioms one can extract a program provably realizing $A$.
For the proof we use an intermediate system $\IFP'$ which in the rules
for induction and coinduction for the least and greatest fixed point of 
an operator $\Phi$ requires in addition a proof of monotonicity of $\Phi$.
We provide a recursive definition of the program extracted from
an $\IFP$ derivation and give explicit constructions 
of realizers for derived principles such as wellfounded induction and its
variants introduced in Sect.~\ref{sec-ifp}.

\emph{Section~\ref{sec-realnumbers}} is devoted to a case study on exact 
real number computation that utilizes all the concepts introduced so far.
It is shown that the well-known signed digit representation and also
the infinite (and partial!) Gray code representation
can be obtained through realizability from simple coinductively
defined predicates $\C$ and $\G$.  
A detailed $\IFP$ proof that $\C$ is contained in $\G$ is given and
from it a program is extracted that converts the signed digit 
representation into infinite Gray code. 
The equivalence of the extracted program with the one given in~\cite{Tsuiki02} is
also proved, which guarantees the correctness of the original program.

\emph{Section~\ref{sec-opsem}} introduces an operational 
semantics 
of programs that is able to capture infinite computation. 
While the First Adequacy Theorem~(Thm.~\ref{thm-adequacy}) states that
an inductively defined bigstep reduction relation $\muprint$ captures 
the semantics of programs $M$ with a finite total denotation, i.e.,
$M\muprint a$ iff $a = \val{M}$,
the Second Adequacy Theorem~(Thm.~\ref{thm-adequacytwo}) establishes an 
equivalence of programs that have a possibly \emph{infinite} and
\emph{partial} denotational semantics 
with a \emph{small step} reduction relation. This means that it is 
possible to incrementally
compute arbitrary close approximations to a program that has an infinite
value.
Sect.~\ref{sec-opsem} closes with a concrete example of 
infinite computation using a concrete instance of the results 
of Sect.~\ref{sec-realnumbers}.

\emph{Section~\ref{sec-conclusion}} concludes the paper with a summary 
and a discussion of open problems and directions for further work.


\section{Intuitionistic fixed point logic}  
\label{sec-ifp}
We introduce the logical system $\IFP$ of    
\emph{intuitionistic fixed point logic} as a basis for the formalization
of proofs which can be subject to program extraction. $\IFP$ can be viewed as
a subsystem of second-order logic with its standard classical set-theoretic
semantics. We first define the language and the proof rules of $\IFP$ and then
draw some simple consequences demonstrating that $\IFP$ includes
common principles such as wellfounded induction and permits a natural
formalization of real numbers as a real closed Archimedean field.
\subsection{The formal system $\IFP$}
\label{sub-ifp}
$\IFP$ is an extension of intuitionistic first-order predicate logic by
least and greatest fixed points of strictly positive operators.
Rather than a fixed system $\IFP$ is a \emph{schema} for
a family of systems suitable to formalize different mathematical fields.  
An \emph{instance} of $\IFP$ is given by a \emph{many-sorted first-order 
language} $\lang$ and a set of \emph{axioms} $\ax$ described below.
Hence $\lang$ consists of
\begin{itemize}
\item[(1)] \emph{Sorts} $\iota, \iota_1,\ldots$ as names for spaces 
of abstract mathematical objects. 
\item[(2)] \emph{Terms} $s,t,\ldots$ 
with a notion of free variables and a notion of
substitution.  First order terms are the main example but
we will also consider term languages with binding mechanism 
(Sect.~\ref{sec-realizability}).
\item[(3)] \emph{Predicate constants}, each of fixed arity $(\vec \iota)$. 
\end{itemize}
Relative to a language $\lang$ we define simultaneously 
\begin{description}
\item[Formulas] $A,B$: 
\emph{Equations} $s=t$ ($s,t$ terms of the same sort),
$P(\vec t)$ ($P$ a predicate which is 
not an abstraction, $\vec t$ a tuple of terms whose sorts fit the arity of
$P$), \emph{conjunction} $A\land B$, \emph{disjunction} $A\lor B$, 
\emph{implication} $A\to B$, \emph{universal and existential quantification}
$\forall x\,A$, $\exists x\,A$. 

\item[Predicates] $P,Q$: 
\emph{Predicate variables} $X,Y,\ldots$ (each of fixed arity),
\emph{predicate constants}, 
\emph{abstraction} 
$\lambda \vec x\,A$ (arity given by the sorts of the variable tuple $\vec x$), 
$\mu(\Phi)$, $\nu(\Phi)$ (arities = arity of $\Phi$).
\item[Operators] $\Phi$: $\lambda X\,P$ where $P$ must be
strictly positive in $X$ (see below) and the arities of $X$ and $P$
must coincide. The arity of $\lambda X\,P$ is this common arity.
\end{description}
\emph{Falsity} is defined as $\False \eqdef\mu(\lambda X\,X)()$ where $X$ is a 
predicate variable of arity $()$.

By an \emph{expression} we mean a formula, predicate, or operator.
When considering an expression it is tacitly assumed that the arity
of a predicate and the sorts of terms it is applied to fit.
The set of free object variables and the set of free predicate variables 
of an expression are defined as expected. 

An occurrence of an expression $E$ is
\emph{strictly positive (s.p.)} in an expression $F$ if that occurrence
is not within the premise of an implication.
A predicate $P$ is strictly positive in a predicate variable $X$
if every occurrence of $X$ in $P$ is strictly positive.
The requirement of strict positivity could be easily relaxed to mere 
positivity. However, since non-strict positivity will not be required at any 
point, but would come at the cost of more complicated proofs, we refrain from 
this generalization. A similar remark applies to the strict positivity
condition for fixed point types in Sect.~\ref{sub-types}.

We adopt the following \emph{notational conventions\/}.
Application of an abstraction 
to terms, $(\lambda \vec x\,A)(\vec t)$, is defined as $A[\vec t/\vec x]$ 
(therefore $P(\vec t)$ is now defined for all predicates $P$ and 
terms $\vec t$ of fitting arity). 
Application of an operator $\Phi = \lambda X\,P$ to a predicate
$Q$, $\Phi(Q)$, is defined as $P[Q/X]$. 
Instead of $P(\vec t)$ we also write $\vec t \in P$ and
a definition $P \eqdef \mu(\Phi)$ will also be written 
$P \eqmu \Phi(P)$. 
The notation $P \eqnu \Phi(P)$ has a similar meaning.
If $\Phi = \lambda X \lambda \vec x\,A$, then we also write
$P(\vec x) \eqmu A[P/X]$ and
$P(\vec x) \eqnu A[P/X]$ instead of 
$P \eqdef \mu(\Phi)$ and $P \eqdef \nu(\Phi)$.
Inclusion of predicates (of the same arity), $P\subseteq Q$, is defined as 
$\forall \vec x\,(P(\vec x) \to Q(\vec x))$, intersection, $P\cap Q$, as
$\lambda \vec x\,(P(\vec x) \land Q(\vec x))$, and union, $P\cup Q$, as
$\lambda \vec x\,(P(\vec x) \lor Q(\vec x))$.
Pointwise implication, $P\Rightarrow Q$, is defined as
$\lambda \vec x\,(P(\vec x) \to Q(\vec x))$. 
Hence $P\subseteq Q$ is the same
as $\forall \vec x\,(P\Rightarrow Q)(\vec x)$.
Equivalence, $A \leftrightarrow B$, is defined as 
$(A \to B) \land (B\to A)$, 
and extensional equality of predicates, $P \equiv Q$, as
$P\subseteq Q \land Q \subseteq P$.

Negation, $\neg A$, is defined as $A \to \False$ and inequality, 
$t \neq s$, as $\neg (t=s)$. Bounded quantification, 
$\forall x \in A\,B(x)$ and $\exists x \in A\,B(x)$, is defined, 
as usual, as $\forall x\,(A(x) \to B(x))$ and 
$\exists x\,(A(x) \land B(x))$.
Exclusive `or' and unique existence are defined as 
$A \exor B \eqdef (A \lor B) \land \neg (A \land B)$,
$\exone x\,A(x) \eqdef \exists x\,\forall y\,(A(y) \leftrightarrow x = y)$.
If we write $E=E'$ for expressions that are not terms, 
we mean that $E$ and $E'$ are 
syntactically equal up to renaming of bound variables.

An expression is called \emph{non-computational (nc)} if it
is disjunction-free and contains no free predicate variables.  
Our realizability interpretation (Sect.~\ref{sec-realizability})
will be defined such that nc formulas do not carry computational content and
are interpreted by themselves. The reader may wonder why existential quantifiers
aren't banned from nc formulas as well. The reason is that quantifiers are interpreted
uniformly, in particular, existential quantifiers are not witnessed.
The existential quantifier in intuitionisitc arithmetic, 
which \emph{is} witnessed, can be expressed in $\IFP$ by 
$\exists x\,(\NN(x) \land A(x))$, where $\NN(x)$ means
`$x$ is a natural number' and the predicate $\NN$ is defined using disjunction 
(see Sect.~\ref{subsub-nat}).

The set of \emph{axioms} $\ax$ of an $\lang$-instance of $\IFP$ can be
any set of closed $\lang$-formulas. We denote that instance by 
$\IFP(\ax)$ leaving the language implicit since it is usually 
determined by the axioms.
In order for $\IFP(\ax)$ to admit a sound realizability 
interpretation (Sects.~\ref{sec-realizability} and \ref{sec-soundness}), 
the axioms in $\ax$ are required to be non-computational. 
The reason is that
this guarantees that they have trivial
computational content and are equivalent to their realizability 
interpretations, as will be explained in Sect.~\ref{sec-realizability}.
Therefore, it suffices that the axioms are true in the intended structure 
where ``true'' can be interpreted in the sense of classical logic.   
For example, if one is willing to accept a certain amount of classical 
logic (as we do in this paper) one may include in $\ax$ the 
\emph{stability axiom}
\[ \forall \vec x (\neg\neg A \to A) \]
for every nc formula $A$ with free variables $\vec x$.

The \emph{proof rules} of $\IFP$ include the usual natural deduction rules for 
intuitionistic first-order logic with equality 
(see below or e.g.~\cite{SchwichtenbergWainer12}). In addition there are 
the following rules for strictly positive induction and coinduction:
\[
  \infer[\CL(\Phi)]
{
  \Phi(\mu(\Phi)) \subseteq \mu(\Phi)   
  }
{
 }
\qquad
  \infer[\IND(\Phi,P)]
{
  \mu(\Phi) \subseteq P 
  }
{
 \Phi(P) \subseteq P
 }
\]
\[
  \infer[\COCL(\Phi) ]
{
  \nu(\Phi) \subseteq \Phi(\nu(\Phi))   
  }
{
 }
\qquad
  \infer[\COIND(\Phi,P)]
{
  P \subseteq \nu(\Phi) 
  }
{
  P \subseteq \Phi(P)
 }
\]
These rules can be applied in any context, that is, in 
the presence of free object and predicate variables as well as assumptions.

Intuitively, $\mu(\Phi)$ is the predicate defined inductively by
the rules encoded by the operator $\Phi$. For example natural numbers 
(viewed as a subset of the real numbers) can be defined as 
$\NN \eqdef \mu(\lambda X\,\lambda x\,(x=0\lor X(x-1)))$ corresponding to 
the rules `$\NN(0)$' and `if $\NN(x-1)$, then $\NN(x)$'.
The closure axiom $\CL(\Phi)$ expresses that $\mu(\Phi)$
is closed under the rules, the induction rule $\IND(\Phi,P)$ says that
$\mu(\Phi)$ is the smallest predicate closed under the rules 
(see also Sect.~\ref{subsub-nat}). 
Dually, $\nu(\Phi)$ is a coinductive predicate defined by `co-rules'. 
For example, the elements of a partial order which start an infinite
descending path can be characterized by the predicate 
$\pat = \nu(\lambda X\,\lambda x\,\exists y\,(y < x \land X(y)))$ 
(see also  Sect.~\ref{sub-wf}). 
Hence if $\pat(x)$, then $\pat(y)$ for some $y<x$ ($\COCL(\Phi)$), and 
$\pat$ is the largest predicate with that property ($\COIND(\Phi,P)$).

The existence of $\mu(\Phi)$ and $\nu(\Phi)$ is guaranteed, essentially,
by Tarski's fixed point theorem applied to the complete lattice of
predicates (of appropriate arity) ordered by inclusion
and the operator $\Phi$, which is monotone due to its strict positivity.
A simple but important observation is that $\mu(\Phi)$ and 
$\nu(\Phi)$ are (provably in $\IFP$)
{least and greatest}
 fixed points of $\Phi$, {respectively}. 
For example, $\mu(\Phi)\subseteq\Phi(\mu(\Phi))$ follows by induction:
One has to show $\Phi(\Phi(\mu(\Phi)))\subseteq \Phi(\mu(\Phi))$, which,
by monotonicity, follows from the closure axiom 
$\Phi(\mu(\Phi))\subseteq \mu(\Phi)$.

Induction and coinduction can be strengthened to the derivable principles 
of \emph{strong and half-strong induction and coinduction} (which will
be used in Sect.~\ref{sec-realnumbers}).
\[
  \infer[\si(\Phi,P)]{
  \mu(\Phi) \subseteq P 
}{
\Phi(P \cap \mu(\Phi)) \subseteq P 
}
\qquad
 \infer[\sci(\Phi,P)]{
  P \subseteq \nu(\Phi)
}{
P \subseteq \Phi(P \union \nu(\Phi))
}
\]
\[
  \infer[\hsi(\Phi,P)]{
  \mu(\Phi) \subseteq P 
}{
\Phi(P) \cap \mu(\Phi) \subseteq P 
}
\qquad
 \infer[\hsci(\Phi,P)]{
  P \subseteq \nu(\Phi)
}{
P \subseteq  \Phi(P) \union \nu(\Phi)
}
\]
It is clear that these proof rules are indeed strengthenings of ordinary 
s.p.\ induction since their premises are weaker due to the inclusions
\[ \Phi(P \cap \mu(\Phi)) \subseteq \Phi(P) \cap \mu(\Phi) \subseteq \Phi(P)\] 
\[ \Phi(P \cup \nu(\Phi)) \supseteq \Phi(P) \cup \nu(\Phi) \supseteq \Phi(P)\] 
which follow from the monotonicity of $\Phi$.
The derivations of these principles in $\IFP$ are straightforward. For example,
assuming the premise of $\si(\Phi,P)$, $\Phi(P \cap \mu(\Phi)) \subseteq P$, 
one defines another operator $\Psi = \lambda X\,\,\Phi(X \cap \mu(\Phi))$ so    
that $\Psi(P) \subseteq P$. Then, $\mu(\Psi) \subseteq P$ by induction.  
Hence it suffices to show $\mu(\Phi) \subseteq \mu(\Psi)$.
The converse inclusion, $\mu(\Psi) \subseteq \mu(\Phi)$, follows by induction 
since $\Psi(\mu(\Phi)) \equiv \Phi(\mu(\Phi))\subseteq \mu(\Phi)$.
But now $\mu(\Phi) \subseteq \mu(\Psi)$ follows by induction since
$\Phi(\mu(\Psi)) \subseteq \Phi(\mu(\Psi) \cap \mu(\Phi)) = \Psi(\mu(\Psi)) \subseteq
\mu(\Psi)$. The other proofs are similar.
Despite their derivability we will adopt these strengthenings of induction 
and coinduction as genuine rules of $\IFP$ since we can realize 
them by programs that are simpler than those that would be extracted from their 
derivations (see Thm.~\ref{thm-soundness}).

When defining the syntax of $\IFP$ we deliberately left open the exact
structure of terms. This will give us greater flexibility regarding different
instantiations of $\IFP$.
All we need to require of terms, in order to guarantee that the theorems of $\IFP$ 
are true with respect to the usual Tarskian semantics, is that their semantics 
satisfies a `substitution lemma', that is 
\[\valu{t[r/x]}{\eta} = \valu{t}{\eta[x \mapsto \valu{r}{\eta}]}.\] 
The rest follows from the Tarskian soundness of the rules of intuitionistic 
predicate logic and the existence of least and greatest fixed points of 
monotone predicate transformers as explained above.

Note that, since $\False$ is defined as $\mu(\lambda X\,X)()$, the 
schema 
\emph{ex-falso-quodlibet}, $\False \to A$, 
follows from $A\to A$ by induction. 

For the proof of the Soundness Theorem and the description of the
program extraction procedure (Sect.~\ref{sec-soundness}) it will be convenient
to denote $\IFP$ derivations by derivation terms and describe the
proof calculus through an inductive definition of a set of
derivation judgements $\Gamma \vdash d : A$ 
where $\Gamma$ is a context of assumptions, 
$d$ is a derivation term in that context, 
and $A$ is the formula proved by the derivation.
Derivations are defined relative to a given set $\ax$ of 
\emph{axioms} consisting of pairs $(o,A)$ where 
$A$ is any closed formula and
$o$ is the name of the axiom,
though the Soundness Theorem holds only under nc axioms.

\begin{center}
$\Gamma, u:A \vdash u : A$
\hspace{3em} 
$\Gamma \vdash o:A$\quad ($(o,A)\in \ax$)
\\[0.5em]
$\Gamma \vdash \refl{t} : t=t$
\hspace{3em} 
\AxiomC{$\Gamma\vdash d:A[s/x]$}
\AxiomC{$\Gamma\vdash e:s=t$}
             \BinaryInfC{$\Gamma \vdash \congr{d}{e}{\lambda x\,A} : A[t/x]$}
            \DisplayProof 
\\[0.5em]
\AxiomC{$\Gamma\vdash d:A$}
\AxiomC{$\Gamma\vdash e:B$}
             \BinaryInfC{$\Gamma \vdash \andi{d}{e} : A \land B$}
            \DisplayProof 
\hspace{1em} 
\AxiomC{$\Gamma \vdash d: A \land B$}
       \UnaryInfC{$\Gamma \vdash \andel{d} : A$}
            \DisplayProof 
\hspace{1em} 
\AxiomC{$\Gamma \vdash d: A \land B$}
       \UnaryInfC{$\Gamma \vdash \ander{d} : B$}
            \DisplayProof 
\\[0.5em]
\AxiomC{$\Gamma\vdash d:A$}
             \UnaryInfC{$\Gamma \vdash \oril{d}{B}:A\lor B$}
            \DisplayProof 
\hspace{3em} 
\AxiomC{$\Gamma\vdash d:B$}
             \UnaryInfC{$\Gamma \vdash \orir{d}{A}:A\lor B$}
            \DisplayProof \\[0.5em]
\AxiomC{$\Gamma\vdash d:A \lor B$}
\AxiomC{$\Gamma\vdash e:A \to C$}
\AxiomC{$\Gamma\vdash f:B\to C$}
             \TrinaryInfC{$\Gamma \vdash \ore{d}{e}{f}:C$}
            \DisplayProof \ \ \ \ 
\\[0.5em]
\AxiomC{$\Gamma, u:A\vdash d:B$}
             \UnaryInfC{$\Gamma \vdash \impi{u:A}{d} : A\to B$}
            \DisplayProof 
\hspace{3em} 
\AxiomC{$\Gamma\vdash d:A \to B$}  
\AxiomC{$\Gamma\vdash e:A$}
             \BinaryInfC{$\Gamma \vdash \impe{d}{e}:B$}
            \DisplayProof \ \ \ \ 
\\[0.5em]
\AxiomC{$\Gamma \vdash d : A$}
             \UnaryInfC{$\Gamma \vdash \alli{x}{d}:\forall x\,A$}
            \DisplayProof
($x$ not free in $\Gamma$) 
\AxiomC{$\Gamma \vdash d : \forall x\,A$}
             \UnaryInfC{$\Gamma \vdash \alle{d}{t}:A[t/x]$}
            \DisplayProof 
\\[0.5em]
\AxiomC{$\Gamma \vdash d : A[t/x]$}
             \UnaryInfC{$\Gamma \vdash \exi{t}{d}{\lambda x\,A}:\exists x\,A$}
            \DisplayProof
\AxiomC{$\Gamma \vdash d : \exists x\,A$}
\AxiomC{$\Gamma \vdash e : \forall x\,(A \to B)$}
             \BinaryInfC{$\Gamma \vdash \exe{d}{e}:B$}
            \DisplayProof 
($x$ not free in $B$) 
\\[0.5em]
$\Gamma \vdash \clos{\Phi} : \Phi(\mu(\Phi))\subseteq \mu(\Phi)$
\hspace{3em} 
\AxiomC{$\Gamma \vdash d : \Phi(P)\subseteq P$}
             \UnaryInfC{$\Gamma \vdash \induct{\Phi,P}{d} : \mu(\Phi)\subseteq P$}
            \DisplayProof 
\\[0.5em]
$\Gamma\vdash\cocl{\Phi}:\nu(\Phi) \subseteq \Phi(\nu(\Phi))$
\hspace{3em} 
\AxiomC{$\Gamma \vdash d : P \subseteq \Phi(P)$}
             \UnaryInfC{$\Gamma \vdash \coinduct{\Phi,P}{d} : P \subseteq \nu(\Phi)$}
            \DisplayProof 
\\[0.5em]
\AxiomC{$\Gamma \vdash d : \Phi(P)\cap \mu(\Phi)\subseteq P$}
             \UnaryInfC{$\Gamma \vdash \hsinduct{\Phi,P}{d} : \mu(\Phi)\subseteq P$}
            \DisplayProof 
\hspace{3em} 
\AxiomC{$\Gamma \vdash d : \Phi(P\cap\mu(\Phi))\subseteq P$}
             \UnaryInfC{$\Gamma \vdash \sinduct{\Phi,P}{d} : \mu(\Phi)\subseteq P$}
            \DisplayProof 
\\[0.5em]
\AxiomC{$\Gamma \vdash d : P \subseteq \Phi(P)\cup\nu(\Phi)$}
             \UnaryInfC{$\Gamma \vdash \hscoinduct{\Phi,P}{d} : P \subseteq \nu(\Phi)$}
            \DisplayProof
\hspace{3em}  
\AxiomC{$\Gamma \vdash d : P \subseteq \Phi(P\cup\nu(\Phi))$}
             \UnaryInfC{$\Gamma \vdash \scoinduct{\Phi,P}{d} : P \subseteq \nu(\Phi)$}
            \DisplayProof 
\end{center}
Note that symmetry and transitivity of equality can be derived
from reflexivity and the congruence rule.

\subsection{Wellfounded induction and Brouwer's Thesis}
\label{sub-wf}
The principle of \emph{wellfounded induction} is
an induction principle for elements in the accessible or wellfounded part 
of a binary relation $\less$
(definable in the language of the given instance of $\IFP$).
We show that it is an instance of strictly positive induction:
The 
\emph{accessible part} of $\less$ is defined 
inductively by
\[\acc_{\less}(x) \eqmu \forall y \less x\, \acc_{\less}(y)\] 
that is, $\acc_{\less} = \mu(\Phi)$ where 
$\Phi \eqdef \lambda X\,\lambda x\,\forall y \less x\,X(y)$.
A predicate $P$ is called
\emph{progressive} if $\Phi(P)\subseteq P$, that is, $\prog_\less(P)$ holds where
\[\prog_{\less}(P) \eqdef 
    \forall x\,(\forall y\less x P(y) \to P(x))\,.  \]
Therefore, the principle of wellfounded induction, 
which states that a progressive predicate holds on 
the accessible part of $\less$, 
is a direct instance of the rule of strictly positive induction:
\[
\infer[\hbox{$\wfi_\less(P)$}]{
  \acc_\less\subseteq P 
}{
\prog_{\less}(P)
}
\]
In most applications $P$ is of the form
$A \Rightarrow P$.  
The progressivity of $A\Rightarrow P$ can be equivalently
written as progressivity of $P$ relativized to $A$,
\[\prog_{\less,A}(P) \eqdef 
    \forall x\in A\,(\forall y \in A\, (y\less x \to P(y)) \to P(x))\]
and the conclusion becomes
$\acc_\less\subseteq A\Rightarrow P$, equivalently,
$\acc_\less\cap A \subseteq P$.
\[
\infer[\hbox{$\wfi_{\less,A}(P)$}]{
  \acc_\less\cap A\subseteq P 
}{
\prog_{\less,A}(P)
}
\]
Dually to the accessibility predicate one can define
for a binary relation a path predicate
\[\pat_{\less}(x) \eqnu \exists y \less x\, \pat_{\less}(y)\] 
that is, $\pat_{\less} = \nu(\Phi)$ where 
$\Phi \eqdef \lambda X\,\lambda x\,\exists y \less x\,X(y)$.
Intuitively, $\pat_{\less}(x)$ states that there is an infinite descending
path $\ldots x_2 \less x_1 \less x$.

With the axiom of choice and classical logic one can show that 
$\neg\pat_{\less}(x)$ implies $\acc_{\less}(x)$ 
(the converse holds even intuitionistically),
which can be viewed as an abstract form of Brouwer's Thesis:
\[ \BT_\less \qquad \forall x\,(\neg\pat_{\less}(x) \to \acc_{\less}(x)) \]
In conjunction with wellfounded induction, $\BT_\less$
says that $\less$-induction is valid for all elements without 
infinite $\less$-descending path.

If $\less$ is defined in a 
disjunction-free way, then $\BT_\less$ is a true nc formula which
can be postulated as an nc axiom.
By $\BTnc$ we denote the schema $\BT_{\less}$ for any binary nc predicate 
$\less$.
In Sect.~\ref{subsub-ai} we will use $\BTnc$ to justify
a principle called `Archimedean Induction' 
which in turn will be needed in Sect.~\ref{sec-realnumbers}.

\emph{Remark\/}.                         
Brouwer's original thesis which he used
to justify Bar Induction is obtained from $\BT_\less$ by defining
for a `bar predicate' $P$ on finite sequences of
natural numbers the relation 
$y \less x \eqdef \neg P(x) \land \exists a\, (y = a x)$,
where $a x$ denotes the sequence $x$ prefixed with $a$
(see e.g.~\cite{Veldman06}).
Bar Induction on a predicate $Q$ is then equivalent to $\wfi_{\less}(Q)$.

\subsection{Example: Real numbers}
\label{sub-ex-reals}
We illustrate the concepts introduced so far by an instance of $\IFP$
providing an abstract specification of real numbers. 
This 
will be
the basis for the program extraction case study in 
Sect.~\ref{sec-realnumbers}. Hence we will take care to 
postulate only non-computational axioms.
\subsubsection{The language of real numbers}
\label{subsub-basic} 
The language of the real numbers is given by
\begin{itemize}
\item[(1)] \emph{Sorts}: 
One sort $\iota$ as a name for the set of real numbers.
\item[(2)] \emph{Terms}: 
First-order terms built from the constants and function symbols 
$0,1,+,-,*,/,2^{(\cdot)} \mbox{(exponentiation)},\max$.
Further function symbols may be added on demand.
We set $|x| \eqdef \max(x,-x)$.
\item[(3)] \emph{Predicate constants}: $<,\le$.
\end{itemize}

\subsubsection{The axioms of real numbers}
\label{subsub-axioms}

As axioms we may choose any disjunction-free formulas that 
are true in the real numbers.  
As such, we define $\ax_R$ that consists of
a disjunction-free 
formulation of the axioms of real-closed fields, equations 
for exponentiation, the defining axiom for $\max$
\[
\max(x,y)\le z \leftrightarrow y\le z\land x\le z,
\] 
stability of $=,\le,<$, 
as well as 
$\AP$ (Archimedean property) that will be defined in Sect.~\ref{subsub-infinite} and
Brouwer's Thesis for nc predicates ($\BTnc$) introduced in Sect.~\ref{sub-wf}.
In the remainder of Sect.~\ref{sec-ifp} and also in Sect.~\ref{sec-realnumbers}
all proofs take place in $\IFP(\ax_R)$.

\subsubsection{Natural numbers}
\label{subsub-nat}

The natural numbers, considered as a subset of the 
real numbers, can be defined inductively by
\begin{eqnarray*}
\NN(x) &\eqmu& x = 0 \lor \NN(x-1)
\end{eqnarray*}
which is shorthand for  $\NN \eqdef \mu(\Phi_\NN)$ where 
  $\Phi_\NN
  \eqdef \lambda X\,\lambda x\,\,(x = 0 \lor X(x-1))$.      
Equivalently, one could define 
$\NN(x) \eqmu x = 0 \lor \exists y\,(\NN(y) \land x = y+1)$.
The closure and induction rules for $\NN$ are literally
\[
  \infer[]
{
  \forall x\,((x=0 \lor \NN(x-1)) \to \NN(x))
  }
{
 }
\qquad
  \infer[]
{
  \forall x\in\NN\, P(x)
  }
{
  \forall x\,((x=0 \lor P(x-1)) \to P(x))
 }
\]
equivalently (using equality reasoning and axioms for real numbers),
\[
  \infer[]
{
  0 \in \NN 
  }
{
 }
\qquad
  \infer[]
{
  \forall x\in \NN\,(x+1\in\NN)
  }
{
 }
\qquad
  \infer[]
{
  \forall x\in\NN\,P(x)
  }
{
  P(0) \land \forall x\,(P(x) \to P(x+1))
 }
\]
The missing Peano axiom, $\forall x\,(\NN(x) \to x+1\neq 0)$,
follows from the formula $\forall x\,(\NN(x) \to 0 \le x)$ which
can be proven by induction.

Strong induction on natural numbers is equivalent to
\[
  \infer[]
{
  \forall x\in\NN\,P(x)
  }
{
  P(0) \land \forall x\in\NN\,(P(x) \to P(x+1))
 }
\]

The 
rational numbers $\QQ$ can be defined from the 
natural numbers as usual, 
for example 
$\QQ(q) \eqdef \exists x,y,z\in\NN\,(z \neq 0 \land q\cdot z = x-y)$.

 \begin{example}\label{ex-sum1}
   We prove that the sum of two natural numbers is a natural number, which is expressed as 
   $\forall x, y\, (\NN(x) \to \NN(y) \to \NN(x+y))$.  
   An addition program for natural numbers will be extracted from this proof in Example \ref{ex-sum2}.
   Suppose that $x$ satisfies $\NN(x)$.  We prove $\forall y (\NN(y) \to \NN(x+y))$ by induction.
   Thus, we need to prove $\forall y\ (y = 0 \lor \NN(x+(y-1)) \to \NN(x+y))$. If $y=0$, then $\NN(x+y)$ holds by
   the assumption and 
$x+0 = x$.  
If $\NN(x+(y-1))$, then $y = 0 \lor \NN((x+y)-1)$ since $x+(y-1) = (x+y)-1$.  
Therefore, $\NN(x+y)$ holds by  the closure rule.
 \end{example}

\subsubsection{Infinite numbers and the Archimedean property}
\label{subsub-infinite}
As an example of a coinductive definition we define infinite numbers
by
\begin{eqnarray*}
\infty(x) &\eqnu& x\ge 0 \land \infty(x-1)\,.
\end{eqnarray*}
Hence a real number is infinite iff by repeatedly subtracting $1$ one
always stays non-negative (and hence positive).

The Archimedean property of real numbers can be expressed by stating
that there are no infinite numbers: 
\[\AP\qquad \forall x\,\neg\infty(x)\]
Since this is a true nc formula we include it as an axiom for the real numbers.

To give simple examples of proofs by induction and coinduction we show
\begin{lemma}
\label{lem-infty}
$\forall x\,(\infty(x) \leftrightarrow \forall y\in\NN\,y \le x)$.
\end{lemma}
\begin{proof}
For the implication from left to right we show 
\[\forall y\in\NN\,\forall x\,(\infty(x) \to y \le x) \]
by induction on $y\in\NN$. $\forall x\,(\infty(x) \to 0 \le x)$
holds by the coclosure axiom for $\infty$.
In the step, the induction hypothesis is 
$\forall x\,(\infty(x) \to y \le x)$. We have to show
$\forall x\,(\infty(x) \to y+1 \le x)$. Hence assume $\infty(x)$.
By coclosure, $\infty(x-1)$. Therefore $y\le x-1$, by the induction
hypothesis. It follows $y+1\le x$.

The implication from right to left can be shown by coinduction.
Setting $P(x) \eqdef \forall y\in\NN\,y \le x$ we have to show
that $P(x)$ implies $x \ge 0$ and $P(x-1)$. Hence assume $P(x)$. 
$x\ge 0$ holds since $0\in\NN$. To show $P(x-1)$ let $y\in\NN$.
Then $y+1\in\NN$ and therefore, since $P(x)$, $y+1\le x$.
It follows $y\le x-1$. 
\end{proof}

\subsubsection{Archimedean induction}
\label{subsub-ai}
Now we study a formulation of the Archimedean property as an induction 
principle. This principle will be needed to prove the conversion of
the signed digit representation into Gray code (Thm.~\ref{thm-c-g'}).

If we set $y \less x \eqdef x \ge 0 \land y = x-1$, then clearly
$\infty(x)$ is equivalent to $\pat_{\less}(x)$. Therefore,
by the Archimedean property, $\pat_{\less}$ is empty, and
by Brouwer's Thesis, $\BT$ (Sect.~\ref{sub-wf}), it follows that
$\acc_{\less}(x)$ holds for all $x$. 
Hence, wellfounded induction on $\less$, $\wfi_{\less}(P)$, is equivalent 
to the rule
\[
\infer[\hbox{$\AI(P)$}]{
  \forall x \,  P(x)
}{
 \forall x \, ((x \ge 0 \to P(x-1)) \to P(x))
}
\]
since clearly its premise is equivalent to $\prog_{\less}(P)$.

A useful variant of $\AI$ is obtained by defining 
\[y \less_q x \eqdef |x| \leq q \land y = 2x\]
where here and in the following we assume $q>0$.
Then, as we will prove in Lemma \ref{lem-pathq},
$\acc_{\less_q}(x)$ is equivalent to $x\neq 0$.
Therefore, 
half strong induction, $\hsi(\Phi,P)$, for 
$\Phi = \lambda X\, \lambda x\ (\forall y \less_q x \, X(y))$,
yields the rule 
\[
  \infer[\AI_q(P)]{
    \forall x \ne 0\ P(x)
}{
\forall x \ne 0\ ((|x| \leq q \to P(2x)) \to P(x))
}
\]
since its premise is equivalent to $\prog_{\less_q}(P)$.
We call the principles $\AI(P)$ and $\AI_q(P)$ 
\emph{Archimedean induction}. 
Therefore, we have shown:
\begin{lemma}
\label{lem-ai}
Archimedean induction is derivable in 
$\IFP(\ax_R)$.
\end{lemma}

\begin{lemma}
\label{lem-pathq}
$\acc_{\less_q}(x)$ iff $x \neq 0$.
\end{lemma}
\begin{proof}
The `only if' part follows by induction on $\acc_{\less_q}(x)$: Since
$\acc_{\less_q}(x) \eqmu \forall y \,(|x| \leq q \land y = 2x \to \acc_{\less_q}(y))$ is equivalent to 
$\acc_{\less_q}(x) \eqmu (|x| \leq q \to \acc_{\less_q}(2x))$
it suffices to show that $(|x| \leq q \to 2x \neq 0)$ 
implies $x \neq 0$, which is immediate (using $2\cdot 0 = 0$).

The `if' part reduces, by $\BTnc$, to the implication
$x\neq 0 \to \neg \pat_{\less_q}(x)$. Therefore, we assume $x\neq 0$ and
$\pat_{\less_q}(x)$ with the aim to arrive at a contradiction.
Recall that 
$\pat_{\less_q}(x) \eqnu \exists y\, (|x| \leq q \land y = 2x \land \pat_{\less_q}(y))$, 
which is equivalent to 
$\pat_{\less_q}(x) \eqnu (|x| \leq q \land \pat_{\less_q}(2x))$.
By induction on $\NN$ we can prove 
$$\forall n \in \NN\ \forall x (\pat_{\less_q}(x) \to |x|  \leq q2^{-n})\,.$$
Therefore, if $\pat_{\less_q}(x)$, then for all $n\in\NN$,
$q/{|x|} \geq 2^n \ge n$, which, by Lemma~\ref{lem-infty} and
$\AP$, is impossible.
\end{proof}

In most applications Archimedean induction
is used with a predicate of the form $B \Rightarrow P$, 
and its premise is stated in an intuitionistically 
slightly stronger (though classically equivalent) form. 
 \[
 \infer[\hbox{$\AIB(B,P)$}]{
   \forall x \in B \, P(x)
 }{
  \forall x \in B\, 
      (P(x) \lor (x \ge 0 \land B(x-1) \land (P(x-1) \to P(x))))
 }
 \]
 \[
 \infer[\hbox{$\AIB_q(B,P)$}]{
   \forall x \in B\setminus\{0\} \, P(x)
 }{
  \forall x \in B\setminus\{0\}\, 
      (P(x) \lor (|x| \le q \land B(2x) \land (P(2x) \to P(x))))
 }
 \]
\begin{lemma}
\label{lem-aib}
$\AI$ implies $\AIB$. $\AI_q$ implies $\AIB_q$.
\end{lemma}
\begin{proof}
The premise of $\AIB(B,P)$ implies the premise of 
$\AI(B \Rightarrow P)$.
The premise of $\AIB_q(B,P)$ implies the premise of 
$\AI_q(B \Rightarrow P)$.
\end{proof}


\section{Realizability}
\label{sec-realizability}

In this section we define a realizability interpretation of $\IFP$.      
The interpretation will be formalized in a system $\RIFP$ defined 
in Sect.~\ref{sub-rifp} which is  
another instance of $\IFP$ with extra sorts and 
terms for extracted programs 
and their types 
(Sect.~\ref{sub-programs} and Sect.~\ref{sub-types})
as well as axioms describing them (Sect.~\ref{sub-rifp}).

Our programming language is an untyped language with 
a type assignment system.  It is similar to the language studied in
\cite{MacQueenPlotkinSethi86}, but simpler in that recursive types are restricted to
strictly positive ones.%

Programs will be interpreted in a Scott domain $D$ satisfying a recursive domain
equation, types will be interpreted as subdomains of $D$ (Sect.~\ref{sub-domain}).

A lazy operational semantics of this language will be studied in Sect.~\ref{sec-opsem}
and shown to be equivalent to the denotational semantics. 
Our domain-theoretic model of untyped programs originates in work by 
Scott~\cite{Scott82}. An overview and comparison of different models of
untyped $\lambda$-calculi can be found in~\cite{Plotkin93}

To define the realizability interpretation,
we first assign types to $\IFP$ expressions (Sect.~\ref{sub-tau}) and
then define the set of realizers of an expression 
as a subset of the subdomain defined by its type
(Sect.~\ref{sub-realizability}).
  We also show that typable $\RIFP$ programs can be translated into Haskell
programs (Sect.~\ref{sub-haskell}) and explain how  Haskell programs can 
be directly extracted from  IFP proofs in Section \ref{sec-soundness}.

\subsection{The domain of realizers and its subdomains}
\label{sub-domain} 
Extracted programs will be interpreted in a Scott domain $D$  
defined by the recursive domain equation 
\[
D = (\Nil +
\Left(D) + \Right(D) + \Pair(D\times D) + \Fun(D\to D))_\bot
\]
where 
$D \to D$ is the domain of continuous functions from $D$ to $D$,
$+$ denotes the disjoint sum of partial orders and $(\cdot)_\bot$
adds a new bottom element.
$\Nil,
\Left, \Right, \Pair, \Fun$
denote the injections of the various components of the sum into $D$.
$\Nil,
\Left, \Right, \Pair$ (but not $\Fun$) 
are called \emph{constructors}.
$D$ carries a natural partial order $\dle$ 
with respect to which $D$ is a 
countably based
\emph{Scott domain} 
(\emph{domain} for short),
that is, a bounded complete algebraic dcpo with least 
element $\bot$ and 
countably many compact 
elements.
The theory of Scott domains and recursive domain equations is standard and can 
be found e.g.\ in~\cite{AbramskyJung94,GierzHofmannKeimelLawsonMisloveScott03}.

Since domains are closed under suprema of increasing chains,
$D$ contains not only finite but also infinite combinations of 
the constructors.
For example, writing $a:b$ for $\Pair(a,b)$, 
an infinite sequence of domain elements $(d_i)_{i\in\NN}$
is represented in $D$ as the stream
\[d_0 : d_1 : \ldots \eqdef
\sup_{n\in\NN} \Pair(d_0,\Pair(d_1,\ldots \Pair(d_n,\bot)\ldots))\,.\]
Since Scott domains and continuous functions form a cartesian closed category, 
$D$ can be equipped with the structure of a partial
combinatory algebra (PCA, \cite{GierzHofmannKeimelLawsonMisloveScott03})
by defining a continuous application operation $a\,b$ such that 
$\Fun(f)\,b \eqdef f(b)$ and otherwise $a\,b \eqdef \bot$, 
as well as combinators 
$K$ and $S$
satisfying $K\,a\,b =b$ and $S\,a\,b\,c = a\,c\,(b\,c)$ (where application associates
to the left).
In particular $D$ has a continuous least fixed point operator 
which can be defined by Curry's $Y$-combinator 
or as the mapping $(D\to D) \ni f \mapsto \sup_n f^n(\bot) \in D$.

Besides the PCA structure we will use the algebraicity of $D$, 
that is, the fact that every element of $D$ is the directed supremum 
of compact elements. Compact elements have a strongly finite character 
which will be exploited in the proof of uniqueness of certain fixed points 
(Sect.~\ref{sub-types}) and in the proof of the 
Computational Adequacy Theorem (Thm.~\ref{thm-adequacy}).
The finiteness of compact element is captured by
their defining property ($d\in D$ is compact iff for every directed set 
$A\subseteq D$, if $d \dle \bigsqcup A$, then $d\dle a$ for some $a\in A$) 
and the existence of a function assigning to every compact element $a$
a \emph{rank}, $\rk(a)\in\NN$, satisfying 
\begin{itemize}
\item[$\rk 1$]
  If $a$ has the form $C(a_1,\ldots,a_k)$ for a constructor $C$,
  then $a_1,\ldots,a_k$ are compact and $\rk(a) > \rk(a_i)$ $(i \leq k)$.
\item[$\rk 2$] If $a$ has the form $\Fun(f)$, then 
for every $b\in D$, $f(b)$ is compact with $\rk(a) > \rk(f(b))$ and 
there exists a compact $b_0 \dle b$ such that $\rk(a) > \rk(b_0)$ and $f(b_0)=f(b)$.
Moreover, there are finitely many compacts
$b_1,\ldots b_n$ with $\rk(b_i) < \rk(a)$ such that 
$f(b) = \bigsqcup \{f(b_i) | i=1,\ldots n, b_i\dle b\}$.
\end{itemize}

In Sect.~\ref{sub-types} we will model types as \emph{subdomains} of $D$, 
that is, subsets of $D$ that are downwards closed and
closed under suprema of bounded subsets. 
We write $\subdom{X}$ if
$X$ is a subdomain of $D$ and denote by $\subdoms$ the set of all subdomains
of $D$. 
It is easy to see that a subdomain $X$ is a domain with respect 
to the partial order inherited from $D$ and the notions of supremum and 
compact element in $X$ are the same as taken with respect to $D$. 
The following is easy to see.

\begin{lemma}
\label{lem-subdom}
\begin{itemize}
\item[(a)] $\subdoms$ is a complete lattice. The meet operation is intersection.
\item[(b)] $\subdoms$ is closed under the following operations.

$(X+Y)_\bot \eqdef \{\Left(a)\mid a\in X\} \cup 
                           \{\Right(b) \mid b \in Y\} \cup \{\bot\}$,

$(X\times Y)_\bot \eqdef \{\Pair(a,b) \mid a\in X, b\in Y \} \cup \{\bot\}$,

$(\ftyp{X}{Y})_\bot \eqdef \{\Fun(f) \mid f:D\to D\hbox{ cont., }
                    \forall a\in X (f(a) \in Y) \} \cup \{\bot\}$.
\end{itemize}
\end{lemma}
By Lemma~\ref{lem-subdom}~(a), for every set $S\subseteq D$ there is a smallest 
subdomain $X$ containing $S$, called the subdomain \emph{generated} by $S$.
Hence for any subdomain $Y$, $S\subseteq Y$ iff $X\subseteq Y$.
Furthermore, any subdomain is generated by the set of its compact elements. 
\subsection{Programs}
\label{sub-programs}
In order to formally denote elements of $D$ 
we introduce terms $M,N,\ldots$ of 
a new sort $\dsort$, called \emph{programs}.
\begin{eqnarray*}
  \mathit{Programs} \owns M,N 
& :: = & a,b\ \ 
  \text{(program variables, i.e.\ variables of sort $\dsort$)}\\ 
  & | &  \Nil\ |\ \Left(M)\ |\ \Right(M)\ |\ \Pair(M,N)\ \\
& | & \case\,M\,\of\,\{Cl_1;\ldots;Cl_n\} \\ 
& | & \lambda a.\,M\\
& | & M\,N \\ 
& | & \rec\,M \\
& | & \botexp
\end{eqnarray*}
%
In the case-construct each $Cl_i$ is a \emph{clause}  of the form 
$C(\vec a) \to N$ where $C$ is a constructor and $\vec a$ is a 
tuple of different
variables whose free occurrences in $N$ are bound by the clause.
Furthermore, for different clauses $C(\vec a)\to M$ and $C'(\vec a')\to M'$,
the constructors $C$ and $C'$ must be different.
The intuitive meaning of a case-expression, say 
$\case\,M\,\of\,\{\ldots; \Left(a) \to L; \ldots\}$, is that 
if $M$ evaluates to a term matching the pattern $\Left(a)$, 
say $\Left(M')$, then the whole case-expression evaluates to
$L[M'/a]$.
The recursion construct $\rec\,M$ defines the least fixed point of $M$. 
It could be defined as $Y\,M$ where $Y$ is the well-known combinator
$\lambda f\,.\,(\lambda a\,.\,f\,(a\,a))\, (\lambda a\,.\,f\,(a\,a))$,
however, we prefer an explicit construct for general recursion
since it better matches programming practice~(Sect.~\ref{sub-haskell}) 
and it can be naturally assigned a type~(Sect.~\ref{sub-types}) while
$Y$ involves self-application which is not typable in our system
(see the remark after Lemma~\ref{lem-typ-rifp}).
The constant $\bot$ represents the `undefined' domain element $\bot$.
It could be defined as a non-terminating recursion but it is more convenient
to have it as a constant.
Overall, our goal is to have a programming language 
that enables us to naturally express the 
computational contents of
IFP expressions and proofs.

Substitution of programs, $M[N/a]$, is defined as usual in term 
languages with binders 
so that a substitution lemma holds~(Lemma~\ref{lem-subst}).
We also identify $\alpha$-equal programs, that is, programs
that are equal up to renaming of bound variables.
Composition, sum, pairing, and projections are defined as
\begin{eqnarray*} 
M\circ N &\eqdef& \lambda a.\,M (N\,a)\\
\funsum{M}{N}&\eqdef&
 \lambda c.\,\case\,c\,\of\,\{\Left(a)\to M\,a;\Right(b)\to N\,b\}\\
\funpair{M}{N}&\eqdef& \lambda c.\,\Pair(M\,c,N\,c)\\
\projl\,M &\eqdef& \,\case\,M\,\of\,\{\Pair(a,b)\to a\}\\
\projr\,M &\eqdef& \,\case\,M\,\of\,\{\Pair(a,b)\to b\}
\end{eqnarray*}
We write
$a \eqrec M$ for $a \eqdef \rec(\lambda a.\,M)$, and
$a\,b \eqrec M$ for  $a \eqrec \lambda b.\,M$.
Occasionally we will use 
generalized clauses such as $\Right(\Pair(a,b)) \to M$ as an abbreviation
for $\Right(c) \to \case\ c\ \of\ \{\Pair(a,b) \to M\}$.

Since $D$ is a combinatory algebra
every program $M$ denotes an element $\valu{M}{\eta}\in D$ 
depending continuously (w.r.t.\ the Scott topology) on the
environment $\eta$ that maps program variables to elements of $D$.
\begin{eqnarray*}
\valu{a}{\eta} &=& \eta(a)\\
\valu{C(M_1,\ldots,M_k)}{\eta} &=& C(\valu{M_1}{\eta},\ldots,
                                       \valu{M_k}{\eta})\\
\valu{\case\,M\,\of\,\{Cl_1;\ldots;Cl_n\}}{\eta} &=& 
       \valu{K}{\eta[\vec a \mapsto \vec d]} \quad
\hbox{if $\valu{M}{\eta} = C(\vec d)$}\\
&&\quad \hbox{and some $Cl_i$ is of the form $C(\vec a) \to K$}\\
\valu{\lambda a.\,M}{\eta} &=& \Fun(f)\quad  
  \hbox{where $f(d) = \valu{M}{\eta[a\mapsto d]}$}\\
\valu{M\,N}{\eta} &=& f(\valu{N}{\eta})\quad 
                        \hbox{if $\valu{M}{\eta}= \Fun(f)$}\\
\valu{\rec\,M}{\eta} &=& \hbox{the least fixed point of $f$}\\
&&\quad \quad\hbox{if $\valu{M}{\eta}=\Fun(f)$}\\
\valu{M}{\eta} &=& \bot\ \ 
\hbox{in all other cases, 
in particular $\valu{\botexp}{\eta} = \bot$}
\end{eqnarray*}
For closed terms the environment is redundant and may therefore be omitted.
The following lemma is standard.
\begin{lemma}[Substitution]
\label{lem-subst}
$\valu{M[N/a]}{\eta} = \valu{M}{\eta[a \mapsto \valu{N}{\eta}]}$.
\end{lemma}

\subsection{Types}
\label{sub-types}
We introduce simple recursive 
types which are interpreted as subdomains of 
the domain $D$ defined in Sect.~\ref{sub-domain}.
 
\emph{Types} are expressions defined by the grammar
\[ Types \ni \rho, \sigma ::=  \alpha\ (\hbox{type variables})
                         \mid \one
                         \mid \rho + \sigma
                         \mid \rho \times \sigma
                         \mid \ftyp{\rho}{\sigma}
                         \mid \tfix{\alpha}{\rho} .\]
where in $\tfix{\alpha}{\rho}$ the type $\rho$ must be strictly positive in $\alpha$.

Given an environment $\zeta$ that assigns to each type variable a 
subdomain of $D$, every type $\rho$ denotes a subdomain 
$\tval{\rho}{\zeta}$ of $D$:
\begin{eqnarray*}
\tval{\alpha}{\zeta} &=& \zeta(\alpha),\\
\tval{\one}{\zeta} &=& \{\Nil,\bot\},\\
\tval{\rho\diamond\sigma}{\zeta} &=& 
(\tval{\rho}{\zeta}\diamond\tval{\sigma}{\zeta})_{\bot}\quad 
(\diamond\in\{+,\times,\ftyp{}{}\}),\\
\tval{\tfix{\alpha}{\rho}}{\zeta} &=& 
\bigcap\{\subdom{X} \mid \tval{\rho}{\zeta[\alpha \mapsto X]} \subseteq X \}
\end{eqnarray*}
\begin{lemma}
\label{lem-typ-fix}
\[ \tval{\tfix{\alpha}{\rho}}{\zeta} = 
   \tval{\rho}{\zeta[\alpha\mapsto \tval{\tfix{\alpha}{\rho}}{\zeta}]} =
   \tval{\rho[\tfix{\alpha}{\rho}/\alpha]}{\zeta} .
\]
\end{lemma}
\begin{proof}
By strict positivity, $\tval{\rho}{\zeta[\alpha\mapsto X]}$ is monotone
in $X$. Therefore, the left equation holds by Tarski's fixed point theorem. 
The right equation is an instance of the usual substitution lemma.
\end{proof}
As an example, we consider the type of natural numbers,
\[\nat \eqdef \tfix{\alpha}{1+\alpha}.\]
 By Lemma~\ref{lem-typ-fix}, 
$\tval{\nat}{} = (\tval{\one}{} + \tval{\nat}{})_{\bot}$. 
It is easy to see that
\[\tval{\nat}{} 
    = \{\Right^n(d) \mid  n\in\NN, d \in\{\bot,\Left(\bot),\Left(\Nil)\}\} \cup
     \{\sqcup_{n\in\NN}\Right^n(\bot)\}.\]
By identifying
$\Left(\bot)$ with $\Left(\Nil)$, 
one obtains an isomorphic copy of the domain of lazy natural numbers
where $\Left(\Nil)$ represents $0$ and $\Right$ represents the successor
operation.
See Remark~1 at the end of this section for a discussion on the relation between 
these domains.

Lemma~\ref{lem-typ-fix} says that $\tval{\tfix{\alpha}{\rho}}{\zeta}$
is a fixed point of the type operator $\alpha \mapsto \rho$.
We show that it is the \emph{unique} fixed point under a 
regularity condition.
The regularity conditions excludes type operators of the form 
$\alpha \mapsto \tfix{\beta_1}{\ldots\tfix{\beta_n}{\alpha}}$
(where the $\beta_i$ are all different from $\alpha$)
which, obviously, have  every subdomain of $D$ as fixed point. 
It turns out that if fixed points
of such type operators are excluded then uniqueness of fixed points holds.
Therefore, we call a type \emph{regular} if it contains no sub-expression 
of the form $\tfix{\alpha}{\tfix{\beta_1}{\ldots\tfix{\beta_n}{\alpha}}}$.
\begin{lemma}
\label{lem-typ-reg}
\begin{itemize}
\item[(a)] Regular types are closed under substitutions.
\item[(b)] Every regular type is semantically equal to a non-fixed-point type,
that is, a type which is not of the form $\tfix{\alpha}{\rho}$. 
\end{itemize} 
\end{lemma}
\begin{proof}
(a) is easy.

(b) can be proved 
by induction on the \emph{fixed point height} of a type  
which is the unique number $n$ such that the type is of the form
$\tfix{\alpha_1}{\ldots\tfix{\alpha_n}{\rho_0}}$ and $\rho_0$ is not a 
fixed point type. Let $\rho$ be a regular type. 
If the fixed point height of $\rho$ is $0$ we are done. If the fixed point height
of $\rho$ is $n+1$, then $\rho$ is of the form $\tfix{\alpha}{\sigma}$ where 
$\sigma$ has fixed point height $n$. By Lemma~\ref{lem-typ-fix}, 
$\rho$ is semantically equal to
$\sigma[\rho/\alpha]$ which has fixed point height $n$ as well since $\rho$ is
regular. Moreover, by (a), $\sigma[\rho/\alpha]$ is regular. Hence 
the induction hypothesis can be applied.
\end{proof}

Let $X,Y$ range over $\subdoms$ and set 
$\subrank{X}{n} \eqdef \{a \in X \mid a\hbox{ compact}, \rk(a)\le n\}$.

\begin{lemma}
\label{lem-subrank}
If $\subrank{X}{n}\subseteq Y$ for all $n$, then $X\subseteq Y$.
\end{lemma}
\begin{proof}
This is clear since a domain is the completion of the subset of its compact elements,
and with increasing $n$, $\subrank{X}{n}$ exhausts all compact elements of $X$.
\end{proof}

Define $\depth{\alpha}{\rho}\in\NN\cup\{\infty\}$ by recursion on 
$\rho$ as follows.
$\depth{\alpha}{\rho} = \infty$ if $\alpha$ is not free in $\rho$. 
Otherwise (using the expected order on $\NN\cup\{\infty\}$ and setting
$1+\infty = \infty$)
\begin{eqnarray*}
\depth{\alpha}{\alpha} &=& 0\\
\depth{\alpha}{\rho_1\diamond\rho_2} &=& 1 + \min_i\,\depth{\alpha}{\rho_i}
            \quad \diamond \in \{+,\times\}\\
\depth{\alpha}{\ftyp{\rho_1}{\rho_2}} &=& \depth{\alpha}{\rho_2}\\
\depth{\alpha}{\tfix{\beta}{\rho}} &=& \depth{\alpha}{\rho}
\end{eqnarray*}
The following lemma exploits regularity in an essential way and is key to
proving uniqueness of fixed points~(Thm~\ref{thm-sub-fix}).
\begin{lemma}
\label{lem-depth-subrank}
Let $\rho$ be regular and s.p.\ in $\alpha$. 

If $\subrank{X}{n}\subseteq Y$, then 
$\subrank{\tval{\rho}{\zeta[\alpha\mapsto X]}}{(n+\depth{\alpha}{\rho})}
 \subseteq \tval{\rho}{\zeta[\alpha\mapsto Y]}$.
\end{lemma}
\begin{proof}
Suppose that $\subrank{X}{n}\subseteq Y$.
We write $\rho(X)$ for $\tval{\rho}{\zeta[\alpha\mapsto X]}$ and
show that for all compact elements 
$a \in \subrank{\rho(X)}{(n+\depth{\alpha}{\rho})}$, we have $a \in \rho(Y)$.
The proof is by induction on $\rk(a)$.
We do a case analysis on $\rho$. Thanks to Lemma~\ref{lem-typ-reg}~(b) we may
skip fixed point types.

Let $a \in \subrank{\rho(X)}{(n+\depth{\alpha}{\rho})}$. 
If $a=\bot$ then the assertion holds since all subdomains contain $\bot$.
Therefore in the following we assume $a\neq\bot$.

\emph{Case $\alpha$ is not free in $\rho$\/}. Then $\rho(X)=\rho(Y)$ and therefore
the assertion holds trivially. 

\emph{Case $\rho=\alpha$\/}. Then $\rho(X) = X$, $\rho(Y) = Y$ and 
$\depth{\alpha}{\rho} = 0$. Therefore, the assertion is again trivial.

\emph{Case $\rho=\rho_1+\rho_2$\/}.  
W.l.o.g. $a = \Left(b)$ with $b \in \rho_1(X)$. 
Since $\rk(a) \le n+\depth{\alpha}{\rho} \le n+ 1 + \depth{\alpha}{\rho_1}$
and $\rk(a) = 1 + \rk(b)$ it follows $\rk(b) \le n+ \depth{\alpha}{\rho_1}$, 
that is, $b \in\subrank{\rho_1(X)}{(n+\depth{\alpha}{\rho_1})}$. 
By induction hypothesis $b \in\rho_1(Y)$, hence $a \in \rho(Y)$.

\emph{Case $\rho=\rho_1\times\rho_2$\/}.
Then $a = \Pair(a_1,a_2)$ with $a_i \in \rho_i(X)$ ($i=1,2$). 
Since $\rk(a) \le n+ \depth{\alpha}{\rho} \le 1 + n+ \depth{\alpha}{\rho_i}$
and $\rk(a) \ge 1 + \rk(a_i)$ it follows $\rk(a_i) \le n+ \depth{\alpha}{\rho_i}$, 
that is, $a_i \in\subrank{\rho_i(X)}{(n+\depth{\alpha}{\rho_i})}$. 
By induction hypothesis $a_i \in\rho_i(Y)$, hence $a \in \rho(Y)$.

\emph{Case $\rho=\ftyp{\rho_1}{\rho_2}$\/}.
Then $a = \Fun(f)$ with $f\in D\to D$ such that $f[\rho_1(X)] \subseteq \rho_2(X)$
and $\rk(f(a_1)) < \rk(a)$ for all $a_1\in D$. 
We have to show $a\in \rho(Y)$, that is $f[\rho_1(Y)] \subseteq \rho_2(Y)$.
Hence we assume $a_1\in\rho_1(Y)$ and show $f(a_1)\in\rho_2(Y)$.
Since $\rho$ is s.p. in $\alpha$, $\alpha$ is not free in $\rho_1$.
Therefore $\rho_1(X) = \rho_1(Y)$ and we have $a_1\in\rho_1(X)$. 
Since $\rk(f(a_1)) < \rk(a) \le n+ \depth{\alpha}{\rho} = n + \depth{\alpha}{\rho_2}$
it follows $\rk(f(a_1)) \le n+ \depth{\alpha}{\rho_2}$, i.e.\ 
$f(a_1) \in \subrank{\rho_2(X)}{(n+\depth{\alpha}{\rho_2})}$. 
By induction hypothesis $f(a_1) \in\rho_2(Y)$.
\end{proof}






%
\begin{lemma}
\label{lem-sub-sub}
Let $\rho$ be regular and s.p.\ in $\alpha$ with $\depth{\alpha}{\rho}>0$. 
Assume $X \subseteq \tval{\rho}{\zeta[\alpha\mapsto Z]}$ 
and  $\tval{\rho}{\zeta[\alpha\mapsto Y]} \subseteq Y$,
where $Z$ is the least subdomain containing $X \cup Y$. 
Then $X \subseteq Y$.

\end{lemma}
\begin{proof}

By Lemma~\ref{lem-subrank} it suffices to show that
$\subrank{X}{n} \subseteq Y$ for all $n\in\NN$. We induct on $n$. 

$n=0$: $\subrank{X}{0} = \{\bot\} \subseteq Y$.

$n+1$: By induction hypothesis, $\subrank{X}{n} \subseteq Y$. 
Hence, $\subrank{Z}{n} \subseteq Y$. 
Since $\depth{\alpha}{\rho}>0$ it follows with 
Lemma~\ref{lem-depth-subrank} that
$\subrank{\tval{\rho}{\zeta[\alpha\mapsto Z]}}{(n+1)}
 \subseteq \tval{\rho}{\zeta[\alpha\mapsto Y]}$. Therefore,
\[ 
\subrank{X}{(n+1)} 
\subseteq
\subrank{\tval{\rho}{\zeta[\alpha\mapsto Z]}}{(n+1)}
\subseteq 
\tval{\rho}{\zeta[\alpha\mapsto Y]}
\subseteq
Y
\]
\end{proof}

\begin{theorem}[Uniqueness of fixed points]
\label{thm-sub-fix}
Let $\tfix{\alpha}{\rho}$ be regular.
\begin{itemize}
\item[(a)] If $X \subseteq \tval{\rho}{\zeta[\alpha\mapsto Y]}$
for all $Y \supseteq X \cup \tval{\tfix{\alpha}{\rho}}{\zeta}$,  
then $X \subseteq \tval{\tfix{\alpha}{\rho}}{\zeta}$.
\item[(b)] If $X \supseteq \tval{\rho}{\zeta[\alpha\mapsto Y]}$ 
for all $Y \subseteq X \cap \tval{\tfix{\alpha}{\rho}}{\zeta}$,  
then $X \supseteq \tval{\tfix{\alpha}{\rho}}{\zeta}$.
\end{itemize}
In particular $X = \tval{\rho}{\zeta[\alpha\mapsto X]}$
iff $X = \tval{\tfix{\alpha}{\rho}}{\zeta}$,
\end{theorem}
\begin{proof}
For the first implication use Lemma~\ref{lem-sub-sub} with 
$Y \eqdef \tval{\tfix{\alpha}{\rho}}{\zeta}$, 
noting that $\tval{\rho}{\zeta[\alpha\mapsto Y]} = Y$ 
by Lemma~\ref{lem-typ-fix},
and $\depth{\alpha}{\rho}>0$ since $\tfix{\alpha}{\rho}$ is regular. 
For the second implication the argument is similar. 
\end{proof}

\emph{Remark\/}. In \cite{MacQueenPlotkinSethi86} a similar result is obtained
for a larger type system that includes universal and existential type 
quantification as well as union and intersection types,
and permitting fixed point types without positivity condition.
Types are interpreted 
as ideals, which are similar to subdomains but are only closed under directed 
suprema. Subdomains are called strong ideals in~\cite{MacQueenPlotkinSethi86}.
The existence of fixed points is proven using the Banach Fixed Point Theorem
w.r.t.\ a metric $d$ such that for $X\neq Y$ as 
$d(X,Y) \eqdef \min\{2^{-n}\mid \subrank{X}{n}\neq\subrank{Y}{n}\}$. 
We added strict positivity since this provides stronger information about
extracted programs (see e.g.\ Lemma~\ref{lem-typ-E} and Thm.~\ref{thm-pe}, 
and the remark after Lemma~\ref{lem-typ-rifp}) 
and the definition of fixed points is more direct.

\subsection{The formal system $\RIFP$}
\label{sub-rifp}
We introduce an extension $\RIFP$ of $\IFP$ suitable for a formal definition
of realizability and a formal proof of its soundness.
In addition to the sorts of $\IFP$, $\RIFP$ contains the sorts $\dsort$ 
denoting the domain $D$, and $\subd$ denoting the set of subdomains of $D$.
Programs are terms of sort $\dsort$, types are terms of sort $\subd$.
We also add a new relation symbol $:$ of arity $(\dsort,\subd)$ where 
$a:\alpha$ means that $a$ is an element of the subdomain $\alpha$. 
We write $\forall a:\rho\,A$ for $\forall a\,(a:\rho \to A)$ and
$\exists a:\rho\,A$ for $\exists a\,(a:\rho \land A)$.
We identify a type $\rho$ with the predicate $\lambda a.\,a:\rho$, 
so that $\rho(a)$ stands for $a:\rho$ and, for example,  
$\rho\subseteq \sigma$ means
$\forall a\,(a : \rho \to a : \sigma)$.

In addition to the axioms and rules of $\IFP$, which are extended to
the language of $\RIFP$ in the obvious way 
(and which include stability of equations),
$\RIFP$ contains 
(universally quantified) axioms that reflect the denotational semantics of programs and types 
as well as those that express injectivity, range disjointness 
and compactness of constructors.
Since we will not apply a realizability interpretation to $\RIFP$ we do not 
need to restrict axioms to nc formulas.
We use the abbreviation
$\isfun(a) \eqdef \exists b\,(a = \lambda c\,.\,(b\,c))$.
\paragraph{Axioms for programs}
\label{par-ax-prog}

\begin{align*}
 \hbox{(i) }&    \case\,C_i(\vec b)\,
   \of\{C_1(\vec a_1) \to M_1; \ldots; C_n(\vec a_n)\to M_n\}     
       =  M_i[\vec b/\vec a_i]
  \\
\hbox{(ii) }&  
    \bigwedge_{i} \forall \vec b\,a \neq C_i(\vec b) \too 
    \case\,a\,\of\{C_1(\vec a_1) \to M_1;\ldots;C_n(\vec a_n) \to M_n\} 
    =  \botexp 
  \\
 \hbox{(iii) }&(\lambda b.\,M)\,a          =  M[a/b] \\  
 \hbox{(iv) }&\neg \isfun(a) \too a\,b =\botexp\\
\hbox{(v) }&\isfun(a)\land \isfun(b) \land 
\forall c\,(a\,c =b\,c) \too a = b\\
{\hbox{(vi) }}&{
    \bigexor_{C\ \hbox{constructor}} \exone\,\vec b\,(a=C(\vec b))\ \exor\ 
    \isfun(a) \ \exor\ a = \botexp}\\ 
\hbox{(vii) }&\rec\,a = a\,(\rec\,a)\\
\hbox{(viii) }& 
P(\botexp) \land \forall b\,(P(b) \too P(a\,b)) \to P(\rec\,a)
\quad\hbox{($P$ admissible)} 
\end{align*}
where an $\RIFP$ predicate of arity $(\dsort)$ 
is called \emph{admissible} 
if it contains neither free predicate variables nor existential quantifiers
nor inductive definitions.
\paragraph{Axioms for types}
\label{par-ax-typ}
\begin{align*}
%
\hbox{(ix) }&    \bot : \alpha\\  
\hbox{(x) }&  
  \rho[\tfix{\alpha}{\rho}/\alpha] \equiv \tfix{\alpha}{\rho}\\
\hbox{(xi) }& 
\forall \gamma\,(\beta \cup (\tfix{\alpha}{\rho}) \subseteq \gamma \to 
          \beta \subseteq \rho[\gamma/\alpha]) 
   \too \beta \subseteq \tfix{\alpha}{\rho}
   \quad\hbox{($\tfix{\alpha}{\rho}$ regular)}\\ 
\hbox{(xii) }&   
   c : \one \toot (c = \Nil \lor c = \botexp )\\
\hbox{(xiii) }&  
   c : \alpha\times\beta \toot (
           \exists a:\alpha,b:\beta\,(c=\Pair(a,b)) \lor c = \botexp )\\
\hbox{(xiv) }&   
  c : \alpha + \beta \toot (
 \exists a:\alpha\,(c = \Left(a)) \lor \exists b:\beta\,(c = \Right(b))
  \lor c =\botexp )\\
\hbox{(xv) }&    c : \ftyp{\alpha}{\beta} \toot (
  {(\isfun(c)} \land \forall a:\alpha\,(c\,a:\beta)) \lor c =\botexp )\\
\hbox{(xvi) }& 
        \exists\alpha\,\forall\beta\,
           (P\subseteq\beta \leftrightarrow \alpha\subseteq\beta)
        \quad\hbox{($P$ an $\RIFP$ predicate of arity $(\subd)$)}
\end{align*}
Clearly, axioms (i-vii) and (xii-xv) are correct in $D$ respectively 
in $\subdoms$.
Axiom (viii) is a restricted form of Scott-induction,
a.k.a.\ fixed point induction.
It is a way of expressing that $\rec\,a$ is the \emph{least} 
fixed point of $a$, that is, the supremum of the chain 
$\bot \dle a\,\bot \dle a\,(a\,\bot) \dle \ldots$.  
Scott-induction holds more generally for predicates that are closed 
under suprema of chains (such predicates are called 
\emph{inclusive}~in~\cite{Winskel93}).
It is easy to see that admissible predicates have this property.
As an example of Scott-induction, we show 
that every type is closed under 
least fixed points of endofunctions, that is,
\[a : \ftyp{\alpha}{\alpha} \too \rec\,a : \alpha.\]
Indeed, assuming $a : \ftyp{\alpha}{\alpha}$, the admissible predicate
$P \eqdef (\lambda b\,.\,b : \alpha)$ satisfies the premises of (viii)
since $\bot:\alpha$ by axiom (ix) (which is valid since all subdomains of $D$
contain $\bot$).
Obviously, Scott-induction is also valid for admissible predicates of 
more than one argument, e.g.\ 
\[P(\botexp,\botexp) \land \forall b_1,b_2\,(P(b_1,b_2) 
\too P(a_1\,b_1,a_2\,b_2)) \to P(\rec\,a_1,\rec\,a_2)\]
and Axiom~(viii) should be understood in this more general form.
Axioms (x) and (xi) hold by Lemma~\ref{lem-typ-fix} and Theorem~\ref{thm-sub-fix}.
Axiom (xvi) expresses the existence of the subdomain generated by $P$
and can be viewed as a form of comprehension.

We set $\RIFP(\ax) \eqdef \IFP(\ax\cup\ax')$ where
$\ax'$ consist of the axioms (i-xvi) for programs and types above.
We write $\RIFP$ for $\RIFP(\ax)$ if the set of axioms is not important.

The following lemma will be used later to simplify 
extracted programs.
\begin{lemma}
\label{lem-bot}
$\RIFP(\emptyset)$ proves:
If $f$ is strict, that is, $f\,\botexp = \botexp$, then
\begin{align*}
f(\case\,M\, \of\, \{C_1(\vec a_1) \to L_1;\ldots; C_n(\vec a_n) \to L_n\})\\
&\hspace*{-6cm} = 
 \case\,M\, \of\, \{C_1(\vec a_1)\to f\,L_1;\ldots; C_n(\vec a_n)\to f\,L_n\}\,.
\end{align*}
\end{lemma}
\begin{proof}
Let $f$ be strict. We have to prove the equation $f\,K = K'$  where $K \eqdef \case\,M\, \of\, \{C_1(\vec a_1) \to L_1;\ldots; C_n(\vec a_n) \to L_n\}$ and $K' \eqdef \case\,M\, \of\, \{C_1(\vec a_1) \to f\,L_1;\ldots; C_n(\vec a_n) \to f\,L_n\}$. 
Since we have to prove an equation and 
we assume equations to be $\neg\neg$-stable, we may use classical logic.
If $M=C_i(\vec b)$ for some $i$ and $\vec b$, then $K= L_i[\vec b/\vec a_i]$ 
and $K'=f\,L_i[\vec b/\vec a_i]$, by axiom~(i), and the equation holds. 
Otherwise, $K=K'=\botexp$ by axiom (ii), and the equation holds since $f$ 
is strict.
\end{proof}
\begin{lemma}
\label{lem-typ-rifp}
The following typing rules are derivable in $\RIFP(\emptyset)$ (where
$\Gamma$ is a typing context, that is, a list of assumptions
$a_1:\rho_1,\ldots a_n:\rho_n$).
\begin{center}
$\Gamma, a:\rho \vdash a:\rho$
\hspace{3em} 
$\Gamma \vdash \Nil:\one$
\hspace{3em}
$\Gamma \vdash \bot:\rho$
\hspace{3em}
\\[0.5em]
\AxiomC{$\Gamma\vdash M:\rho$}
             \UnaryInfC{$\Gamma \vdash \Left(M) : \rho + \sigma$}
            \DisplayProof 
\hspace{3em} 
\AxiomC{$\Gamma\vdash M:\sigma$}
             \UnaryInfC{$\Gamma \vdash \Right(M) : \rho + \sigma$}
            \DisplayProof \ \ \ \ 
\\[0.5em]
\AxiomC{$\Gamma\vdash M:\rho$}
\AxiomC{$\Gamma\vdash N:\sigma$}
             \BinaryInfC{$\Gamma \vdash \Pair(M,N) : \rho\times\sigma$}
            \DisplayProof 
\AxiomC{$\Gamma \vdash M:\rho\times\sigma$}
\AxiomC{$\Gamma, a:\rho, b:\sigma \vdash N:\tau$}
       \BinaryInfC{$\Gamma \vdash \case\,M\,\of\,\{\Pair(a,b) \to N\} : \tau$}
            \DisplayProof 
\\[0.5em]
\AxiomC{$\Gamma \vdash M:\rho + \sigma$}
\AxiomC{$\Gamma, a:\rho \vdash L:\tau$}
\AxiomC{$\Gamma, b:\sigma \vdash R:\tau$}
\TrinaryInfC{$\Gamma\vdash\case\,M\,\of\,\{\Left(a)\to L\,;\,\Right(b)\to R\}:\tau$}
            \DisplayProof \ \ \ \ 
\\[0.5em]
\AxiomC{$\Gamma, a:\rho\vdash M:\sigma$}
             \UnaryInfC{$\Gamma \vdash \lambda a.\,M : \ftyp{\rho}{\sigma}$}
            \DisplayProof 
\hspace{3em} 
\AxiomC{$\Gamma\vdash M:\ftyp{\rho}{\sigma}$}
\AxiomC{$\Gamma\vdash N:\rho$}
             \BinaryInfC{$\Gamma \vdash M\,N : \sigma$}
            \DisplayProof \ \ \ \ 
\\[0.5em]
\AxiomC{$\Gamma \vdash M : \rho[\tfix{\alpha}{\rho}/\alpha]$}
\RightLabel{{{\bf ROLL}}}
             \UnaryInfC{$\Gamma \vdash M : \tfix{\alpha}{\rho}$}
            \DisplayProof 
\hspace{3em} 
\AxiomC{$\Gamma \vdash M : \tfix{\alpha}{\rho}$}
\RightLabel{{{\bf UNROLL}}}
             \UnaryInfC{$\Gamma \vdash M : \rho[\tfix{\alpha}{\rho}/\alpha]$}
            \DisplayProof 
\\[0.5em]
\AxiomC{$\Gamma, a:\rho\vdash M\,a:\rho$}
             \UnaryInfC{$\Gamma \vdash \rec\,M : \rho$}
            \DisplayProof 
{($a$ not free in $M$)}

\hspace{3em} 
\end{center}
\end{lemma}
\begin{proof}
Immediate from the axioms for types.
\end{proof}
\emph{Remark\/}.
 Only terms typable with these rules will be extracted in Sect.~{\ref{sub-pe}}.
Note that the $Y$-combinator~(Sect.~\ref{sub-programs}) is not typable by these rules
since its type must be of the form $\ftyp{(\ftyp{\rho}{\rho})}{\rho}$, 
and in order to type the self application $(a\,a)$ occurring in $Y$ 
one needs a type $\sigma$
satisfying $\sigma \equiv \ftyp{\sigma}{\rho}$, that is, a fixed point
of a non-positive type operator. 

\subsection{Translation to Haskell }
\label{sub-haskell}

We sketch how to translate
typable $\RIFP$ programs into Haskell.
First we define a Haskell type $\Haskell{\rho}$ for each type $\rho$ and 
a sequence of Haskell {algebraic data type} declarations.
We begin with the declaration
$$
\mathsf {data\  One = Nil},
$$
and then define
\begin{itemize}
\item[(i)] $\Haskell{\one} = \mathsf{One}$
\item[(ii)] $\Haskell{\alpha} = \alpha$
\item[(iii)] $\Haskell{\rho + \sigma} = {\mathsf{Either}}\ \Haskell{\rho}\ \Haskell{\sigma} $
\item[(iv)]  $\Haskell{\rho \times \sigma}  = (\Haskell{\rho}, \Haskell{\sigma}) $
\item[(v)]  $\Haskell{\ftyp{\rho}{\sigma}}  = \Haskell{\rho} \to \Haskell{\sigma} $
\item[(vi)]  $\Haskell{\tfix{\alpha}{\rho}}  = \cona\,\vec{\beta}$
  \end{itemize}
In case (v), $\to$ is Haskell's function type constructor,
in case (vi), $\cona$ is a new name generated from $\alpha$ and $\rho$,
and $\vec{\beta}$ is a list of the free type variables in 
$\tfix{\alpha}{\rho}$.
The list of Haskell data type declarations is extended by 
the following
{recursive and possibly polymorphic} data type $\cona$
with one constructor which we again call $\cona$.
$$
\mathsf{data}\ \cona\,\vec{\beta} = 
  \cona\, \Haskell{\rho}[\cona\,\vec{\beta}/\alpha]
$$
To accommodate the typing rules {\bf ROLL} and {\bf UNROLL} we need
Haskell programs 
\begin{align*}
\mathsf{roll}_{\cona} :: 
          \Haskell{\rho}[\cona\,\vec{\beta}/\alpha] 
\to
            \cona\,\vec{\beta}  
&\qquad
\mathsf{roll}_{\cona} \,\mathsf{x} = \cona\,\mathsf{x}\\
\mathsf{unroll}_{\cona} :: 
               \cona\,\vec{\beta}  
\to 
\Haskell{\rho}[\cona\,\vec{\beta}/\alpha]
&\qquad
\mathsf{unroll}_{\cona} (\cona\,\mathsf{x}) = \mathsf{x}
\end{align*}
{and for recursive programs a 
fixed point combinator
\begin{align*}
\mathsf{rec} :: 
(\alpha\to\alpha)\to\alpha
&\qquad
\mathsf{rec}\, \mathsf{f} = \mathsf{f}\, (\mathsf{rec}\, \mathsf{f})
\end{align*}
Now suppose that $d$ is a type derivation of $M:\rho$ built from the 
typing rules in Lemma \ref{lem-typ-rifp}.
We define a Haskell program $\Haskell{d}$ of type $\Haskell{\rho}$ as follows.
By considering $\Pair(M,N)$ as the Haskell term $(M, N)$,
our program is an untyped Haskell program.
$\Haskell{d}$ is obtained by inserting appropriate 
$\mathsf{roll}\underline{\ }$ and $\mathsf{unroll}\underline{\ }$ to $M$
following the type derivation $d$.
We do not modify $M$ for rules other than {\bf ROLL} and {\bf UNROLL}.
If $d$ ends with {\bf ROLL} with $\rho =  \tfix{\alpha}{\rho'}$,  we define
$\Haskell{d} = \mathsf{roll}_{\conad} \Haskell{d'}$.
If $d$ ends with {\bf UNROLL} and
$\rho =  \rho[\tfix{\alpha}{\rho'}/\alpha]$,  we define
$\Haskell{d} = \mathsf{unroll}_{\conad} \Haskell{d'}$.
Here, $d'$ is the derivation of the premise of {\bf ROLL} and
{\bf UNROLL}.
With the Haskell program $\Haskell{d}$ obtained in this way,
we have a sound derivation of the typing $\Haskell{d}:: \Haskell{\rho}$
in Haskell.

One can optimize this translation in several ways.
For example,  one can treat a type of the form $\tfix{\alpha}{\rho_1 + \ldots + \rho_k}$
so that it is {translated} to a data type with $k$ constructors.
One can also use Haskell's list type for 
$\tfix{\alpha}({\tau \times \alpha + \one})$ 
(i.e., finite/infinite list type)
and $\tfix{\alpha}({\tau \times \alpha})$ (i.e., infinite list type).

\subsection{Types of $\IFP$ expressions}
\label{sub-tau}

We inductively assign to every $\IFP$-expression (i.e., formula or predicate) $E$
a type $\tau(E)$.
The idea is that $\tau(A)$, for a formula $A$,
is the type of potential realizers of $A$.
We call an expression \emph{Harrop} if it contains neither 
disjunctions nor free predicate variables at strictly positive positions.
This deviates from the usual definition of the Harrop 
property~\cite{TroelstraSchwichtenberg96} since 
existential quantifiers at strictly positive positions are permitted. 
The reason for this is
that quantifiers are interpreted uniformly, that is, not witnessed by realizers.
Like nc formulas, Harrop formulas have no computational content, however, they
differ from nc formulas in that they need not coincide with their own 
realizability interpretation
(see Remark 3 at the end of this section).

We define $\tau(E)$ so that the type $\one$ is assigned to 
an expression iff it is Harrop.
In the following definition,  we say that
a predicate $P$ is \emph{$X$-Harrop} if $\lambda X\, P$ is Harrop, that is, 
if $P$ is strictly positive in $X$ and 
$P[\pcv{X}/X]$ is Harrop where $\pcv{X}$ is a predicate constant associated with $X$.
\begin{align*}
\tau(P(\vec t)) &= \tau(P)\\
\tau(A \lor B) &= \tau(A) + \tau(B)\\
\tau(A \land B) &= \tau(A) \times \tau(B) &\hbox{($A,B$ non-Harrop)}\\
                &= \tau(A)  &\hbox{($B$ Harrop, $A$ non-Harrop)}\\             
                &= \tau(B)  &\hbox{($A$ Harrop, $B$ non-Harrop)}\\             
               &= \one  &\hbox{($A,B$ Harrop)}\\             
\tau(A \to B) &= \ftyp{\tau(A)}{\tau(B)}  &\hbox{($A,B$ non-Harrop)}\\
              &= \tau(B)  &\hbox{(otherwise)}\\
\tau(\diamond x\,A) &= 
  \tau(A) &\hbox{($\diamond \in\{\forall,\exists\}$)}\\[.5em]
\tau(X) &= \alpha_X & \hspace{-1cm}
         \hbox{($X$ a predicate variable, $\alpha_X$ a fresh type variable)}\\
\tau(P) &= \one &\hbox{($P$ a predicate constant)}\\
\tau(\lambda \vec x\,A) &= \tau(A)\\
\tau(\diamond (\lambda X\,P)) &= \tfix{\alpha_X}{\tau(P)}
                  &\hbox{($\diamond \in\{\mu,\nu\}$, $P$ not $X$-Harrop)}\\
         &= \one &\hbox{($\diamond \in\{\mu,\nu\}$, $P$ $X$-Harrop)}
\end{align*}

\emph{Remark\/}.                         %
Though the semantics 
$\tval{\one}{\zeta}$ of the type $\one$ is $\{\Nil,\bot\}$,
we will stipulate in Section \ref{sub-realizability} 
that only $\Nil$ is a possible realizer of a Harrop expression. 
We will also define simplified realizers for products and 
implications if some of their components are Harrop
and therefore have corresponding simplified definitions of $\tau(A)$ 
for these cases.  
Note that 
we define the type of a (co)inductively defined Harrop
  predicate $\diamond (\lambda X\,P)$ to be $\one$.
Without this simplified type assignment 
a non-regular type may be assigned to a predicate, for example, 
$\tau(\False) = \tau(\mu (\lambda X\,X))$ would become 
$\tfix{\alpha_X}{\alpha_X}$.

\begin{lemma}
\label{lem-typ-subst}
For every expression $E$ (formula or predicate) and predicate $P$,
\begin{itemize}
\item[(a)] $E$ is Harrop if and only if $\tau(E) = \one$,
\item[(b)] $\tau(E)$ is regular,
\item[(c)] if $P$ is non-Harrop, then $\tau(E[P/X]) = \tau(E)[\tau(P)/\alpha_X]$,
\item[(d)] 
If $P$ is Harrop, then 
$\tau(E[P/X]) = \tau(E[\pcv{X}/X])$.
\end{itemize}
\end{lemma}
\begin{proof}
Straightforward structural induction on $E$.
\end{proof}

\subsection{Realizers of expressions}
\label{sub-realizability}
In this section,
we define the notion that $a : \tau(A)$ 
is a \emph{realizer} of a formula $A$, 
  written $\ire{a}{A}$. This intuitively means that $a$ is a 
computational content of  the formula $A$.  
In intuitionistic logic, a proof of
$A \lor B$ provides evidence that $A$ is true or $B$ is true,
together with an indicator of which of the two cases holds.
We construct our notion of realizer by treating this as the primitive 
source of computational content.
Therefore, we defined an expression \emph{non-computational} (nc) 
if it contains neither disjunctions nor free predicate variables
(Sect. \ref{sub-ifp}).
A more general notion of an expression 
with trivial 
computational content is 
provided by the Harrop property 
which forbids the occurrence of
disjunctions and free predicate variables only at strictly positive positions
(Sect. \ref{sub-tau}).

In order to formalize realizability in $\RIFP$ we define for every 
$\IFP$ formula 
$A$ an $\RIFP$ predicate $\rea(A)$ of arity $(\dsort)$ 
that specifies the set of domain elements $a$ such that $\ire{a}{A}$ holds.
For defining $\rea(A)$, 
we simultaneously define $\reah(B)$ for Harrop formulas $B$ 
which expresses that $B$ is realizable, however with trivial 
computational content $\Nil$.
We define for every $\IFP$-expression an  $\RIFP$-expression,
more precisely we define for a
\medbreak
\noindent
\begin{tabular}{l@{\ \ }l}
& \emph{formula} $A$ a predicate $\rea(A)$ of arity $(\dsort)$;\\
&\emph{non-Harrop predicate} $P$ of arity $(\vec\iota)$ 
a predicate $\rea(P)$ of arity $(\vec\iota,\dsort)$;\\ 
&\emph{non-Harrop operator} $\Phi$ of arity $(\vec\iota)$
an operator $\rea(\Phi)$ of arity $(\vec\iota,\dsort)$;\\
&\emph{Harrop  formula} $A$ a formula $\reah(A)$;\\
&\emph{Harrop predicate} $P$ a predicate $\reah(P)$ of the same arity;\\
&\emph{Harrop operator} $\Phi$ an operator $\reah(\Phi)$ of the same arity.
\end{tabular}

\medskip

In the definition of realizability below 
we assume that to every $\IFP$ predicate variable $X$ of arity 
$(\vec \iota)$ there are assigned, in a one-to-one fashion, an 
$\RIFP$ predicate variable $\reali{X}$ of arity $(\vec \iota, \dsort)$
and a type variable $\alpha_X$. 
Furthermore, we write $\ire{a}{A}$ for $\rea(A)(a)$ and 
$\re A$ for $\exists a\,\ire{a}{A}$.
Recall that
a predicate $P$ is 
\emph{$X$-Harrop} if it is strictly positive in $X$ and 
$P[\pcv{X}/X]$ is Harrop 
where $\pcv{X}$ is a fresh predicate constant associated with $X$. In this situation we write 
$\reah_X(P)$ for $\reah(P[\pcv{X}/X])[X/\pcv{X}]$.
The idea is that $\reah_X(P)$ is the same as $\reah(P)$ but considering $X$ 
as a (non-computational) predicate constant.
\begin{align*}
\ire{a}{A}  &= (a = \Nil \land \reah(A))   &\hbox{($A$ Harrop)}\\
\ire{a}{P(\vec t)} &= \rea(P)(\vec t,a) &\hbox{($P$ non-H.)}\\
\ire{c}{(A\lor B)}   &=\ex{a}(c=\inl{a}\land\ire{a}{A})\lor
                        \ex{b}(c=\inr{b}\land\ire{b}{B}) \hspace*{-2cm}\\
\ire{c}{(A\land B)}  &=\exists a,b\,(c = \Pair(a,b) \land
      \ire{a}{A}\land \ire{b}{B})
              \ &\hbox{($A,B$ non-H.)}\\
\ire{a}{(A\land B)} &=\ire{a}{A} \land \reah(B)
   \quad&\hspace*{-1cm}\hbox{($B$ Harrop, $A$ non-H.)}\\
\ire{b}{(A\land B)} &=\reah(A) \land \ire{b}{B}
   \quad&\hspace*{-1cm}\hbox{($A$ Harrop, $B$ non-H.)}\\
\ire{c}{(A\to B)}    &= 
 c:\ftyp{\tau(A)}{\tau(B)} \land  
          \all{a}(\ire{a}{A}\to\ire{(c\,a)}{B}) 
              \ &\hspace*{-2cm}\hbox{($A, B$ non-H.)}\\
\ire{b}{(A\to B)}    &=
 b:\tau(B) \land  
     (\reah(A) \to \ire{b}{B})  
   \quad &\hspace*{-2cm}\hbox{($A$ Harrop, $B$ non-H.)}\\
\ire{a}{\allex x\,A}  &=\allex x\,(\ire{a}{A})
  \quad&\hspace*{-2cm}\hbox{($\allex\in\{\forall,\exists\}$, $A$ non-H.)}
\end{align*}
\begin{align*}
\rea(X) &= \reali{X}\\
\rea(\lambda \vec x\,A) &= \lambda (\vec x,a)\,(\ire{a}{A})  
&\hbox{($A$ non-H.)}\\
%
%
%
\rea(\munu(\lambda X\,P)) &= \munu(\lambda\reali{X}\,\rea(P)[\rho/\alpha_X])
  \quad&\hbox{($\munu\in\{\mu,\nu\}$, $\lambda X\,P$ non-H.)}\\
  & \qquad \hbox{where } \rho = \tfix{\alpha_X}{\tau(P)}\\ 
%
\rea(\lambda X\,P) &= 
\lambda \reali{X}\,\rea(P)
&\hbox{($P$ non-H.)}\\
&&\\
\reah(P(\vec t)) &= \reah(P)(\vec t)
&\hbox{($P$ Harrop)}\\
\reah(A\land B)  &=
      \reah(A)\land \reah(B) &\hbox{($A, B$ Harrop)}\\
\reah(A\to B)    &=
                     \re A 
                   \to \reah(B) &\hbox{($B$ Harrop)}\\
\reah(\allex x\,A)  &=\allex x\,\reah(A)
  \quad&\hspace*{-1cm}\hbox{($\allex\in\{\forall,\exists\}$, $A$ Harrop)}\\
&&\\
\reah(P) &= P\quad &\hbox{($P$ a predicate constant)}\\
\reah(\lambda \vec x\,A) &= \lambda \vec x\,\reah(A) &\hbox{($A$ Harrop)}\\
\reah(\munu(\Phi)) &= \munu(\reah(\Phi))
  \quad&\hbox{($\munu\in\{\mu,\nu\}$, $\Phi$ Harrop)}\\
\reah(\lambda X\,P) &=    \lambda X\,\reah_X(P)   &\hbox{($P$ $X$-Harrop)}
\end{align*}
To see that $\rea(\munu(\Phi))$ and $\reah(\munu(\Phi))$ are 
wellformed one needs to prove simultaneously that if an expression $E$ is s.p.\ in
$X$, then $\rea(E)$ is s.p.\ in $\reali{X}$, and if $E$ is $X$-Harrop, then 
$\reah_X(E)$ ($=\reah(E[\pcv{X}/X])$) is
s.p.\ in $\pcv{X}$.

In the following we use the notation
\[\adummy{\rho} \eqdef \lambda (\vec x,a)\,(a:\rho),\]
so that 
$Q\subseteq \adummy{\rho} \equiv \forall \vec x,a\,(Q(\vec x,a) \to a:\rho)$.
\begin{lemma}
\label{lem-realizability}
\item[(a)] 
If $P$ is non-Harrop, then 
$\reah(A[P/X])  =  \reah(A)[\rea(P)/\reali{X}][\tau(P)/\alpha_X]$
if $A[P/X]$ is Harrop, 
and $\rea(A[P/X])  =  \rea(A)[\rea(P)/\reali{X}][\tau(P)/\alpha_X]$
if $A[P/X]$ is non-Harrop. 
\item[(b)] If $P$ is Harrop, then
$\reah(A[P/X])  =  \reah(A[\pcv{X}/X])[\reah(P)/\pcv{X}]$
if $A$ is $X$-Harrop,
and $\rea(A[P/X])  =  \rea(A[\pcv{X}/X])[\reah(P)/\pcv{X}]$
if $A$ is not $X$-Harrop.
\item[(c)] If $A$ is Harrop, then $\reah(A) \leftrightarrow \re A$.
\item[(d)] If $E$ is an nc expression, then $\reah(E) = E$,
in particular, $\reah(\False) = \False$.
\item[(e)] $\rea(P) \subseteq \adummy{\tau(P)}$
under the assumptions 
$\reali{X} \subseteq \adummy{\alpha_X}$ 
for all free predicate variables $X$ in $P$.
\end{lemma}
\begin{proof}
The statements are proven by induction on the size of expressions 
suitably generalizing the statements to formulas or predicates. 
Parts (a-d) are easy. We only show details of (a) for the case 
of an inductive predicate $\mu(\lambda Y\,Q)$.
If $\mu(\lambda Y\,Q)[P/X]$ is non-Harrop, then
we set $\rho := \tfix{\alpha_Y}{\tau(Q)}$ and 
$\sigma := \tfix{\alpha_Y}{\tau(Q[P/X])}=\rho[\tau(P)/\alpha_X]$, and we have
%
\begin{eqnarray*}
&&\rea(\mu(\lambda Y\,Q)[P/X])\\  
&=&\rea(\mu(\lambda Y\,Q[P/X]))\\  
&=& 
  \mu(\lambda \reali{Y}\,\rea(Q[P/X])
                 [\sigma/\alpha_Y])\\
&\stackrel{\hbox{i.h.}}{=}& 
  \mu(\lambda \reali{Y}\,\rea(Q)
                 [\rea(P)/\reali{X}]
                 [\tau(P)/\alpha_X]
                 [\sigma/\alpha_Y])\\
&=&
  \mu(\lambda \reali{Y}\,\rea(Q)
                 [\rea(P)/\reali{X}]
                 [\tau(P)/\alpha_X]
                 [\rho/\alpha_Y]
                 [\tau(P)/\alpha_X])\\
&=& 
  \mu(\lambda \reali{Y}\,\rea(Q)
                [\rho/\alpha_Y])
                [\rea(P)/\reali{X}]
                [\tau(P)/\alpha_X]\\
&=& 
  \rea(\mu(\lambda Y\,Q))
                [\rea(P)/\reali{X}]
                [\tau(P)/\alpha_X]
\end{eqnarray*}
If $\mu(\lambda Y\,Q)[P/X]$ is Harrop, then $Q[P/X][\pcv{Y}/Y]$ and $Q[\pcv{Y}/Y]$
are Harrop. Hence,
\begin{eqnarray*}
&&\reah(\mu(\lambda Y\,Q)[P/X])\\  
&=&\reah(\mu(\lambda Y\,Q[P/X]))\\  
&=& \mu(\lambda Y\,\reah(Q[P/X][\pcv{Y}/Y])[Y/\pcv{Y}])\\
&=& \mu(\lambda Y\,\reah(Q[\pcv{Y}/Y][P/X])[Y/\pcv{Y}])\\
&\stackrel{\hbox{i.h.}}{=}& 
\mu(\lambda Y\,\reah(Q[\pcv{Y}/Y])[\rea(P)/\reali{X}][\tau(P)/\alpha(X)][Y/\pcv{Y}])\\
&=& 
\mu(\lambda Y\,\reah(Q[\pcv{Y}/Y])[Y/\pcv{Y}])[\rea(P)/\reali{X}][\tau(P)/\alpha(X)]\\
&=& 
\reah(\mu(\lambda Y\,\reah(Q))[\rea(P)/\reali{X}][\tau(P)/\alpha(X)]
\end{eqnarray*}



For (e), the only difficult case is a non-Harrop predicate $P$ of the form
$\munu(\lambda X\,Q)$ ($\munu\in\{\mu,\nu\}$). In that case 
$\tau(P) = \tfix{\alpha_X}{\tau(Q)}$
and
%
\[\rea(P) = \munu(\lambda\reali{X}\,\rea(Q)
                [\tau(P)/\alpha_X]). \]
%

Let $\exists\,\rea(P) := \{a\in D\mid \exists \vec x\,\rea(P)(\vec x,a)\}$ and let 
$\beta$ be the least subdomain containing $\exists\,\rea(P)$ and $\gamma$ the least subdomain
containing $\tau(P) \cup \exists\,\rea(P)$. 
$\beta$ and $\gamma$ exist by Axiom~(xvi).
It suffices to show 
$\beta \subseteq \tau(Q)[\gamma/\alpha_X]$ 
since then, by Axiom~(xi), 
$\beta \subseteq \tfix{\alpha_X}{\tau(Q)} = \tau(P)$
and consequently $\rea(P) \subseteq \adummy{\beta} \subseteq \adummy{\tau(P)}$.

For the proof of $\beta \subseteq \tau(Q)[\gamma/\alpha_X]$ it 
suffices to show (by the minimality of $\beta$) that 
$\exists\,\rea(P) \subseteq \tau(Q)[\gamma/\alpha_X]$,
i.e.\ $\rea(P) \subseteq \adummy{\tau(Q)[\gamma/\alpha_X]}$.

Substituting in the formula $\reali{X} \subseteq \adummy{\alpha_X}$ 
the predicate variable
$\reali{X}$ by $\rea(P)$ and the type variable $\alpha_X$ by $\gamma$,
one obtains the provable formula $\rea(P) \subseteq \adummy{\gamma}$. 
Hence, by the induction hypothesis,
$\rea(Q)[\gamma/\alpha_X][\rea(P)/\reali{X}] \subseteq 
\adummy{\tau(Q)[\gamma/\alpha_X]}$ is provable.
Since
$\rea(P) \equiv 
  \rea(Q)[\tau(P)/\alpha_X][\rea(P)/\reali{X}] \subseteq 
  \rea(Q)[\gamma/\alpha_X][\rea(P)/\reali{X}]$,  
we are done.
\end{proof}
\emph{Remarks\/}.
1.
Since $\Nil$ is the only possible realizer of a Harrop formula, 
one could as well 
define $\one$ as $\{\bot\}$ and use $\bot$ as the realizer of a realizable
 Harrop formula.  
Then, the domain $D_\nat$ for $\nat \eqdef \tau(\NN)$ (see Sect.~\ref{sub-rnat}) would
  be isomorphic to the domain of lazy natural numbers,
and the domain 
$D_\bool$ for $\bool \eqdef \one + \one$ would be isomorphic to
the domain of truth values $\{\bf{true}, \bf{false}, \bot\}$
(see Sect.~\ref{sub-gray}). 
However, 
using $\bot$ as a realizer of Harrop formulas contradicts 
our intuitive understanding that $\bot$ means non-termination. 
  One could as well obtain these isomorphisms without modifying the realizer of
  a Harrop formula 
by adding 
nullary constructors $\Left_0$ (representing 
$\Left(\Nil)$) and $\Right_0$ (representing 
$\Right(\Nil)$) to $D$ 
and  corresponding constructors to type expressions.
However,  we refrain from these
additions since their comparably small benefits would not
match the considerable complications they would create.

2. 
While $\ire{a}{(\forall x\,A)} \equiv \forall x\,(\ire{a}{A})$ holds, 
$\re(\forall x\,A) \equiv \forall x\,\re\, A$ does not hold in general
since $\re(\forall x\,A) = \exists a\, \ire{a}{(\forall x\,A)} = \exists a\, \forall x\, \ire{a}{A}$ whereas
$\forall x\,\re\, A = \forall x\,\exists a\,\ire{a}{A}$.

3. Regarding (c) vs.\ (d) we note that for Harrop formulas $A$, 
$\reah(A)$ need not be equivalent to $A$. 
In fact, $A$ and $\reah(A)$ may even contradict each other.
For example, if 
$A$ is the Harrop formula 
$\neg\forall x\,(x=0 \lor x \neq 0)$, then
$\reah(A)$ is $\neg\exists a\,\forall x\,(a = \Left(\Nil) \land x=0 \lor 
a = \Right(\Nil) \land x \neq 0)$ which is intuitionistically provable from 
$0\neq 1$. On the other hand $\neg A$ is provable in classical logic.
Hence, $\re A \to A$ is classically contradictory and therefore 
unprovable in $\RIFP$. 
The reason for this difference between $A$ and $\re A$ is \emph{logical}, 
more precisely it lies in the uniform interpretation of the 
universal quantifier which forbids a realizer of a formula 
$\forall x\,B$ to depend on $x$. In contrast, in Kleene realizability
the main source of discrepancy between realizability and truth is 
\emph{computational}
and follows from the existence of undecidable predicates. For example,
the formula $C \eqdef \forall x\in\NN (\halt(x) \lor \neg \halt(x))$ 
is classically true but not realizable since any realizer, 
which in Kleene realizability has to be computable, would solve the 
halting problem (and hence $\neg C$ is classically false but realizable). 
In our setting $C$ is realizable since the domain $D$ admits non-computable
functions. 

4. A crucial property of our realizability interpretation is that
$\bot$ can be a realizer of a formula.
For example, $\ire{a}{(\False \to A)}$ for any $a : \tau(A)$. In particular,
$\ire{\bot}{(\False \to A)}$ for any 
non-Harrop formula $A$.
This enables us to manipulate non-terminating computation in logic and
extract non-terminating programs from logical proofs.
On the other hand, $\ire{\bot}{A}$ does not hold if $A$ is a Harrop formula.

5. Although, by Lemma~\ref{lem-realizability}~(e), realizers are typable,
they may be partial as remarked above. Therefore our realizability is
closer to Kleene's realizability by (codes of) partial recursive 
functions~\cite{Kleene45},
rather than Kreisel's modified realizability~\cite{Kreisel59} whose 
characteristic feature is that realizers are typed and \emph{total}.
For example, it is easy to see that the schema \emph{Independence of Premise}, 
$(A \to \exists x\in\NN\,B) \to \exists x\in\NN\,(A \to B)$ where $A$ is a 
Harrop formula, which is realizable in modified realizability, 
is not realizable in our system.

6.
Despite the availability of classical logic through disjunction-free axioms 
our interpretation is very different from Krivine's classical
realizability~\cite{Krivine01,Krivine03}. While our interpretation
fundamentally rests on the intuitionistic interpretation of disjunction
as a problem whose solution requires a decision between two alternatives,
Krivine's classical realizability is formulated in the negative fragment
of logic given with implication, conjunction and universal quantification as the
only logical connectives. In \cite{OlivaStreicher08} it is shown that
Krivine's realizability (roughly) corresponds to G\"odel's negative translation
followed by intuitionistic realizability.


 

\section{Soundness}
\label{sec-soundness}
The Soundness Theorem, stating that provable formulas are realizable, is the
theoretical foundation for program extraction.

\begin{theorem}[Soundness]
\label{thm-soundness}
Let $\ax$ be a set of nc axioms.
From an $\IFP(\ax)$ proof of a closed 
formula $A$ one can extract a closed program 
$M:\tau(A)$ such that $\ire{M}{A}$ is provable in $\RIFP(\ax)$.

More generally, let $\Gamma$ be a set of Harrop formulas and
$\Delta$ a set of non-Harrop formulas. 
Then, from an $\IFP(\ax)$ proof of 
a formula $A$ 
from the assumptions $\Gamma,\Delta$
one can extract a program $M$ with $\FV(M) \subseteq \vec u$ 
such that $\vec u : \tau(\Delta) \vdash M: \tau(A)$ 
is derivable by the typing rules
of Lemma~\ref{lem-typ-rifp}
and $\ire{M}{A}$ is 
provable in $\RIFP(\ax)$
from the assumptions $\reah(\Gamma)$,
$\ire{\vec u}{\Delta}$, 
and 
$\reali{X}\subseteq\adummy{\alpha_X}$ for all predicate variables
$X$ occurring in $\Gamma,\Delta,A$.




\end{theorem}
In this section we prove this theorem {(Sect.~\ref{sub-soundness-proof})} 
and read off from it a 
program extraction procedure for $\IFP$-proofs {(Sect.~\ref{sub-pe})}. 
We also 
study the realizers of natural numbers {(Sect.~\ref{sub-rnat})} 
and wellfounded induction {(Sect.~\ref{sub-rwf})}.

\emph{Remarks\/}. 
1. From the general version of the Soundness Theorem 
one sees that Harrop formulas $B$ can be freely used as assumptions 
(or axioms) as long as their Harrop interpretations $\reah(B)$ are true. 
For example, $\BT_{\less}$ (Brouwer's Thesis, defined in Sect.~\ref{sub-wf})
is a Harrop formula (for an arbitrary 
relation $\less$) and one can show that $\reah(\BT_{\less})$ is
equivalent to $\BT_{\re\,\less}$ and hence true. Therefore,
$\BT_{\less}$ (without restriction on the relation $\less$) 
can be used as an axiom in a proof without spoiling program extraction.

2. Since $\RIFP(\ax)$ is an instance of $\IFP$ it follows from the Tarskian
soundness of $\IFP$ (see Sect.~\ref{sub-ifp}) that the statements
$M: \tau(A)$ and $\ire{M}{A}$ in the Soundness Theorem are true, in particular
$M$ denotes indeed a realizer of $A$.

\subsection{Proof of the Soundness Theorem} 
\label{sub-soundness-proof}
The expected proof of the Soundness Theorem 
by structural induction on $\IFP$ derivations faces the obstacle that
in order to prove realizability of s.p.\ induction and coinduction
one needs realizers for the monotonicity of the operators in question, 
and this, in turn, requires the realizability of s.p. induction 
and coinduction. 
We escape this circularity 
by introducing an equivalent system 
$\IFP'$ for which soundness can be proven by induction
on the length of 
derivations. The only difference between
the two systems is that $\IFP'$ requires a monotonicity proof for the operator
as an additional premise of s.p.\ induction and coinduction.

Let $\monprop{}{\Phi}$ be the following formula expressing the monotonicity of $\Phi$:
$$\monprop{}{\Phi} \eqdef X \subseteq Y \to \Phi(X) \subseteq \Phi(Y)$$ where 
$X$ and $Y$ are fresh predicate variables. The system $\IFP'$ is obtained from $\IFP$
by replacing the rules $\IND(\Phi,P)$ and $\COIND(\Phi,P)$ by
\[
  \infer[\IND'(\Phi,P)\quad (*)]
{
  \mu(\Phi) \subseteq P 
  }
{
  \Phi(P) \subseteq P\ \ \ \ \monprop{}{\Phi}
 }
\]
\qquad
\[
  \infer[\COIND'(\Phi,P)\quad (*)]
{
  P \subseteq \nu(\Phi) 
  }
{
  P \subseteq \Phi(P) \ \ \ \ \monprop{}{\Phi}
 }
\]
$(*)$ is the side condition that the free assumptions in the proof of 
$\monprop{}{\Phi}$ must not contain $X$ or $Y$ free.

The modified rules $\si'(\Phi,P)$, $\hsi'(\Phi,P)$, $\sci'(\Phi,P)$, 
$\hsci'(\Phi,P)$, are defined similarly.

By the \emph{length of a derivation} we mean the number of occurrences of 
derivation rules.

\begin{lemma}
\label{lem-subst-deriv}
If $\IFP$, $\IFP'$, or $\RIFP$ proves $\Gamma\vdash A$, 
then the same system proves 
$\Gamma[P/X]\vdash A[P/X]$, $\Gamma[P/\pcv{X}]\vdash A[P/\pcv{X}]$ and, 
if applicable, $\Gamma[\rho/\alpha]\vdash A[\rho/\alpha]$,
with the same derivation length, 
where $A$, $P$, $X$, $\rho$, $\alpha$ are arbitrary formulas, predicates, 
predicate variables, types, type variables
respectively, and $\pcv{X}$ is an arbitrary predicate constant that does 
not appear in any axiom.
\end{lemma}
\begin{proof}
Easy structural induction on derivations.
\end{proof}
\emph{Remark\/}. 
Important instances of Lemma~\ref{lem-subst-deriv}
are derivations of $\monprop{}{\Phi}$, which occur as premises of
the rules $\IND'$ and $\COIND'$. If we replace in $\monprop{}{\Phi}$ 
one or both of the predicate variables $X$ and $Y$ by different fresh 
predicate constants, say $\pcv{X}$ and $\pcv{Y}$, then, 
by Lemma~\ref{lem-subst-deriv}, the resulting
formulas have derivations of the same length. 
This fact will be used in the soundness proof for $\IFP'$ 
(Thm~\ref{thm-ifp'}).

In Theorem~\ref{thm-ifp'} we will use 
the following monotone predicate transformers:
\[
\begin{array}{llllllll}
(f^{-1}\circ Q)(\vec x,a)  &\eqdef&  Q(\vec x,f\,a) && 
(f\circ Q)(\vec x,b)       &\eqdef&  \exists a\,(f\,a = b \land Q(\vec x,a)) \\
(a^{-1} * Q)(\vec x)       &\eqdef&  Q(\vec x,a) &&
     (a * P)(\vec x,b)     &\eqdef&  a = b \land P(\vec x) \\
\timesd{P}(\vec x,b)       &\eqdef&  P(\vec x) &&
\exists(Q)(\vec x)         &\eqdef&  \exists a\,Q(\vec x,a) 

\end{array}
\]

The next lemma states their relevant properties. We omit the easy proofs.
\begin{lemma}
\label{lem-trans}
{\emph{Equivalences\/}.}
\[
\begin{array}{lll}
f^{-1} \circ (g^{-1} \circ Q) \equiv (g \circ f)^{-} \circ Q &&
f \circ (g \circ Q) \equiv (f \circ g) \circ Q \\
a^{-1} * (f^{-1} \circ Q) \equiv (f\,a)^{-1} * Q &&
f \circ (a  * P) \equiv (f\,a) * P \\
f^{-1} \circ \timesd{P} \equiv \timesd{P} &&
\exists(f\circ Q) \equiv \exists(Q)
\end{array}
\]
\begin{eqnarray*}
f^{-1}\circ P\cap g^{-1}\circ Q &\equiv&\funpair{f}{g}^{-1}
\circ(\projl^{-1}\circ P\cap\projr^{-1}\circ Q)\\
f\circ P\cup g\circ Q &\equiv&\funsum{f}{g}\circ(\Left\circ P\cup\Right\circ Q)
\end{eqnarray*}

\emph{Adjunctions\/}.
\[
\begin{array}{lll}
Q \subseteq f^{-1}\circ Q' &\leftrightarrow& f \circ Q \subseteq Q'\\
P \subseteq a^{-1} * Q &\leftrightarrow& a * P \subseteq Q\\
Q \subseteq \timesd{P} &\leftrightarrow& \exists(Q) \subseteq P
\end{array}
\]
\emph{Realizability\/}. Below let $Q,Q'$ be non-Harrop predicates, 
$P,P'$ Harrop predicates, $f:\ftyp{\tau(Q)}{\tau(Q')}$ and $a:\tau(Q)$:
\[
\begin{array}{lllll}
\ire{f}{(Q \subseteq Q')} 
   &\leftrightarrow& \rea(Q) \subseteq f^{-1}\circ \rea(Q') 
   &\leftrightarrow& f \circ \rea(Q) \subseteq \rea(Q')\\
\ire{a}{(P \subseteq Q)} 
   &\leftrightarrow& \reah(P) \subseteq a^{-1} * \rea(Q) 
   &\leftrightarrow& a * \reah(P) \subseteq \rea(Q)\\
\reah(Q \subseteq P) 
   &\leftrightarrow& \rea(Q) \subseteq \timesd{\reah(P)} 
   &\leftrightarrow& \exists(\rea(Q)) \subseteq \reah(P)\\
\reah(P \subseteq P') 
   &\leftrightarrow& \reah(P) \subseteq \reah(P') 
   &&
\end{array}
\]
\begin{eqnarray*}
\rea(Q\cap Q') &\equiv& \projl^{-1} \circ\rea(Q) \cap \projr^{-1} \circ \rea(Q')\\
\rea(Q\cup Q') &\equiv& \Left \circ\rea(Q) \cup \Right \circ \rea(Q')
\end{eqnarray*}
\end{lemma}




\begin{theorem}[$\IFP'$ version of Soundness]\label{thm-ifp'}
Let $\ax$ be a set of nc axioms.
From an $\IFP'(\ax)$ proof of a formula $A$ one can extract a closed program 
$M:\tau(A)$ such that $\ire{M}{A}$ is provable in $\RIFP(\ax)$.

More generally, let $\Gamma$ be a set of Harrop formulas,
$\Delta$ a set of non-Harrop formulas,
and $A$, and let $\vec X$ be the free predicate variables of $\Gamma,\Delta,A$.
Then, from an $\IFP'(\ax)$ proof of $A$ 
from the assumptions $\Gamma,\Delta$ and any vector $\vec u$ of distinct program variables
such that each variable in $\vec u$ is assigned to a unique formula in $\Delta$,
one can extract a program $M$ with $\FV(M) \subseteq \vec u$ 
such that 
\begin{enumerate}
\item[(i)] $\vec u : \tau(\Delta) \vdash M: \tau(A)$ 
is derivable by the typing rules
of Lemma~\ref{lem-typ-rifp},
and 
\item[(ii)] $\ire{M}{A}$ is 
provable in $\RIFP(\ax)$
from the assumptions $\reah(\Gamma)$,
$\ire{\vec u}{\Delta}$, 
and $\reali{\vec X}\subseteq\adummy{\alpha_{\vec X}}$.
\end{enumerate}
In (ii), we mean by $\reali{\vec X}\subseteq\adummy{\alpha_{\vec X}}$ 
all formulas 
$\reali{X}\subseteq\adummy{\alpha_{X}}$ where $X\in\vec X$, i.e.\ 
$X$ is a free predicate variable of $\Gamma,\Delta,A$.

The above statement holds for all formulas $A$. However, if
$A$ is Harrop, then it simplifies to:
If $\IFP'(\ax)$ proves $A$
from the assumptions $\Gamma,\Delta$, 
then  $\RIFP(\ax)$ proves $\reah(A)$ 
from the assumptions $\reah(\Gamma)$,
$\re\,\Delta$,
and $\reali{\vec X}\subseteq\adummy{\alpha_{\vec X}}$.




\end{theorem}

\begin{proof}
We first observe that in (ii), to prove $\ire{M}{A}$, we may assume in addition
$\vec u : \tau(\Delta)$, 
and $M:\tau(A)$.
This is so, since, by Lemma~\ref{lem-realizability}~(e), the assumptions
$\reali{\vec X}\subseteq\adummy{\alpha_{\vec X}}$ 
and $\ire{\vec u}{\Delta}$ imply $\vec u : \tau(\Delta)$, 
and therefore, by (i), $M:\tau(A)$ (since the typing rules are provable in 
$\RIFP$).

In the following we will refer to the assumptions
$\reali{\vec X}\subseteq\adummy{\alpha_{\vec X}}$ as 
the \emph{type correctness assumptions}.

The proof is by induction on the length of $\IFP'$ derivations.

In the following we mean by `induction hypothesis' always an induction
hypothesis of the induction on the length of derivations.
In order to avoid confusion with $\IFP'$ induction on a strictly positive 
inductive predicate $\mu(\Phi)$, we will refer to  the latter always 
as `s.p.~induction'.
Furthermore, when writing $\vec u:\rho\vdash M:\sigma$, we mean that this judgment
is derivable by the typing rules. Finally, `provable' means `provable in $\RIFP(\ax)$'.

We only look at some critical cases. 
In Section~\ref{sub-pe} the extracted programs
for all rules are shown.

\underline{Assumption rule}. 
Let $A$ be derived from the assumption $A$.


\emph{(a) Case $A$ is non-Harrop}.
Then $A\in\Delta$ and therefore some $u\in\vec{u}$ is assigned to $A$.
We choose that $u$ as the extracted program.
Hence, we have to show
\begin{enumerate}
\item[(i)] $\vec u : \tau(\Delta) \vdash u: \tau(A)$, 
\item[(ii)] $\ire{u}{A}$ is 
provable from the assumptions $\reah(\Gamma)$,
$\ire{\vec u}{\Delta}$, 
and $\reali{\vec X}\subseteq\adummy{\alpha_{\vec X}}$.
\end{enumerate}
But this is trivially the case since $u: \tau(A)$ occurs in $\vec u:\tau(\Delta)$
and $\ire{u}{A}$ occurs in $\ire{\vec u}{\Delta}$.

\emph{(b) Case $A$ is Harrop}. 
Then $A\in\Gamma$.
Hence $\RIFP(\ax)$ proves $\reah(A)$ from the assumptions $\reah(\Gamma)$.

\underline{Implication introduction}.
Assume we derived $A\to B$ from $\Gamma,\Delta$ by implication introduction,
i.e.~from a derivation of $B$ from the assumptions $\Gamma$, $\Delta$ and  $A$. 

\emph{(a) Case $A$ and $B$ are both non-Harrop}.
Then $\tau(A\to B)$ is $\ftyp{\tau(A)}{\tau(B)}$ and
$\ire{f}{(A\to B)}$ is $f:\ftyp{\tau(A)}{\tau(B)} \land 
\forall a\,(\ire{a}{A}\to \ire{(f\,a)}{B})$. 

By the induction hypothesis we have a program $M$ such that
\begin{enumerate}
\item[(ih-i)] $\vec{u}:\tau(\Delta), u:\tau(A)\vdash M:\tau(B)$,
\item[(ih-ii)] $\ire{M}{B}$ is provable from the assumptions 
$\reah(\Gamma)$,
$\ire{\vec u}{\Delta}$, $\ire{u}{A}$, 
and $\reali{\vec X}\subseteq\adummy{\alpha_{\vec X}}$.
\end{enumerate}
We let $f \eqdef \lambda u\,M$ be the program extracted from the 
given proof of $A\to B$.
We have to show
\begin{enumerate}
\item[(i)] $\vec u : \tau(\Delta) \vdash 
f: \tau(A \to B)$, 
\item[(ii)] $\ire{f}{(A\to B)}$ is 
provable from $\reah(\Gamma)$,
$\ire{\vec u}{\Delta}$, 
and $\reali{\vec X}\subseteq\adummy{\alpha_{\vec X}}$.
\end{enumerate}
From (ih-i), we get $\vec{u}:\tau(\Delta)\vdash f:\ftyp{\tau(A)}{\tau(B)}$,
by the lambda-abstraction rule. Since $\ftyp{\tau(A)}{\tau(B)} = \tau(A\to B)$,
this shows (i).
To prove (ii), assume
$\reah(\Gamma)$,
$\ire{\vec u}{\Delta}$, 
and $\reali{\vec X}\subseteq\adummy{\alpha_{\vec X}}$.
By the initial observation, we may in additon assume 
$f:\tau(A\to B)$. Therefore, it only remains to show that
$\forall a\,(\ire{a}{A}\to\ire{(f\,a)}{B})$.
But this follows immediately from (ih-ii).



\emph{(b) Case $A$ is Harrop, $B$ is non-Harrop}.
Then $\tau(A\to B)$ is $\tau(B)$ and
$\ire{b}{(A\to B)}$ is $b:\tau(B) \land (\reah(A)\to \ire{b}{B})$. 

By the induction hypothesis we have a program $M$ such that
\begin{enumerate}
\item[(ih-i)] $\vec{u}:\tau(\Delta)\vdash M:\tau(B)$, 
\item[(ih-ii)] $\ire{M}{B}$ is provable from the assumptions 
$\reah(\Gamma)$, 
$\reah(A)$,
$\ire{\vec u}{\Delta}$,
and $\reali{\vec X}\subseteq\adummy{\alpha_{\vec X}}$.
\end{enumerate}
We choose $M$ as the program extracted from the 
given proof of $A\to B$.
We have to show
\begin{enumerate}
\item[(i)] $\vec u : \tau(\Delta) \vdash M: \tau(B)$ 
(which is the same as (ih-i)),
\item[(ii)] $\ire{M}{(A\to B)}$ is 
provable from $\reah(\Gamma)$,
$\ire{\vec u}{\Delta}$, 
and $\reali{\vec X}\subseteq\adummy{\alpha_{\vec X}}$.
\end{enumerate}
To prove (ii), assume
$\reah(\Gamma)$,
$\ire{\vec u}{\Delta}$, 
and $\reali{\vec X}\subseteq\adummy{\alpha_{\vec X}}$.
By the initial observation, we may in additon assume 
$M:\tau(B)$. Therefore, it only remains to show 
$\reah(A)\to\ire{M}{B}$.
But this follows immediately from (ih-ii).

\emph{(c) Case $A$ is non-Harrop, $B$ is Harrop}.
Then $A\to B$ is a Harrop formula and $\reah(A\to B)$ is $\re\,A \to \reah(B)$.

By the induction hypothesis, $\RIFP(\ax)$ proves $\reah(B)$ 
from the assumptions $\reah(\Gamma)$,
$\re\,\Delta$, $\re\,A$,
and $\reali{\vec X}\subseteq\adummy{\alpha_{\vec X}}$.
Hence, $\RIFP(\ax)$ proves $\reah(A\to B)$ 
from the assumptions $\reah(\Gamma)$,
$\re\,\Delta$,
and $\reali{\vec X}\subseteq\adummy{\alpha_{\vec X}}$.

\emph{(d) Case $A$ and $B$ are both Harrop}.
Then $A\to B$ is a Harrop formula and $\reah(A\to B)$ is $\reah(A) \to \reah(B)$.

By the induction hypothesis, $\RIFP(\ax)$ proves $\reah(B)$ 
from the assumptions $\reah(\Gamma)$, $\reah(A)$,
$\re\,\Delta$,
and $\reali{\vec X}\subseteq\adummy{\alpha_{\vec X}}$.
Hence, $\RIFP(\ax)$ proves $\reah(A\to B)$ 
from the assumptions $\reah(\Gamma)$,
$\re\,\Delta$,
and $\reali{\vec X}\subseteq\adummy{\alpha_{\vec X}}$.

\underline{Implication elimination}.
Assume we have derived $B$ under the assumptions $\Gamma$, $\Delta$ 
by implication elimination from $A\to B$ and $A$, each under the same assumptions. 
We may assume that $A$ contains no free predicate variables other than those
in $\Gamma,\Delta,A$. Otherwise, we use Lemma~\ref{lem-subst-deriv} and 
substitute any new free predicate variable
in $A$ by some closed predicate, say $\lambda \vec x.\False$.
We distinguish cases whether $A$ and $B$ are Harrop or not but we are less detailed
regarding the type correctness assumptions $\reali{X}\subseteq\adummy{\alpha_X}$ 
since they are dealt with
exactly as in the case of implication introduction.

\emph{(a) Case $A$ and $B$ are both non-Harrop}.
By the induction hypotheses we have programs $M$ and $N$ such that
$\vec{u}:\tau(\Delta) \vdash M:\ftyp{\tau(A)}{\tau(B)}$ and 
$\vec{u}:\tau(\Delta) \vdash N:\tau(A)$, 
and furthermore $\RIFP(\ax)$ proves $\ire{M}{(A\to B)}$ and $\ire{N}{A}$
from the assumptions
$\reah(\Gamma)$, $\ire{\vec u}{\Delta}$, and $\vec u : \tau(\Delta)$.
Hence, 
$\vec{u}:\tau(\Delta) \vdash (M\,N):\tau(B)$
and $\RIFP(\ax)$ proves 
$\ire{(M\,N)}{B}$
from the assumptions $\reah(\Gamma)$, $\ire{\vec u}{\Delta}$, 
and $\vec u : \tau(\Delta)$.

\emph{(b) Case $A$ is Harrop, $B$ is non-Harrop}.
By the first induction hypothesis, we have program $M$ such that
$\vec{u}:\tau(\Delta) \vdash M:\tau(B)$ and
$\RIFP(\ax)$ proves $\ire{M}{(A\to B)}$, that is,
$M:\tau(B) \land (\reah(A) \to \ire{M}{B})$
from the assumptions
$\reah(\Gamma)$, $\ire{\vec u}{\Delta}$, and $\vec u : \tau(\Delta)$.
By the second induction hypothesis, 
$\RIFP(\ax)$ proves $\reah(A)$ from the same assumptions.
Hence, $\RIFP(\ax)$ proves 
$\ire{M}{B}$ from the same assumptions.

\emph{(c) Case $A$ is non-Harrop, $B$ is Harrop}.
By the first induction hypothesis, 
$\RIFP(\ax)$ proves $\reah(A\to B)$, that is,
$\exists a\,(\ire{a}{A}) \to\reah(B)$,
from the assumptions
$\reah(\Gamma)$, $\ire{\vec u}{\Delta}$, and $\vec u : \tau(\Delta)$.
By the second induction hypothesis, we have program $N$ such that
$\vec{u}:\tau(\Delta) \vdash N:\tau(A)$ an
$\RIFP(\ax)$ proves $\ire{N}{A}$
from the assumptions
$\reah(\Gamma)$, $\ire{\vec u}{\Delta}$, and $\vec u : \tau(\Delta)$.
Hence, $\RIFP(\ax)$ proves 
$\reah(B)$ from the same assumptions.

\emph{(d) Case $A$ and $B$ are both Harrop}. Easy.

\underline{Existence elimination}.
Assume we have derived $B$ by existence elimination from 
$\exists x\,A$ and $\forall x\,(A \to B)$ where $x$ is not free in $B$.
As in the previous case, we may assume that $A$ does not contain
new predicate variables.

\emph{(a) Case $A$ and $B$ are both non-Harrop}.
By the induction hypotheses we have programs $M$ and $N$ such that
$\vec{u}:\tau(\Delta) \vdash M:\tau(A)$ and 
$\vec{u}:\tau(\Delta) \vdash N:\ftyp{\tau(A)}{\tau(B)}$, 
and furthermore $\RIFP(\ax)$ proves $\exists x\,(\ire{M}{A})$ and 
$\forall x\,\forall a\,((\ire{a}{A}) \to \ire{(N\,a)}{B})$, 
both from the assumptions
$\reah(\Gamma)$, $\ire{\vec u}{\Delta}$, and $\vec u : \tau(\Delta)$.
Hence 
$\vec{u}:\tau(\Delta) \vdash (N\,M):\tau(B)$
and $\RIFP(\ax)$ proves 
$\ire{(N\,M)}{B}$
from the assumptions $\reah(\Gamma)$, $\ire{\vec u}{\Delta}$, 
and $\vec u : \tau(\Delta)$.

The other cases are similar.



For s.p.~induction and s.p.~coinduction we consider an operator
$\Phi = \lambda X\,Q$, hence $\Phi(P) = Q[P/X]$.

\underline{$\IND'(\Phi,P)$}. 
Assume we have derived $\mu(\Phi)\subseteq P$ under the assumptions $\Gamma,\Delta$, 
by s.p~induction from 
$\Phi(P)\subseteq P$ (i.e.\ $Q[P/X]\subseteq P$), 
and $\monprop{}{\Phi}$ (i.e.\ $X \subseteq Y \to Q \subseteq Q[Y/X]$),
from the same assumptions, where $Y$ is a fresh predicate variable.
Hence $X$ and $Y$ do not occur free in $\Gamma,\Delta$ and $\mu(\Phi)\subseteq P$.

\emph{(a) Case $\Phi$ and $P$ are both non-Harrop\/}. 
By the induction hypothesis for $Q[P/X]\subseteq P$ 
and Lemma~\ref{lem-typ-subst}~(c), 
we have a program $s$ such that
\begin{enumerate}
\item[(a-i)] 
$\vec{u}:\tau(\Delta)\vdash s:\ftyp{\tau(Q)[\tau(P)/\alpha_X]}{\tau(P)}$, 
\item[(a-ii)] $\ire{s}{(Q[P/X]\subseteq P)}$ is provable from the assumptions 
$\reah(\Gamma)$, $\ire{\vec u}{\Delta}$, 
and $\reali{\vec X}\subseteq\adummy{\alpha_{\vec X}}$.
\end{enumerate}
By the induction hypothesis for $\monprop{}{\Phi}$, 
we have a program $m$ such that
\begin{enumerate}
\item[(b-i)] $\vec{u}:\tau(\Delta)\vdash 
  m:\ftyp{(\ftyp{\alpha_X}{\alpha_Y})}{(\ftyp{\tau(Q)}{\tau(Q)[\alpha_Y/\alpha_X]})}$, 
\item[(b-ii)] $\ire{m}{\monprop{}{\Phi}}$ is provable from the assumptions 
$\reah(\Gamma)$, $\ire{\vec u}{\Delta}$, 
and $\reali{\vec X}\subseteq\adummy{\alpha_{\vec X}}$, 
$\reali{X}\subseteq\adummy{\alpha_{X}}$,
$\reali{Y}\subseteq\adummy{\alpha_{Y}}$.
\end{enumerate}
We define recursively $f\eqrec s \circ m\,f$, and choose $f$ as 
the program extracted from the given proof of $\mu(\Phi)\subseteq P$.
We set $\rho\eqdef\tfix{\alpha_X}{\tau(Q)} = \tau(\mu(\Phi))$,
so that $\tau(\mu(\Phi)\subseteq P) = \ftyp{\rho}{\tau(P)}$.
Hence, we have to show
\begin{enumerate}
\item[(i)] $\vec u : \tau(\Delta) \vdash f: \ftyp{\rho}{\tau(P)}$, 
\item[(ii)] $\ire{f}{(\mu(\Phi)\subseteq P)}$ is 
provable from $\reah(\Gamma)$,
$\ire{\vec u}{\Delta}$, 
and $\reali{\vec X}\subseteq\adummy{\alpha_{\vec X}}$.
\end{enumerate}
Since, as one easily sees, the typing rules are closed under type substitution,
(b-i) remains true after a type substitution, 
hence we also have
$\vec{u}:\tau(\Delta)\vdash 
 m:\ftyp{(\ftyp{\rho}{\tau(P)})}{(\ftyp{\tau(Q)[\rho/\alpha_X]}{\tau(Q)[\tau(P)/\alpha_X]})}$.
Together with (a-i) and the typing rule for recursively defined functions, (i) follows.

To prove (ii), we assume $\reah(\Gamma)$,
$\ire{\vec u}{\Delta}$, and $\reali{\vec X}\subseteq\adummy{\alpha_{\vec X}}$.
By Lemma~\ref{lem-trans} and the definition of $\rea(\mu(\Phi))$
and since we know $f :\ftyp{\rho}{\tau(P)}$ (by the initial observation),
our goal $\ire{f}{(\mu(\Phi)\subseteq P)}$ is equivalent to
\[\mu(\lambda\reali{X}\,\rea(Q)[\rho/\alpha_X])\subseteq f^{-1} \circ \rea(P).\]
We show by s.p.\ induction the stronger statement
\[\mu(\lambda\reali{X}\,\rea(Q)[\rho/\alpha_X])\subseteq 
                         (f^{-1} \circ \rea(P)) \cap\adummy{\rho}.\]
Hence, we show
\[\rea(Q)[\rho/\alpha_X][(f^{-1}\circ\rea(P))\cap\adummy{\rho}/\reali{X}] 
     \subseteq (f^{-1} \circ \rea(P))\cap\adummy{\rho}.\]
From Lemma~\ref{lem-realizability}~(e), we get 
$\rea(Q)[\rho/\alpha_X][(f^{-1}\circ\rea(P))\cap\adummy{\rho}/\reali{X}] 
     \subseteq \adummy{\rho}$. 
Hence, it remains to show
\[\rea(Q)[\rho/\alpha_X][(f^{-1}\circ\rea(P))\cap\adummy{\rho}/\reali{X}] 
     \subseteq f^{-1} \circ \rea(P).\]
The formula $\ire{m}{\monprop{}{\Phi}}$ of (b-ii) is, 
by Lemma~\ref{lem-trans} and by Lemma~\ref{lem-realizability}~(a), 
equivalent to
\[ \forall f : \ftyp{\alpha_X}{\alpha_Y}\,
     (\reali{X}\subseteq f^{-1}  \circ \reali{Y} \to  
         \rea(Q) \subseteq (m\,f)^{-1}\circ
                           \rea(Q)[\reali{Y}/\reali{X}][\alpha_Y/\alpha_X])\]
and is provable under the extra assumptions
$\reali{X}\subseteq\adummy{\alpha_{X}}$ and
$\reali{Y}\subseteq\adummy{\alpha_{Y}}$.
Instantiating $f$ with our extracted program and setting
$\alpha_X \eqdef \rho$, 
$\alpha_Y \eqdef\tau(P)$,
$\reali{X}\eqdef (f^{-1}  \circ \rea(P)) \cap \adummy{\rho}$,
$\reali{Y}\eqdef \rea(P)$,
the assumption $f : \ftyp{\alpha_X}{\alpha_Y}$ becomes provable
(by the initial observation),
the assumption $\reali{X}\subseteq\adummy{\alpha_{X}}$ becomes a tautology,
the assumption $\reali{Y}\subseteq\adummy{\alpha_{Y}}$ become provable 
(by Lemma~\ref{lem-realizability}~(e)),
and the premise $\reali{X}\subseteq f^{-1}  \circ \reali{Y}$ becomes a tautology.
Hence we have,
%
%
using
Lemma~\ref{lem-subst-deriv} for $\RIFP'$,
\begin{equation}
\label{eq-atwo}
\rea(Q)[\rho/\alpha_X][(f^{-1}\circ\rea(P))\cap\adummy{\rho}/\reali{X}]
          \subseteq (m\,f)^{-1}\circ \rea(Q)[\rea(P)/\reali{X}][\tau(P)/\alpha_X]
\end{equation}
By (a-ii), we have $\ire{s}{(Q[P/X]\subseteq P)}$, which, by Lemma~\ref{lem-trans},
is equivalent to
\begin{equation}
\label{eq-aone}
\rea(Q[P/X]) \subseteq s^{-1} \circ \rea(P) .
\end{equation}
Since, by Lemma~\ref{lem-realizability}~(a), 
$\rea(Q)[\rea(P)/\reali{X}][\tau(P)/\alpha_X] = \rea(Q[P/X])$ and since composition
is monotone w.r.t.~inclusion, we obtain
\begin{eqnarray*}
\rea(Q)[\rho/\alpha_X][(f^{-1}\circ\rea(P))\cap \adummy{\rho}/\reali{X}] 
&\stackrel{\hbox{(\ref{eq-atwo})}}{\subseteq}&
(m\,f)^{-1} \circ \rea(Q[P/X])\\
&\stackrel{\hbox{(\ref{eq-aone})}}{\subseteq}&
(m\,f)^{-1} \circ (s^{-1} \circ \rea(P))\\
&\equiv&
(s \circ m\,f)^{-1} \circ \rea(P)
\end{eqnarray*}
where the last step uses Lemma~\ref{lem-trans}.
Since $s \circ m\,f = f$, we are done.

In the remaining cases we are less detailed regarding the type correctness assumptions
since they can be dealt with exactly as above.

\emph{(b) Case $\Phi$ and $P$ are both Harrop\/} (then $\mu(\Phi)$ and 
$Q[P/X]$ are Harrop). 
We aim to prove $\reah(\mu(\Phi)\subseteq P)$, that is,
$\mu(\lambda X\,\reah_X(Q)) \subseteq \reah(P)$.
We try s.p.~induction, so our goal is to prove
$\reah_X(Q)[\reah(P)/X] \subseteq \reah(P)$, i.e. 
\[\reah(Q[\pcv{X}/X])[\reah(P)/\pcv{X}] \subseteq \reah(P)\]
By the  first induction hypothesis (the second induction hypothesis is not needed) 
we have $\reah(\Phi(P)\subseteq P)$, i.e.\  
$\reah(Q[P/X]) \subseteq \reah(P)$.
Since, by Lemma~\ref{lem-realizability}~(b), 
$\reah(Q[P/X]) = \reah(Q[\pcv{X}/X])[\reah(P)/\pcv{X}]$,
we are done.

\emph{(c) Case $\Phi$ is non-Harrop, $P$ is Harrop\/} (then $\mu(\Phi)$ and 
$Q[P/X]$ are non-Harrop). 
We aim to prove $\reah(\mu(\Phi)\subseteq P)$, which,
by Lemma~\ref{lem-trans}, is equivalent to
$\mu(\lambda \reali{X}\,\rea(Q)[\rho/\alpha_X]) \subseteq \timesd{\reah(P)}$
where $\rho \eqdef\tfix{\alpha_X}{\tau(Q)}$.
%
We show by s.p.\ induction the stronger statement
\[\mu(\lambda\reali{X}\,\rea(Q)[\rho/\alpha_X])\subseteq 
                               \timesd{\reah(P)} \cap\adummy{\rho}\]
(note that 
$(\timesd{\reah(P)} \cap\adummy{\rho})(\vec x,b) \equiv \reah(P)(\vec x) \land b:\rho$).
Hence, we show
\[\rea(Q)[\rho/\alpha_X][\timesd{\reah(P)}\cap\adummy{\rho}/\reali{X}] 
     \subseteq \timesd{\reah(P)}\cap\adummy{\rho}.\]
From Lemma~\ref{lem-realizability}~(e), we get 
$\rea(Q)[\rho/\alpha_X][\timesd{\reah(P)}\cap\adummy{\rho}/\reali{X}] 
     \subseteq \adummy{\rho}$. 
Hence, it remains to show
\[\rea(Q)[\rho/\alpha_X][\timesd{\reah(P)}\cap\adummy{\rho}/\reali{X}] 
     \subseteq \timesd{\reah(P)}.\]
By the first  induction hypothesis we have 
$\reah(\Phi(P)\subseteq P)$, that is,
\begin{equation}
\label{eq-cone}
\rea(Q[P/X]) \subseteq \timesd{\reah(P)}.
\end{equation}
Furthermore we have an $\IFP'$ derivation of 
$X\subseteq Y \to Q \subseteq Q[Y/X]$ and therefore,
by Lemma~\ref{lem-subst-deriv}, also an $\IFP'$ derivation of
$X\subseteq P \to Q \subseteq Q[P/X]$ of the same height.
Hence, by the induction hypothesis, we have
$m: \ftyp{\tau(Q)}{\tau(Q[P/X])}$ such that $\RIFP$ derives 
$\ire{m}{(X\subseteq P \to Q \subseteq Q[P/X])}$, that is,
\[\reali{X}\subseteq\timesd{\reah(P)} \to \rea(Q) \subseteq m^{-1}\circ \rea(Q[P/X]),\] 
from the 
extra assumption $\reali{X}\subseteq\adummy{\alpha_X}$.
Using this with $\reali{X}\eqdef\timesd{\reah(P)}\cap\adummy{\rho}$ 
and $\alpha_X\eqdef\rho$ we obtain
\begin{equation}
\label{eq-ctwo}
\rea(Q)[\rho/\alpha_X][\timesd{\reah(P)}\cap\adummy{\rho}/\reali{X}] \subseteq 
              m^{-1}\circ\rea(Q[P/X]),
\end{equation}
without the assumption $\reali{X}\subseteq\adummy{\alpha_X}$.
Now
\begin{eqnarray*}
\rea(Q)[\rho/\alpha_X][\timesd{\reah(P)}\cap\adummy{\rho}/\reali{X}] 
&\stackrel{\hbox{(\ref{eq-ctwo})}}{\subseteq}&
m^{-1} \circ \rea(Q[P/X])\\
&\stackrel{\hbox{(\ref{eq-cone})}}{\subseteq}&
m^{-1} \circ \timesd{\reah(P)}\\
&=&
\timesd{\reah(P)}
\end{eqnarray*}

\emph{(d) Case $\Phi$ is Harrop, $P$ is non-Harrop\/}. 

\emph{Subcase $X$ is not free in $Q$\/}. 
The goal to find a realizer 
$\tilde{a}:\tau(P)$
of $\mu(\Phi) \subseteq P$
can be written as $\mu(\lambda X\,\reah_X(Q)) \subseteq \tilde{a}^{-1}\circ \rea(P)$
whose s.p.~inductive proof, in this case, boils down to proving 
$\reah(Q) \subseteq \tilde{a}^{-1}\circ\rea(P)$. But such an $\tilde{a}$ is
provided by the  induction hypothesis as a realizer of $\Phi(P)\subseteq P$.

\emph{Subcase $X$ is free in $Q$\/} (then $Q$, $Q[P/X]$ and 
$\monprop{}{\Phi}[\pcv{X}/X][P/Y]$ are non-Harrop). 
We need to find 
$\tilde{a}:\tau(P)$
such that $\ire{\tilde{a}}{(\mu(\Phi)\subseteq P)}$, 
which is equivalent to 
$\mu(\lambda X\,\reah_X(Q)) \subseteq \tilde{a}^{-1} * \rea(P)$.
A proof attempt by s.p.~induction leads to the goal
\[\reah_X(Q)[\tilde{a}^{-1} * \rea(P)/X] \subseteq \tilde{a}^{-1} * \rea(P).\]
By the  induction hypothesis we have 
$s:\ftyp{\tau(Q)[\tau(P)/\alpha_X]}{\tau(P)}$ such that
$\ire{s}{(Q[P/X]\subseteq P)}$, equivalently,
\begin{equation}
\label{eq-done}
\rea(Q[P/X]) \subseteq s^{-1} \circ \rea(P),
\end{equation}
and, with a similar justification as previously, some
$m:\ftyp{\tau(P)}{\tau(Q)[\tau(P)/\alpha_X]}$ realizing the formula
$\pcv{X}\subseteq P \to Q[\pcv{X}/X] \subseteq Q[P/X]$, i.e.\ 
\[\forall a:\tau(P)\,(\pcv{X}\subseteq a^{-1} * \rea(P)\to 
\reah(Q[\pcv{X}/X]) \to (m\,a)^{-1} * \rea(Q[P/X])).\]
We define recursively $\tilde{a} \eqrec s\, (m\,\tilde{a})$ which clearly has 
type $\tau(P)$, as required.
Using the above formula with $a \eqdef \tilde{a}$ and 
$\pcv{X}\eqdef\tilde{a}^{-1} * \rea(P)$ we obtain
\begin{equation}
\label{eq-dtwo}
\reah(Q[\pcv{X}/X])[\tilde{a}^{-1} * \rea(P)/\pcv{X}] \subseteq (m\,a)^{-1} * \rea(Q[P/X])
\end{equation}
We show, by s.p.~induction, that $\tilde{a}$ realizes $\mu(\Phi)\subseteq P$:  
\begin{eqnarray*}
\reah_X(Q)[\tilde{a}^{-1} * \rea(P)/X]
&=&
\reah(Q[\pcv{X}/X])[\tilde{a}^{-1} * \rea(P)/\pcv{X}]\\
&\stackrel{\hbox{(\ref{eq-dtwo})}}{\subseteq}&
(m\,\tilde{a})^{-1} * \rea(Q[P/X])\\
&\stackrel{\hbox{(\ref{eq-done})}}{\subseteq}&
(m\,\tilde{a})^{-1} * (s^{-1} \circ \rea(P))\\
&=&(s\, (m\,\tilde{a}))^{-1} * \rea(P)
\end{eqnarray*}

\underline{$\COIND'(\Phi,P)$}. 
This is largely dual to $\IND'(\Phi,P)$. 

Assume we have derived $P\subseteq \nu(\Phi)$ under the assumptions 
$\Gamma,\Delta$, 
by s.p~coinduction from $P \subseteq Q[P/X]$, 
and $\monprop{}{\Phi}$
from the same assumptions.

\emph{(a) Case $\Phi$ and $P$ are both non-Harrop\/}. 
By the induction hypothesis for $P\subseteq Q[P/X]$ 
and Lemma~\ref{lem-typ-subst}~(c), 
we have a program $s$ such that
\begin{enumerate}
\item[(a-i)] 
$\vec{u}:\tau(\Delta)\vdash s:\ftyp{\tau(P)}{\tau(Q)[\tau(P)/\alpha_X]}$, 
\item[(a-ii)] $\ire{s}{(P\subseteq Q[P/X])}$ is provable from the assumptions 
$\reah(\Gamma)$, $\ire{\vec u}{\Delta}$, 
and $\reali{\vec X}\subseteq\adummy{\alpha_{\vec X}}$.
\end{enumerate}
By the induction hypothesis for $\monprop{}{\Phi}$ 
we have a program $m$ such that (exactly as in the case of s.p.\ induction)
\begin{enumerate}
\item[(b-i)] $\vec{u}:\tau(\Delta)\vdash 
  m:\ftyp{(\ftyp{\alpha_X}{\alpha_Y})}{(\ftyp{\tau(Q)}{\tau(Q)[\alpha_Y/\alpha_X]})}$,
\item[(b-ii)] $\ire{m}{\monprop{}{\Phi}}$ is provable from the assumptions 
$\reah(\Gamma)$, $\ire{\vec u}{\Delta}$, 
and $\reali{\vec X}\subseteq\adummy{\alpha_{\vec X}}$, 
$\reali{X}\subseteq\adummy{\alpha_{X}}$,
$\reali{Y}\subseteq\adummy{\alpha_{Y}}$.
\end{enumerate}
We define recursively $f\eqrec m\,f\circ s$, and choose $f$ as 
the program extracted from the given proof of $P\subseteq\nu(\Phi)$.
We set $\rho\eqdef\tfix{\alpha_X}{\tau(Q)}$, 
so that $\tau(\nu(\Phi)\subseteq P) = \ftyp{\tau(P)}{\rho}$.
Hence, we have to show
\begin{enumerate}
\item[(i)] $\vec u : \tau(\Delta) \vdash f: \ftyp{\tau(P)}{\rho}$, 
\item[(ii)] $\ire{f}{(P\subseteq\nu(\Phi))}$ is 
provable from $\reah(\Gamma)$,
$\ire{\vec u}{\Delta}$, 
and $\reali{\vec X}\subseteq\adummy{\alpha_{\vec X}}$.
\end{enumerate}
Since by (b-i),
$\vec{u}:\tau(\Delta)\vdash 
 m:\ftyp{(\ftyp{\tau(P)}{\rho})}{(\ftyp{\tau(Q)[\tau(P)/\alpha_X]}{\tau(Q)[\rho/\alpha_X]})}$,
(i) follows with (a-i) and the typing rule for recursively defined functions.

To prove (ii), we assume $\reah(\Gamma)$,
$\ire{\vec u}{\Delta}$, and $\reali{\vec X}\subseteq\adummy{\alpha_{\vec X}}$.
By Lemma~\ref{lem-trans} and the definition of $\rea(\nu(\Phi))$ 
and since we know that $f : \ftyp{\tau(P)}{\rho}$, 
our goal $\ire{f}{(P\subseteq \nu(\Phi))}$ is equivalent to
\[f \circ \rea(P)\subseteq \nu(\lambda\reali{X}\,\rea(Q)[\rho/\alpha_X]) .\]
We show this by s.p.\ coinduction (unlike in the case of s.p.\ induction, no stronger 
statement needed here). Hence we show
\[f\circ\rea(P)\subseteq\rea(Q)[\rho/\alpha_X][(f\circ\rea(P))/\reali{X}] .\]
The formula $\ire{m}{\monprop{}{\Phi}}$ of (b-ii) is 
equivalent to
\[ \forall f : \ftyp{\alpha_X}{\alpha_Y}\,
     (f\circ\reali{X}\subseteq \reali{Y} \to  
  (m\,f)\circ \rea(Q) \subseteq \rea(Q)[\reali{Y}/\reali{X}][\alpha_Y/\alpha_X])\]
and is provable under the extra assumptions
$\reali{X}\subseteq\adummy{\alpha_{X}}$ and
$\reali{Y}\subseteq\adummy{\alpha_{Y}}$.
Instantiating $f$ with our extracted program and setting
$\alpha_X \eqdef \tau(P)$, 
$\alpha_Y \eqdef\rho$,
$\reali{X}\eqdef \rea(P)$,
$\reali{Y}\eqdef f  \circ \rea(P)$,
the assumptions $f : \ftyp{\alpha_X}{\alpha_Y}$ and
$\reali{X}\subseteq\adummy{\alpha_{X}}$ become provable,
and the premise $f\circ\reali{X}\subseteq \reali{Y}$ becomes a tautology.
The assumption $\reali{Y}\subseteq\adummy{\alpha_{Y}}$ becomes
$\forall \vec x,a\,(\rea(P)(\vec x,a) \to (f\,a):\rho)$ which follows from
Lemma~\ref{lem-realizability}~(e) and the earlier established fact that 
$f:\ftyp{\tau(P)}{\rho}$. 
Hence,
\begin{equation}
\label{eq-coatwo}
(m\,f)\circ \rea(Q)[\rea(P)/\reali{X}][\tau(P)/\alpha_X] \subseteq 
   \rea(Q)[f\circ\rea(P)/\reali{X}][\rho/\alpha_X]
\end{equation}
By (a-ii), we have $\ire{s}{(P\subseteq Q[P/X])}$, which
is equivalent to
\begin{equation}
\label{eq-coaone}
s\circ\rea(P)\subseteq \rea(Q[P/X]) .
\end{equation}
Since $\rea(Q)[\rea(P)/\reali{X}][\tau(P)/\alpha_X] = \rea(Q[P/X])$, we have
\begin{eqnarray*}
\rea(Q)[\rho/\alpha_X][(f\circ\rea(P))\cap \adummy{\rho}/\reali{X}] 
&\stackrel{\hbox{(\ref{eq-coatwo})}}{\supseteq}&
(m\,f) \circ \rea(Q[P/X])\\
&\stackrel{\hbox{(\ref{eq-coaone})}}{\supseteq}&
(m\,f) \circ (s \circ \rea(P))\\
&\equiv&
((m\,f)\circ s) \circ \rea(P)\\
&\equiv&
f \circ \rea(P)
\end{eqnarray*}

\emph{(b) Case $\Phi$ and $P$ are both Harrop\/}.  Dual to case (b) for $\IND'$.

\emph{(c) Case $\Phi$ is non-Harrop, $P$ is Harrop\/} (then $\nu(\Phi)$, 
$Q[P/X]$ and $Q[Y/X]$ are non-Harrop). 
By the induction hypothesis we have 
$s:\tau(Q[P/X])$ such that
$\ire{s}{(P\subseteq Q[P/X])}$, that is, 
\begin{equation}
\label{eq-ccone}
s * \reah(P) \subseteq \rea(Q[P/X]),
\end{equation}
Furthermore, by Lemma~\ref{lem-subst-deriv}, the second premise
instantiated with $X \eqdef P$ yields a shorter derivation of
$P\subseteq Y \to Q[P/X] \to Q[Y/X]$, and therefore,
by the induction hypothesis, we have 
$m:\ftyp{\alpha_Y}{\ftyp{\tau(Q[P/X])  }{\tau(Q)[\alpha_Y/\alpha_X]}}$ 
such that $\RIFP$ proves
$\ire{m}{(P\subseteq Y \to Q[P/X] \to Q[Y/X])}$,
that is,
\[\forall a:\alpha_Y\,(a * \reah(P)\subseteq \reali{Y} \to  
    (m\, a)\circ \rea(Q[P/X])\subseteq \rea(Q[Y/X]))\]
from the extra assumption $\reali{Y}\subseteq\adummy{\alpha_Y}$.

We define recursively $\tilde{a}\eqrec m\, \tilde{a}\,s$,
and show that this realizes $P \subseteq \nu(\Phi)$.

Substituting $\alpha_Y$ with $\rho\eqdef\tfix{\alpha_X}{\tau(Q)}$ 
we get $m:\ftyp{\rho}{\ftyp{\tau(Q[P/X])  }{\rho}}$ 
($\rho = \tau(Q)[\rho/\alpha_X]$!) and therefore
$\tilde{a}:\rho$, 
which is the correct type.
Substituting further $\reali{Y}$ with $\tilde{a} * \reah(P)$, we get
\begin{equation}
\label{eq-cctwo}
(m\, \tilde{a})\circ \rea(Q[P/X])\subseteq 
          \rea(Q)[\rho/\alpha_X][\tilde{a}*\reah(P)/\reali{X}].
\end{equation}
The assumption $\reali{Y}\subseteq\adummy{\alpha_Y}$ disappears since it 
becomes a tautology.
It remains to show
$\tilde{a} * \reah(P) \subseteq \nu(\lambda \reali{X}\,\rea(Q)[\rho/\alpha_X])$,
which we do by coinduction:
\begin{eqnarray*}
\rea(Q)[\rho/\alpha_X][\tilde{a} * \reah(P)/\reali{X}]
&\stackrel{\hbox{(\ref{eq-cctwo})}}{\supseteq}&
(m\, \tilde{a})\circ \rea(Q[P/X])\\
&\stackrel{\hbox{(\ref{eq-ccone})}}{\supseteq}&
(m\, \tilde{a})\circ (s * \reah(P))\\
&=&
(m\, \tilde{a}\,s) * \reah(P)
\end{eqnarray*}

\emph{(d) Case $\Phi$ is Harrop, $P$ is non-Harrop\/}. 

\emph{Subcase $X$ is not free in $Q$\/}. 
We have to show $\reah(P \subseteq \nu(\Phi))$, equivalently, 
$\exists(\rea(P)) \subseteq \nu(\lambda X\,\reah(Q))$. 
By s.p.~coinduction, this reduces to $\exists(\rea(P)) \subseteq \reah(Q)$
which is equivalent to the  induction hypothesis, $\reah(P\subseteq Q)$. 

\emph{Subcase $X$ is free in $Q$\/} (then $Q$, $Q[P/X]$ and 
$\monprop{}{\Phi}[P/X][\pcv{X}/Y]$ are non-Harrop). 
We need to prove $\reah(P\subseteq \nu(\Phi))$, that is, 
$\exists(\rea(P)) \subseteq \nu(\lambda X\,\reah_X(Q))$. 
S.p.\ coinduction reduces this to the goal
\[\exists(\rea(P)) \subseteq \reah_X(Q)[\exists(\rea(P))/X].\]
By the  induction hypothesis we have $\ire{s}{(P\subseteq \Phi(P))}$, 
equivalently,
\begin{equation}
\label{eq-cdone}
s \circ \rea(P) \subseteq \rea(Q[P/X]),
\end{equation}
and $\reah(\monprop{}{\Phi}[P/X][\pcv{X}/Y])$, that is,
\[\exists(\rea(P)) \subseteq \pcv{X} \to 
          \exists(\rea(Q[P/X])) \subseteq \reah(Q[\pcv{X}/X]).\]
Using 
Lemma~\ref{lem-subst-deriv} for $\RIFP'$ with
$\pcv{X} \eqdef \exists(\rea(P))$ yields
\begin{equation}
\label{eq-cdtwo}
\exists(\rea(Q[P/X])) \subseteq \reah(Q[\pcv{X}/X])[\exists(\rea(P))/\pcv{X}]
    = \reah_X(Q)[\exists(\rea(P))/X].
\end{equation}
Now, 
\begin{eqnarray*}
\reah_X(Q)[\exists(\rea(P))/X]
&\stackrel{\hbox{(\ref{eq-cdtwo})}}{\supseteq}&
\exists(\rea(Q[P/X]))\\
&\stackrel{\hbox{(\ref{eq-cdone})}}{\supseteq}&
\exists(s \circ \rea(P))\\
&\equiv& \exists(\rea(P)).
\end{eqnarray*}

We conclude the proof with the strong and half strong variants of s.p.\ induction
and coinduction. Since, as remarked in Sect.~\ref{sub-ifp}, these variants
are derivable from ordinary s.p.\ induction and coinduction, they do not need to be
treated separately. We will do this nevertheless in order to obtain simpler
realizers. We only derive these simplified realizers for those instances that
will be used later although simplified realizers can be given in all cases where the
conclusion of a rule is a non-Harrop formula. 
 
\underline{$\hsci'(\Phi,P)$}. 
Assume we have derived $P\subseteq\nu(\Phi)$ from the premises
 $P \subseteq \Phi(P)\cup \nu(\Phi)$, that is, 
$P\subseteq Q[P/X]\cup \nu(\Phi)$, 
and $\monprop{}{\Phi}$, that is, $X \subseteq Y \to Q \subseteq Q[Y/X]$.

\emph{Case $\Phi$ and $P$ are both non-Harrop\/}. 
By the  induction hypothesis we have 
$s:\ftyp{\tau(P)}{(\tau(Q)[\tau(P)/\alpha_X] + \rho)}$ 
such that
$\ire{s}{(P \subseteq \Phi(P) \cup \nu(\Phi))}$, that is, 
by Lemma~\ref{lem-trans} for $\IFP'$, 
\begin{equation}
\label{eq-hscaone}
s \circ \rea(P)\subseteq \Left \circ \rea(Q[P/X]) \cup 
                         \Right \circ \nu(\rea(\Phi))
\end{equation}
Furthermore, the induction hypothesis and Lemma~\ref{lem-typ-subst}~(c), 
yield a program 
$m:\ftyp{(\ftyp{\tau(P)}{\alpha_Y})}
        {\ftyp{\tau(Q)[\tau(P)/\alpha_X]}{\tau(Q)[\alpha_Y/\alpha_X]}}$ 
with $\ire{m}{(\monprop{}{\Phi[P/X]})}$, that is,
\[\forall f:\ftyp{\tau(P)}{\alpha_Y}\,( f \circ \rea(P) \subseteq \reali{Y} 
        \to  m\,f\circ\rea(Q[P/X]) \subseteq \rea(Q[Y/X])),\] 
under the extra assumption $\reali{Y}\subseteq\adummy{\alpha_Y}$.

We define recursively $\tilde{f}\eqrec \funsum{(m\, \tilde{f})}{\idty} \circ s$
and show this realizes $P\subseteq\nu(\Phi)$.

Substituting $\alpha_Y$ with $\rho\eqdef\tfix{\alpha_X}{\tau(Q)}$, we obtain the typing 
$m:\ftyp{\ftyp{\tau(P)}{\rho})}
        {\ftyp{\tau(Q)[\tau(P)/\alpha_X]}{\tau(Q)[\rho/\alpha_X]}}$
and therefore $\tilde{f}:\ftyp{\tau(P)}{\rho}$, as required.
Substituting $f$ with $\tilde{f}$ and $\reali{Y}$ with $\tilde{f}\circ\rea(P)$, 
we obtain, using again Lemma~\ref{lem-subst-deriv}
\begin{equation}
\label{eq-hscatwo}
m\,\tilde{f} \circ \rea(Q[P/X]) \subseteq 
\rea(Q)[\tilde{f}\circ\rea(P)/\reali{X}]\,,
\end{equation}
without extra assumption, since from Lemma~\ref{lem-realizability}~(e) 
we get $\tilde{f}\circ \rea(P)\subseteq\adummy{\rho}$.
It remains to show $\tilde{f}\circ \rea(P)\subseteq \nu(\rea(\Phi))$.
We will prove this by half strong coinduction, so 
our goal is to prove
(since $\rea(\Phi) = \lambda\reali{X}\,\rea(Q)[\rho/\alpha_X]$)
\[\tilde{f}\circ \rea(P)\subseteq 
     \rea(Q)[\rho/\alpha_X][\tilde{f}\circ \rea(P)/\reali{X}]
                                     \cup \nu(\rea(\Phi))\,.\]
Indeed,
\begin{eqnarray*}
&&\rea(Q)[\rho/\alpha_X][\tilde{f}\circ \rea(P)/\reali{X}] \cup \nu(\rea(\Phi)) \\
&\stackrel{(\ref{eq-hscatwo})}{\supseteq}&
(m\, \tilde{f} \circ \rea(Q[P/X])) \cup \nu(\rea(\Phi)) \\
&\stackrel{\hbox{Lemma }\ref{lem-trans}}{\equiv}&
\funsum{(m\, \tilde{f})}{\idty} \circ (\Left \circ \rea(Q[P/X]) \cup 
\Right \circ \nu(\rea(\Phi))) \\
&\stackrel{(\ref{eq-hscaone})}{\supseteq}&
\funsum{(m\, \tilde{f})}{\idty} \circ (s \circ \rea(P)) \\
&\stackrel{\hbox{Lemma }\ref{lem-trans}}{\equiv}&
(\funsum{(m\, \tilde{f})}{\idty} \circ s) \circ \rea(P) 
\end{eqnarray*}

\underline{$\sci'(\Phi,P)$}, \emph{case $\Phi$ and $P$ are both non-Harrop\/}. 
Using the  induction hypothesis with realizers $s$ of $P \subseteq \Phi(P \cup \nu(\Phi))$,
and $m$ of $\monprop{}{\Phi[P/X]}$, one sees, with a similar reasoning as above,
that the recursive definition
$\tilde{f}\eqrec (m\,\funsum{\tilde{f}}{\idty}) \circ s$
provides a realizer of $P \subseteq \nu(\Phi)$.

\underline{$\hsi'(\Phi,P)$}, 
\emph{case $\Phi$ is Harrop but not constant, $P$ is non-Harrop\/}. 
Using the  induction hypothesis with realizers 
$s$ of $\Phi(P) \cap \mu(\Phi) \subseteq P$,
and $m$ of $\monprop{}{\Phi[P/X]}$, one sees that the recursive definition
$\tilde{a}\eqrec s\,(m\,\tilde{a})$
provides a realizer of $\mu(\Phi) \subseteq P$ (which is the same as the realizer for the 
corresponding instance of s.p.\ induction).
\end{proof}

\begin{lemma}\label{lem-mon-new}
  $\monprop{}{\Phi}$ is provable in IFP'.
\end{lemma}
\begin{proof}
We define
$\monprop{X}{P} \eqdef X \subseteq X' \to P \subseteq P[X'/X]$ where 
$X'$ is a fresh variable accompanied with $X$. 
Then, for $\Phi = \lambda X\, P$, $\monprop{}{\Phi}$ is equivalent 
to $\monprop{X}{P}$.  
Therefore,  we prove
$\monprop{X}{P}$ by  induction on $P$. 
That is, prove $\monprop{X}{P}$ assuming that $\monprop{Y}{Q}$ holds for every
operator $\lambda Y\,Q$ such that $Q$ is a subexpression of $P$.

For the case that $P$ has the form $\mu(\lambda Y\,Q)$ 
we assume $X \subseteq X'$ and 
show $\mu(\lambda Y\,Q)\subseteq \mu(\lambda Y\,Q[X'/X])$.
Here, we may assume that $Y \not\in\{X,X'\}$. 
We use IFP'-induction on $\mu(\lambda Y\,Q)$ and hence have to show
\begin{equation}\label{eq:eap}
Q[\mu(\lambda Y\,Q[X'/X])/Y]\subseteq \mu(\lambda Y\,Q[X'/X])
\end{equation}
and $\monprop{}{\lambda Y\, Q}$, that is, $\monprop{Y}{Q}$.
The latter holds by the  induction hypothesis.
But $\monprop{X}{Q}$ also holds.
Therefore, $Q \subseteq Q[X'/X]$.
Thus, by Lemma \ref{lem-subst-deriv}, $Q[\mu(\lambda Y\,Q[X'/X])/Y] \subseteq 
Q[X'/X][\mu(\lambda Y\,Q[X'/X])/Y]$ holds.
Furthermore, by closure, 
$Q[X'/X][\mu(\lambda Y\,Q[X'/X])/Y]\subseteq \mu(\lambda Y\,Q[X'/X])$.
Thus, we have (\ref{eq:eap}).

For the case that $P$ has the form $\nu(\lambda Y\,Q)$,
the argument is completely dual if we replace $\monprop{X}{P}$
by the equivalent formula $X' \subseteq X \to P[X'/X] \subseteq P$. 

The remaining cases are easy using the extracted programs in Sect.~\ref{sub-pe}
as a guide.
\end{proof}

\begin{proof}[Proof of the Soundness Theorem for 
$\IFP$ (Thm.~\ref{thm-soundness})]  
  From an $\IFP$ proof one can obtain an $\IFP$' proof of the same formula
  by Lemma \ref{lem-mon-new}.  Therefore 
we obtain the result by Theorem~\ref{thm-ifp'}.
\end{proof}

\subsection{Program extraction}
\label{sub-pe}
The proof of the Soundness Theorem contains an algorithm
for computing the realizing program $M$ which we now describe.
We also note how to produce
a Haskell program at the end of this section.
For brevity we write a derivation judgement $\Gamma\vdash d:A$
as $d^A$, suppressing the context.

For an $\IFP$ derivation $d^A$ the extracted program $\ep{d^A}$ is defined
as 
\[\ep{d^A} \eqdef \epp{\pt{d^A}}\]
where $\pt{\cdot}$ is the transformation
of $\IFP$ proofs into $\IFP'$ proofs based on Lemma~\ref{lem-mon-new}, 
and $\epp{\cdot}$ is the program extraction procedure based on 
Theorem~\ref{thm-ifp'}. 

The transformation $\pt{d^A}$ simply replaces recursively every 
subderivation of the form $\indu{e^{\Phi(P)\subseteq P}}$ by 
$\induprime{\pt{e^{\Phi(P)\subseteq P}},\monproof{\Phi}^{\monprop{}{\Phi}}}$
where $\monproof{\Phi}^{\monprop{}{\Phi}}$ is the proof described in 
Lemma~\ref{lem-mon-new}.
Similarly, $\coind{e^{\Phi(P)\subseteq P}}$ is replaced by 
$\coindprime{\pt{e^{\Phi(P)\subseteq P}},\monproof{\Phi}^{\monprop{}{\Phi}}}$, 
and so on.

The extraction procedure $\epp{d^A}$ is defined by recursion on derivations as follows:

If $A$ is Harrop then $\epp{d^A} \eqdef \Nil$. 
Hence, in the following we assume that the proven formula is non-Harrop.

Closure and coclosure are realized by the identity:
\begin{eqnarray}
\epp{\clos{\Phi}^{\Phi(\mu(\Phi))\subseteq \mu(\Phi)}} =
\epp{\cocl{\Phi}^{\nu(\Phi)\subseteq\Phi(\nu(\Phi))}} =  \lambda a\,.\,a
\label{eq:closure}
\end{eqnarray}
For induction, in the case where $P$ is non-Harrop, the extracted program is
%
\vspace{1em}

$\epp{\induprime{d^{\Phi(P)\subseteq P}, 
           e^{\monprop{}{\Phi}}}^{\mu(\Phi)\subseteq P}} =$
\begin{eqnarray}
&&
\left\{
\begin{array}{ll}         
\rec\,(\lambda a\,.\,\epp{d} \circ \epp{e}\, a)&\hbox{if $\Phi$ is non-Harrop}\\  
\rec\,(\lambda a\,.\,\epp{d}\, (\epp{e[\pcv{X}/X]}\, a)) &\hbox{otherwise.}
\end{array}
\right.
\label{eq:ind}
\end{eqnarray}
For coinduction in the case where $\Phi$ is non-Harrop, the extracted program is
%
\vspace{1em}

$\epp{\coindprime{d^{P\subseteq \Phi(P)},
              e^{\monprop{}{\Phi}}}^{{P \subseteq \nu(\Phi)}}} =$
\begin{eqnarray}
&&
\left\{
\begin{array}{ll}
\rec\,(\lambda a\,.\,\epp{e}\, a\circ \epp{d}) &\hbox{if $P$ is non-Harrop}\\  
\rec\,(\lambda a\,.\, (\epp{e[\pcv{X}/X]}\, a\,\epp{d}))  &\hbox{otherwise.}
\end{array}
\right.
\label{eq:coind}
\end{eqnarray}

For the strong and half-strong versions of induction and coinduction we only
present a few cases that will be used later.
For half strong induction in the case where  
$\Phi$ is Harrop but not constant and $P$ is non-Harrop, the extracted program is
the same as for induction, namely
\[\epp{\hsindu{d^{\Phi(P)\cap\mu(\Phi)\subseteq P},
           e^{\monprop{}{\Phi}}}^{\mu(\Phi)\subseteq P}} = 
   \rec\,(\lambda a\,.\,\epp{d}(\epp{e[\pcv{X}/X]}\, a)).\]
For half strong coinduction in the case where both 
$\Phi$ and $P$ are non-Harrop, the extracted program is
\begin{equation}
\epp{\hscoind{d^{P\subseteq \Phi(P)\cup\nu(\Phi)},
              e^{\monprop{}{\Phi}}}^{P \subseteq \nu(\Phi)}} =
     \rec\,(\lambda a\,.\,\funsum{\epp{e}\, a}{\idty}\circ \epp{d}).
\label{eq:hscoi}
\end{equation}
For strong coinduction in the case where both 
$\Phi$ and $P$ are non-Harrop, the extracted program is
\[\epp{\scoind{d^{P\subseteq \Phi(P\cup\nu(\Phi))},
              e^{\monprop{}{\Phi}}}^{{P \subseteq \nu(\Phi)}}} =
     \rec\,(\lambda a\,.\,(\epp{e}\funsum{a}{\idty})\circ \epp{d}).\]

Assumptions are realized by variables, and the congruence rule does 
not change the realizer:
\begin{eqnarray*}
\epp{u_i^{A_i}} &=& u_i\\
\epp{\congr{d^{P(s)}}{e^{s=t}}{P}^{P(t)}} &=& \epp{d}
\end{eqnarray*}
The logical rules are realized as follows:
\begin{eqnarray*}
\epp{\oril{d^{A}}{B}^{A\lor B}} &=& \Left(\epp{d}) \\
\epp{\orir{d^{B}}{A}^{A\lor B}} &=& \Right(\epp{d}) \\
\epp{\ore{d^{A\lor B}}{e^{A\to C}}{f^{B\to C}}^C} &=&
     \case\,\epp{d}\,\of\,\\
  &&\quad\{\Left(a) \to \epp{e}*a\,;\,\\
  &&\quad\ \Right(b) \to \epp{f}*b\} 
\end{eqnarray*}
where $\epp{e}*a$ means $\epp{e}\,a$ if 
$A$ in non-Harrop and $\epp{e}$ if $A$ is Harrop. Similarly for $\epp{f}*b$.
\begin{eqnarray*}
\epp{\andi{d^A}{e^B}^{A\land B}} &=& \left\{
\begin{array}{ll}
\epp{d} &\hbox{if $B$ is Harrop}\\
\epp{e} &\hbox{if $A$ is Harrop}\\
\Pair(\epp{d},\epp{e}) &\hbox{otherwise}
\end{array}
\right.
\\
\epp{\andel{d^{A\land B}}^{A}} &=& \left\{
\begin{array}{ll}
\epp{d} &\hbox{if $B$ is Harrop}\\
\projl(\epp{d}) &\hbox{otherwise}
\end{array}
\right.
\\
\epp{\ander{d^{A\land B}}^{B}} &=& \left\{
\begin{array}{ll}
\epp{d} &\hbox{if $A$ is Harrop}\\
\projr(\epp{d}) &\hbox{otherwise}
\end{array}
\right.
\end{eqnarray*}
\begin{eqnarray*}
\epp{(\impi{u^A}{d^B})^{A\to B}} &=& \left\{
\begin{array}{ll}
\epp{d} &\hbox{if $A$ is Harrop}\\
\lambda u.\,\epp{d} &\hbox{otherwise}
\end{array}
\right.
\\
\epp{(\impe{d^{A\to B}}{e^A})^B} &=& \left\{
\begin{array}{ll}
\epp{d} &\hbox{if $A$ is Harrop}\\
\epp{d}\,\epp{e} &\hbox{otherwise}
\end{array}
\right.
\end{eqnarray*}
\begin{eqnarray*}
\epp{\forall^+_{x}(d^A)^{\forall x\,A}} &=& \epp{d} \\
%
\epp{\forall^-_{t}(d^{\forall x\, A})^{A[t/x]}} &=& \epp{d}\\
\epp{\exists^{+}_{\lambda x A,t}({d^{A[t/x]}})^{\exists x\,A}} &=& \epp{d}\\
\epp{\exe{d^{\exists x\,A}}{e^{\forall x\,(A \to B)}}^{B}} &=& \left\{
\begin{array}{ll}
\epp{e} &\hbox{if $A$ is Harrop}\\
\epp{e}\,\epp{d} &\hbox{otherwise}
\end{array}
\right.
\end{eqnarray*}

\paragraph{Extraction into Haskell.}
By the Soundness Theorem (Thm.~\ref{thm-soundness}) one can extract from
a proof of a formula $A$ a realizing program $M$ such that the typing 
$M:\tau(A)$ can be derived using the rules given in Lemma~\ref{lem-typ-rifp}.
The extraction procedure $\epp{\cdot}$ implicitly computes not only $M$
but a typing derivation for $M:\tau(A)$. Composing this with the 
translation of $\RIFP$ programs into Haskell one obtains an 
extraction procedure directly into Haskell. It is easy to see that the 
composed procedure can be obtained by the following small modifications of 
$\epp{\cdot}$ which we call $\epph{\cdot}$.
In addition to replacing $\Pair(M,N)$ by $(M,N)$, the definition
of $\epp{\cdot}$ is changed for closure and coclosure 
derivation rules with $\Phi = \lambda X P$, $\alpha = \alpha_X$, 
and $\rho = \tau(P)$ from (\ref{eq:closure}) to
\begin{eqnarray*}
\epph{\clos{\Phi}^{\Phi(\mu(\Phi))\subseteq \mu(\Phi)}} =
\mathsf{roll}_{\cona}\\
\epph{\cocl{\Phi}^{\nu(\Phi)\subseteq\Phi(\nu(\Phi))}} =  
\mathsf{unroll}_{\cona}
\label{eq:closuren-ew}
\end{eqnarray*}
For induction and coinduction with $\Phi = \lambda X Q$,
$\alpha = \alpha_X$ and $\rho = \tau(Q)$,
we use the following definitions instead of
(\ref{eq:ind}) and (\ref{eq:coind}).

\vspace{1em}
$\epph{\induprime{d^{\Phi(P)\subseteq P}, 
           e^{\monprop{}{\Phi}}}^{\mu(\Phi)\subseteq P}} =$
\begin{eqnarray*}
&&
\left\{
\begin{array}{ll}         
\rec\,(\lambda a\,.\, 
  \epph{d} \circ (\epph{e}\, a) \circ \mathsf{unroll}_{\cona})
  &\hbox{if $\Phi$ is non-Harrop}\\  
\rec\,(\lambda a\,.\,
 \epph{d}\, (\epph{e[\pcv{X}/X]}\, a)) &\hbox{otherwise.}
\end{array}
\right.
\end{eqnarray*}

\vspace{1em}
$\epph{\coindprime{d^{P\subseteq \Phi(P)},
              e^{\monprop{}{\Phi}}}^{{P \subseteq \nu(\Phi)}}} =$
\begin{eqnarray*}
&&
\left\{
\begin{array}{ll}
\rec\,(\lambda a.\ \mathsf{roll}_{\cona} \circ (\epph{e}\, a)\circ \epph{d}) &\hbox{if $P$ is non-Harrop}\\  
\rec\,(\lambda a\,.\, (\epph{e[\pcv{X}/X]}\, a\,\,\epph{d}))  &\hbox{otherwise.}
\end{array}
\right.
\end{eqnarray*}
Similar modifications need to be carried out
for the other induction and coinduction schemes.


\subsection{Realizing natural numbers}
\label{sub-rnat}

In Sect.~\ref{subsub-nat} we defined natural numbers 
as a subset of the real numbers through the inductive predicate
$\NN(x) \ \eqmu\ x = 0 \lor \NN(x-1)$.
This view of natural numbers is abstract since no concrete representation
is associated with it. A concrete representation of natural numbers is provided
through the realizability interpretation of the predicate $\NN$.
Note that the formula $\NN(x)$ is not Harrop since it contains a disjunction
at a strictly positive position.
We have $\tau(\NN) = \nat = \tfix{\alpha}{1+\alpha}$, the type of 
natural numbers (see~Sect.~\ref{sub-types}).
Realizability for $\NN$ works out as
\begin{eqnarray*}
\ire{a}{\NN(x)} &\eqmu& a = \Left(\Nil) \land x = 0 \lor  
\exists b\,(a = \Right(b) \land \ire{b}{\NN(x-1))}\,.
\end{eqnarray*}
Therefore, $\ire{a}{\NN(x)}$ means that $a$ is the unary representation of the
natural number $x$.

\begin{lemma}
\label{lem-rea-nat}
\begin{itemize}
\item[(a)] $(\re\,\NN(x)) \leftrightarrow \NN(x)$
\item[(b)] 
$(\re\,\exists x\in\NN\,A(x)) \leftrightarrow (\exists x\in\NN\,\re\,A(x))$.
\item[(c)] $\reah(\forall x\in\NN\,A(x)) \leftrightarrow \forall x\in\NN\,\reah(A(x))$
if $A(x)$ is a Harrop formula.
\item[(d)] $\ire{a}{\NN(x)} \land \ire{b}{\NN(y)} \to (a=b \leftrightarrow x=y)$.
\end{itemize}
\end{lemma}
\begin{proof}
Both implications of part (a) are easily proven by induction.

Parts (b) and (c) follow immediately from (a).

To prove part (d) one can use that natural numbers are non-negative and 
subtraction is an injective function in its first argument.
\end{proof}
\emph{Remark\/}.  In the parts (a-c) of lemma~\ref{lem-rea-nat}, 
$\NN$ may be replaced by any predicate that contains neither implications 
nor universal quantifiers nor free predicate variables. 
However, (d) depends on the concrete definition of $\NN$ 
and specific properties of the theory of real numbers.

By Lemma~\ref{lem-rea-nat} it is safe to identify natural numbers with their realizers.
Henceforth we will use the variables $n,m,k,l,\ldots$
for both.  Hence, in an $\IFP$ proof a natural number is a special
real number while in an extracted program it is a special domain element.
Recall from Sect.~\ref{subsub-nat} that rational numbers are defined
by the predicate 
$\QQ(q) \eqdef \exists x,y,z\in\NN\,(z \neq 0 \land q\cdot z = x-y)$
which corresponds to a representation of rational numbers by 
triples of natural numbers $(n, m, k)$ ($k \neq 0$) denoting $(n - m)/k$.
Although the corresponding statement of Lemma~\ref{lem-rea-nat} ~(d) 
(i.e.~uniqueness of realizers) does not hold for $\QQ$, 
the generalizations of Lemma~\ref{lem-rea-nat}~(a--c) do apply
to $\QQ$. Therefore, one can use realizers to express rational numbers.

\begin{example}
\label{ex-sum2}
In Example~\ref{ex-sum1}, we proved 
$A \eqdef \forall x, y\ (\NN(x) \to \NN(y) \to \NN(x+y))$.
We have $\tau(A) = \ftyp{\nat} \ftyp{\nat}{\nat}$.
According to Lemma~\ref{lem-mon-new}, 
the formula $\monprop{}{\Phi_\NN}
\eqdef X \subseteq Y \to \Phi_\NN(X) \subseteq \Phi_\NN(Y)$
 expressing the monotonicity of the operator
$\Phi_\NN$  is provable in IFP' and the following program 
$\mon_\NN : \ftyp{(\ftyp{\alpha_X}{\alpha_Y})}{\ftyp{\one + \alpha_X}{\one + \alpha_Y}}$ is
extracted from the proof.
\begin{align*}\label{mon}
  \mon_\NN =  \lambda f.
   \lambda m.\, \case\ m\ of \{\Left(a) \to \Left(a); \Right(b) \to \Right(f(b))\}
\end{align*}
Furthermore, from the proof of the induction premise
we extract the following program
of type $\ftyp{(\one + \nat)}{\nat}$.
$$
s = (\lambda m. \case\, m\, \of\, \{\Left(c) \to n; \Right(c) \to \Right(c)\})
$$
Here, $n$ is the realizer of $\NN(x)$.
Therefore,  by (\ref{eq:ind}) of Sect.~\ref{sub-pe},
the realizer extracted from the proof of $A$ is
the following program of type $\ftyp{\nat} \ftyp{\nat}{\nat}$
\begin{align*}
\mathsf{plus} &=  \lambda n.\,  \rec\ \lambda f.\, s \circ (\mon_\NN\ f)\\
&=  \lambda n.\,  \rec\ \lambda f.\, \lambda m.\,  s ((\mon_\NN\ f) m)\, .
\end{align*}
By  
program axiom (ii), $s$ is strict.
Therefore, by Lemma~\ref{lem-bot}, 
we can rewrite
\begin{align*}
 \mathsf{plus}\, n \, m &\eqrec 
                 \case\ m\ \of \{\Left(a) \to s(\Left(a)); \Right(b) \to s(\Right(\mathsf{plus}\ n\ b))\}\\
   &\eqrec 
     \case\ m\ \of \{\Left(a) \to n; \Right(b) \to \Right(\mathsf{plus}\ n\ b)\}.
\end{align*}
\end{example}

\subsection{Realizing wellfounded induction}
\label{sub-rwf}
In this section we work out 
in detail the realizers of wellfounded induction
and its specializations (Sect.~\ref{sub-wf})
as provided by the Soundness Theorem 
(Thm.~\ref{thm-soundness}).
This will be important for understanding the programs  
extracted in Sect.~\ref{sec-realnumbers}.
\begin{lemma}[Realizer of wellfounded induction] 
\label{lem-rea-wfi}
The schema of wellfounded induction, $\wfi_{\less,A}(P)$, is realized as follows.
If $s$ realizes $\prog_{\less ,A}(P)$ where $P$ is non-Harrop, then 
$\acc_{\less}\cap A \subseteq P$ is realized by
\begin{itemize} 
\item[-] $\tilde{f} \eqrec \lambda a.\,(s\,a\,(\lambda a'.\,\lambda b.\, \tilde{f}\,a'))$
if $\less$ and $A$ are both non-Harrop,
\item[-] $\tilde{f} \eqrec \lambda a.\,(s\,a\,\tilde{f})$
if $\less$ is Harrop and $A$ is non-Harrop,
\item[-] $\tilde{c} \eqrec s\,(\lambda b.\, \tilde{c})$
if $\less$ is non-Harrop and $A$ is Harrop,
\item[-] 
$\rec\, s$
if $\less$  and $A$ are both Harrop.
\end{itemize}
\end{lemma}
\begin{proof}
Since $\wfi_{\less,A}(P)$ follows from $\wfi_{\less}(A \Rightarrow P)$ and the latter is an
instance of induction, the extracted programs shown in the lemma can be obtained from 
Theorem~\ref{thm-soundness}. However, it is instructive to give some details of their
derivations.

Recall that $\acc_{\less } = \mu(\Phi)$ where 
$\Phi(X) = \lambda x\,\forall y \less  x\,X(y)$ and 
$\prog_\less (Q) = \Phi(Q)\subseteq Q$.
Since $\Phi$ is a Harrop operator, $\acc_{\less }$ is a Harrop predicate.

According to Theorem~\ref{thm-soundness} and the program extraction 
procedure described in Sect.~\ref{sub-pe} the extracted realizer of 
$\acc_{\less} \subseteq A \Rightarrow P$  is 
\[\tilde{f} \eqrec s'\,(m\,\tilde{f})\]
provided $\ire{s'}{\prog_\less(A \Rightarrow P)}$ and 
$\ire{m}{(\monprop{}{\Phi}[\pcv{X}/X])}$. 
Because 
$s$ realizes $\prog_{\less ,A}(P)$, 
that is,
\[\forall x\,(x\in A \to \forall y\,(y\in A \to y\less  x \to y\in P) \to x \in P)\]
and $\prog_\less(A \Rightarrow P)$ expands to
\[\forall x\,(\forall y (y\less  x \to y\in A \to y\in P) \to x \in A \to x \in P)\]
it is clear that we can define
\begin{itemize} 
\item[-] $s' \eqdef \lambda g. \lambda a.\,(s\,a\,(\lambda a'.\,\lambda b.\,g\,b\,a'))$ 
if $\less$ and $A$ are both non-Harrop,
\item[-] $s' \eqdef \lambda f. \lambda a.\,(s\,a\,f)$ 
if $\less$ is Harrop and $A$ is non-Harrop,
\item[-] $s' \eqdef s$ if $A$ is Harrop.
\end{itemize}
The realizer $m$ of $\monprop{}{\Phi}[\pcv{X}/X]$, which expands to
\[\pcv{X}\subseteq Y\to\forall x\,(\forall y\less x\,\pcv{X}(y))\to \forall y \less  x\,Y(y)\]
is easily extracted as 
\begin{itemize}
\item[-] $\lambda a.\,\lambda b.\,a$ if $\less$ is non-Harrop,
\item[-] $\lambda a.\,a$ if $\less$ is Harrop.
\end{itemize}
From this, one can easily see that the extracted realizer of 
$\acc_{\less} \subseteq A \Rightarrow P$ is as stated in the lemma.
Since 
$\acc_\less $ is Harrop it follows that 
the same program 
realizes the inclusion 
$\acc_{\less}\cap A\subseteq P$.
\end{proof}
Finally, we exhibit the realizers of Archimedean induction.
We only look at the forms $\AI_q$ and $\AIB_q$ since the principles
$\AI$ and $\AIB$ have the same realizers and will not be used in the following.
\begin{lemma}[Realizers of Archimedean induction]
\label{lem-rea-ai}
\begin{description}
\item[]
\item[$\AI_q$] If $s$ realizes 
$\forall x \ne 0\, ((|x| \leq q \to P(2x)) \to P(x))$, 
where $P$ is non-Harrop, 
then 
$\rec\,s$ 
realizes $\forall x \ne 0\,P(x)$.
\item[$\AIB_q$] If $s$ realizes $\forall x \in B\setminus\{0\}\, 
      (P(x) \lor (|x| \le q \land B(2x) \land (P(2x) \to P(x))))$,
where $B$ and $P$ are non-Harrop, 
then 
\begin{equation}
\label{eq:realaib}
a\,b \eqrec
\case\,s\,b\,\of\,\{\Left(c) \to c; \Right(b',d) \to d\,(a\,b')\}
\end{equation}
realizes $\forall x \in B\setminus\{0\} \, P(x)$.
\end{description}
\end{lemma}
\begin{proof}
$\AI_q$ is derived from 
half strong induction $\hsi$ as is shown in Lemma~\ref{lem-ai}, 
and  the realizer of the monotonicity of the operator in questions clearly is the identity.  Therefore, 
as we studied in Section \ref{sub-pe}, 
it has the realizer $a \eqrec s\,a$, that is, $\rec\,s$.

Clearly, the premise of $\AIB_q(B,P)$ implies the premise of 
$\AI_q(B\Rightarrow P)$.  
From a realizer $s$ of the premise of the former one obtains the realizer
\[s'= \lambda a. \lambda b.\,   
      \case\,s\,b\,\of\,\{\Left(c) \to c; 
                          \Right(b',d) \to d\,(a\,b')\}\]
of the premise of the latter. Therefore, $a \eqrec s'\,a$, that is,
\begin{equation*}
a\,b \eqrec
\case\,s\,b\,\of\,\{\Left(c) \to c; \Right(b',d) \to d\,(a\,b')\}
\end{equation*}
realizes the conclusion $\forall x \in P\setminus\{0\} \, B(x)$.
\end{proof}


\section{Stream representations of real numbers}
\label{sec-realnumbers}
As a first serious application of $\IFP$ we present a case study about 
the specification and extraction of exact representations of real numbers.
This will highlight many features of our system 
such as the use of classical axioms as well as 
partial and infinite realizers.
We will continue the development of the system $\IFP$ in Sect.~\ref{sec-opsem}
with the operational semantics of programs. 

We study three representations of real numbers as 
infinite streams of discrete data:
Cauchy representation, 
signed digit representation,
and 
infinite Gray code~\cite{Gianantonio99,Tsuiki02}.
We first recall each representation 
informally in the style of computable analysis~\cite{Weihrauch00}.
Then we show how it can be obtained as the realizability interpretation of a 
suitable predicate built on the formalization 
of real numbers in $\IFP$ in Sect.~\ref{sub-ex-reals}.
Hence, in this section all formal definitions and proofs take place
in $\IFP(\ax_R)$ where
$\ax_R$ is the non computational axiom system for the real numbers introduced
in Sect.~\ref{sub-ex-reals} which includes the Archimedean property ($\AP)$ and
Brouwer's Thesis for nc relations ($\BTnc$).
In particular, in this instance of $\IFP$ the various versions of 
Archimedean Induction~(Sect.~\ref{subsub-ai}) are valid.

For our purpose  it is most convenient to work
in the interval $[-1,1]$.
Everything could be easily transferred to the 
unit interval $[0,1]$ which is used in~\cite{Tsuiki02}.

We will use the notation $a_0:a_1:\ldots$ to denote infinite streams,
mostly in an informal setting but occasionally also for
elements of the domain $D$ that represent streams
(as we did in Sect.~\ref{sub-domain}). 

\subsection{Cauchy representation}
\label{sub-cauchy} 

\paragraph{Informal definition}
An infinite sequence $a = (a_i)_{i\in\NN}$ of rational numbers
that converges quickly to a 
real number $x$ is called a \emph{Cauchy representation} of $x$:
\[ \A(a,x) \eqdef \forall n \in\NN\, |x-a_n|\le 2^{-n}\,. \]
We consider the Cauchy representation as the standard representation
and call any other representation $R$ of real numbers (in $[-1,1])$
\emph{computable} if it is computably equivalent to the
Cauchy representation %
restricted to $[-1,1]$, 
i.e.\ $R$-representations can be effectively transformed
into Cauchy representations and vice-versa. More precisely, 
if $R(r, x)$ expresses that $r$ is a $R$-representation of $x$ 
we say that $R$ is computable
if there
exist (possibly partial) computable functions $\varphi,\psi$ such that
for all $a,r$ and $x\in[-1,1]$
\[ 
R(r,x) \to \A(\varphi(r),x)\qquad\hbox{and}\qquad
\A(a,x) \to R(\psi(a),x)\,.
\]
Note that all representations we will consider 
are functions or infinite sequences of discrete objects (rational numbers or digits)
possibly extended with undefinedness.
There exist natural notions of computable functions between such   
representations (see~\cite{Rogers67}, \cite{Weihrauch00}, \cite{Tsuiki02}).

\paragraph{Formalization in $\IFP$}
The Cauchy representation can be obtained through the realizability 
interpretation of the predicate
\begin{eqnarray*}
\A(x) &\eqdef& \forall n \in\NN\, \exists q \in\QQ\, |x-q|\le 2^{-n}\,.
\end{eqnarray*}
By unfolding the definition of realizability one obtains $\ire{a}{\A(x)} \leftrightarrow$
\begin{eqnarray*}
&&
\forall n\, (a : \ftyp{\nat}{\rat} \land \forall b\ (\ire{b}{\NN(n)} \to
\exists q\, (\ire{(a\, b)}{\QQ(q)} \land
|x-q|\le 2^{-n})))
\end{eqnarray*}
where $\rat\eqdef\tau(\QQ) = \nat\times\nat\times\nat$.
By identifying natural numbers with their realizers, this simplifies to
\begin{eqnarray*}
\ire{a}{\A(x)} &\leftrightarrow& a : \ftyp{\nat}{\rat} \land 
\forall n \in\NN\, \exists q\, (\ire{(a\, n)}{\QQ(q)} \land
|x-q|\le 2^{-n})
\end{eqnarray*}
and by further expressing rational numbers through their realizers, it becomes 
\begin{eqnarray*}
\ire{a}{\A(x)} &\leftrightarrow& a : \ftyp{\nat}{\rat} \land 
\forall n \in\NN\, |x-a\,n|\le 2^{-n}\,.
\end{eqnarray*}

Therefore $\ire{a}{\A(x)} \leftrightarrow \A(a,x)$ where the infinite 
sequence $a$ is given as a function on the natural numbers.

Alternatively, one can formalize the Cauchy representation coinductively by
\begin{eqnarray*}
\A'(x) &\eqnu& \exists n \in\NN\,(|x-n|\le 1 \land \A'(2x)).
\end{eqnarray*}
Defining the type of streams of type $\rho$ as
\[\stream{\rho} \eqdef \tfix{\alpha}{\rho \times \alpha}\]
the predicate $\A'$ 
has the type $\tau(\A') = \stream{\nat}$ 
and we obtain the realizability interpretation
\begin{eqnarray*}
\ire{a}{\A'(x)}&\eqnu&\exists n \, \in \NN,a'\,
         (a = \Pair(n,a') \land |x-n|\le 1
            \land \ire{a'}{\A'(2x)})
\end{eqnarray*}
by identifying natural numbers with their realizers.
Therefore, the two formalizations lead to different `implementations' of the 
Cauchy representation. However, they are equivalent in the sense that 
one can prove $\A(x) \leftrightarrow \A'(x)$ and extract from the proof 
mutually inverse translations between the representations. 
The stream representation has the advantage that it permits `memoized' 
computation due to a lazy operational semantics 
(see~Sect.~\ref{sec-opsem}).

\subsection{Signed digit representation}
\label{sub-sd}
\paragraph{Informal definition}
For an infinite sequence $p=(p_i)_{i<\omega}$ of signed digits $p_i\in\{-1,0,1\}$ 
set
\begin{eqnarray}\label{eq-sd0}
\val{p} \eqdef \sum_{i<\omega} p_i2^{-i} \in [-1,1]\,.
\end{eqnarray}
If $x = \val{p}$, then $p$ is called a 
\emph{signed digit representation} of $x\in[-1,1]$. We set
\[ \C(p,x) \eqdef \val{p} = x\,. \]
The digit $0$ is redundant since every $x\in [-1,1]$ has a 
\emph{binary representation}, that is, a signed digit
representation $p\in\{-1,1\}^\omega$.
However, the redundancy is needed to render the signed
digit representation computable, in particular to be able to compute from a
Cauchy representation of $x$ a signed digit representation of $x$.

One easily sees that for $d \in \{-1, 0, 1\}$, $p \in \{-1, 0, 1\}^\omega$
and $x \in [-1, 1]$
\begin{eqnarray}\label{eq-sd}
    \C(d:p, x) &\leftrightarrow& |2x -d| \leq 1\land \C(p, 2x-d)
\end{eqnarray}
where $d:p$ denotes the sequence 
beginning with $d$ and continuing with $p$.

\paragraph{Formalization in $\IFP$}
We define a predicate $\C(x)$ expressing that $x$ has a signed digit
representation.
First, we define the property of being a signed digit,
\begin{eqnarray*}
  \SD(x) &\eqdef& (x = -1 \lor x = 1) \lor x = 0\,.
\end{eqnarray*}
We define $\tri \eqdef (\one + \one) + \one$. Then, 
$\tau(\SD) = \tri$ and
\begin{eqnarray*}
\ire{d}{\SD(x)} 
           & =& (d = \Left(\Left(\Nil)) \land x = -1) \lor \\
             &&       (d = \Left(\Right(\Nil)) \land x = 1) \lor \\
             & & (d = \Right(\Nil) \land x = 0)\,.
\end{eqnarray*}
Thus, the three digits $-1, 1, 0$ are realized by the three elements
$\Left(\Left(\Nil)), \linebreak   \Left(\Right(\Nil)), \Right(\Nil)$ of $\tri$.
We identify these natural numbers and their realizers and use variables $d, e$ for both of them.

Next we define a predicate expressing that $d \in \{-1,0,1\}$ is the first digit 
of a signed digit representation of $x$
\begin{eqnarray*}
  \II(d,x)  &\eqdef& |2x-d| \leq 1\,.
\end{eqnarray*}
Finally, in view of (\ref{eq-sd}), we set
\begin{eqnarray*}
\C(x) &\eqnu& \exists d \in \SD\, (\II(d, x) \wedge \C(2x-d))\,.
\end{eqnarray*}
We have
$\tau(\C) = \stream{\tri}$ 
and
\begin{eqnarray*}
  \ire{p}{\C(x)} \eqnu \exists d \, \in \SD, p'\ (p = \Pair(d, p')\ \land \II(d,x) \land 
  \ire{p'}{\C(2x-d)})\,.
\end{eqnarray*}
Because of (\ref{eq-sd}) one easily sees that
$\ire{p}{\C(x)}$ holds iff $p$ is an infinite stream
of signed digits that represents $x$, i.e.\ $\C(p,x)$ holds.

\subsection{Infinite Gray code}
\label{sub-gray}
\paragraph{Informal definition}
Gray code of a real number $x$ in $[-1, 1]$ is defined using the digits
$\LG$ and $\RG$ and an `undefined' digit, $\bot$. 
We first define \emph{total Gray code} of $x$ which is a variant of
the binary representation and which does not use $\bot$.   
For an infinite sequence $q \in \{\LG,\RG\}^\omega$, 
we say that  $q$ is a total Gray code of $x$ if $x = \valg{q}$ where,
identifying $\LG$ with $-1$ and $\RG$ with $1$,
\begin{eqnarray}\label{eq-g0}
\valg{q} =   \sum_{i< \omega} (-\prod_{j \le i}(-q_j)){2^{-i}} \in [-1,1]. 
\end{eqnarray}
There are simple conversion algorithms between
binary representation and total Gray code.
Comparing (\ref{eq-g0}) with (\ref{eq-sd0}), one can see that
if $(q_i)_{i < \omega}$ is a Gray code, then
$p = (p_i)_{i < \omega}$  is a binary representation
of the same number for $p_i = -\prod_{j \le i}(-q_j)$.
This equation means that $p_i$
is 1 iff $q_0,\ldots,q_{j}$ contains an odd number of $\RG$.
Conversely,  if $p = (p_i)_{i < \omega}$  is a binary representation,
then $(q_i)_{i < \omega}$ for
\[
q_{i} = \left\{
\begin{array}{ll}
\LG &\hbox{if $p_{i-1}=p_i$}\\
\RG &\hbox{if $p_{i-1} \ne p_i$}\\
\end{array}
\right.\\
\]
is a total Gray code of the same number.
Here, we temporarily define $p_{-1} = -1$.
Defining the `tent function'
$\tent : [-1, 1] \to [-1, 1]$ as 
\[\tent(x) = 1-2|x|,\]
one can show 
\begin{eqnarray*}
\valtg{a:q} = x &\leftrightarrow&
((x \leq 0 \land a = \LG) \lor (x \geq 0 \land a = \RG))
\land \valtg{q} = \tent(x)
\end{eqnarray*}
for $a \in \{\LG, \RG\}$, $q \in \{\LG, \RG\}^\omega$ and $x \in [-1, 1]$.
This means that $q$ is an itinerary of $x$ along the tent function,
i.e.\ $q_n$ equals $\LG$ or $\RG$ depending on whether $\tent^n(x)$
is negative or positive. If $\tent^n(x) = 0$, then $q_n$ may be either.

Total Gray code is non-unique for the dyadic rationals
in $(-1,1)$, that is, numbers of the form $k/2^l$ where $l\in\NN$
and $k\in\ZZ$ and with $|k| < 2^l$. Such numbers have 
two binary codes of the form  $t(-1)1^\omega$ and $t1(-1)^\omega$
for some finite sequence $t\in\{-1,1\}^*$,
and therefore have exactly  two total Gray codes, namely,
\[ s\LG\RG\LG^\omega~\footnote{If $s=a_0,\ldots,a_{n-1}$, then
$s\LG\RG\LG^\omega = a_0:\ldots:a_{n-1}:\LG:\RG:\LG:\LG:\LG:\ldots$}
\quad\hbox{and}\quad s\RG\RG\LG^\omega \]
for some finite sequence $s\in\{\LG,\RG\}^*$.
These two codes only differ in the first digit after $s$, so 
it is natural to allow this digit to be $\bot$ since it carries no
information.
Therefore, we define the set of \emph{Gray codes} as
\begin{eqnarray*}
\GC  &\eqdef& \{\LG,\RG\}^\omega \cup 
              \{s\bot\RG\LG^\omega \mid s \in \{\LG,\RG\}^*\}\\
      &=&\{q\in\{\LG,\RG,\bot\}^\omega \mid \forall n\, 
       (q_n=\bot \to (q_k)_{k>n} =\RG\LG^\omega)\}      
\end{eqnarray*}
and define $\valg{\cdot}:\GC\to[-1,1]$
as the extension of 
total Gray code $\valg{\cdot}:\{\LG, \RG\}^\omega\to[-1,1]$ 
by setting
\[\valg{s\bot \RG\LG^\omega} \eqdef 
   \valtg{s\LG\RG\LG^\omega} ( = \valtg{s\RG\RG\LG^\omega})\,.\] 
For example $\valg{\bot\RG\LG^\omega} = 0$ and 
$\valg{\RG\bot\RG\LG^\omega} = 1/2$. 
We set
\[\G(q,x) \eqdef q\in\GC \land \valg{q} = x\,.\]
One can see that
\begin{eqnarray}\label{eq-gray}
\G(a:q,x) &\leftrightarrow&
((x \leq 0 \land a = \LG) \lor (x \geq 0 \land a = \RG)
\lor\, (x = 0 \land a = \bot)) \nonumber\\
&&\land\, \G(q,\tent(x))
\end{eqnarray}
for $a \in \{\bot, \LG, \RG\}$, $q \in \{\bot, \LG, \RG\}^\omega$
and $x \in [-1, 1]$. 
Note that $\tent(x)$ in the right conjunction of (\ref{eq-gray}) does not depend on
the first digit $a$ whereas for the signed digit case 
$2x-d$ in the right conjunction of (\ref{eq-sd}) depends on $d$.

While it can be shown that total Gray code is \emph{not} computable, 
Gray code \emph{is}, thanks to the possibility of having an 
undefined digit. 
In \cite{Tsuiki02} one finds
programs translating between Gray code and the signed digit representation.

\paragraph{Formalization in $\IFP$}
We define a predicate $\G(x)$ expressing that $x$ has a Gray code.
We first define a predicate for the digits of Gray code:
\begin{eqnarray*}
\D(x) &\eqdef&  x\neq 0 \to (x\le 0 \lor x \ge 0)\,.
\end{eqnarray*}
We have $\tau(\D) = \bool$ for $\bool \eqdef \one+\one$.  
Note that
$\tval{\bool}{} = \{\Left(\Nil), \Right(\Nil), \bot,  \linebreak \Left(\bot), \Right(\bot)\}$
(See Remark 1 of Sect.~\ref{sub-realizability}).
Setting $\LG \eqdef \Left(\Nil)$ and $\RG \eqdef \Right(\Nil)$, we have
\begin{eqnarray*}
  \ire{a}{\D(x)}  &=& 
a : \bool\, \land\, 
(x \neq 0 \to
                      (a = \LG \land x \le 0) \lor (a = \RG \land x \ge 0))\,.
\end{eqnarray*}
Thus, all elements of $\bool$ realize $\D(0)$.
By considering not only $\bot$ but also
$\Left(\bot)$ and $\Right(\bot)$ as denotations of the Gray code digit $\bot$,
$\ire{a}{\D(x)}$ means that $a$ is the first digit of a Gray code of $x$.
Therefore, we define
\begin{eqnarray*}
\G(x) &\eqnu& (-1 \leq x \leq 1) \land \D(x) \land \G(\tent(x))\,.
\end{eqnarray*}
We have $\tau(\G) = \stream{\bool}$ 
and
%
\begin{eqnarray*}
  \ire{q}{\G(x)}  &\eqnu& (-1 \leq x \leq 1) \land  \exists a, q'\ (q = \Pair(a,q') \land \ire{a}{\D(x)} \land \ire{q'}{\G(\tent(x)))}
\end{eqnarray*}
and hence 
$\ire{q}{\G(x)}$ means that $q$ is 
a Gray code of $x$, i.e.\ $\G(q,x)$ by (\ref{eq-gray}).

\subsection{Extracting conversion from signed digit representation to Gray code}
\label{sub-sdg}

We show $\C\subseteq\G$ and extract from the proof
a program that converts signed digit representation to Gray code.
Proofs are presented in an informal style but are
formalizable in the system $\IFP(\ax_R)$ and
simplifications of programs are proven in $\RIFP(\emptyset)$.
We write $x\in \II_d$ for $\II(d,x)$ and
allow combinations of patterns in case expressions.  For example, 
\begin{eqnarray*}
&&  \case\, M\, \of\, \{-1 \to N_1; 1 \to N_2; 0 \to N_3; \} \eqdef\\
&&\hspace*{1cm} \case\, M\, \of\, \{\Left(a) \to   (\case\, a\, \of\, \{\Left(b) \to N_1; \Right(b) \to N_2\}); \\
&&\hspace*{2.8cm} \Right(a) \to N_3\}\,.
\end{eqnarray*}
Recall that $\C = \nu(\Phi_\C)$ and $\G = \nu(\Phi_\G)$ for
\begin{align*}
\Phi_\C &\eqdef \lambda X\,\lambda x\,\exists d \in \SD\, 
       (x \in \II_d \wedge X(2x-d))\,,\\
\Phi_\G &\eqdef \lambda X\,\lambda x\,
(-1 \leq x \leq 1) \land \D(x) \land X(\tent(x))\,.
\end{align*}
According to Lemma~\ref{lem-mon-new}, 
the formula $\monprop{}{\Phi_\C}
\eqdef X \subseteq Y \to \Phi_\C(X) \subseteq \Phi_\C(Y)$
 expressing the monotonicity of the operator
$\Phi_\C$  is proved in IFP' and the following program 
$\mon : \ftyp{(\ftyp{\alpha_X}{\alpha_Y})}{\ftyp{\tri \times \alpha_X}{\tri \times \alpha_Y}}$ is
extracted from the proof.
\begin{align}\label{mon}
\mon\, f \,p \eqdef \Pair(\projl\,p, f (\projr\,p)). 
\end{align}
It is also the case for $\monprop{}{\Phi_\G}$ and the same program $\mon$
with the type obtained by replacing $\tri$ with $\bool$ is realizing $\monprop{}{\Phi_\G}$.
\begin{lemma}
  \label{lemma-ng}
  $\forall x\, (\C(-x) \to \C(x))$.
\end{lemma}
\begin{proof}
  By coinduction.  Therefore, we show $P \subseteq \Phi_\C(P)$ for
$P(x) \eqdef \C(-x)$, that is, 
\begin{equation}
  \forall x\ (\C(-x) \to \exists d \in \SD\, (x \in \II_d \wedge \C(-(2x-d)))). \label{eq-nh}
\end{equation}
Suppose that $\C(-x)$ holds.   By coclosure, 
for some $e \in \SD$, we have
$-x \in \II_{e} \wedge \C(-2x-{e})$.
Since $-x \in \II_{e}$, we have $x \in \II_{-e}$.
Since $\C(-2x-{e})$, we have $\C(-(2x-d))$ for $d = -e$, and
therefore  
$x \in \II_d \wedge \C(-(2x-d))$.
\end{proof}

The program 
$\stepone : \ftyp{\stream{\tri}}{\tri \times \stream{\tri}}$ 
extracted from the proof of (\ref{eq-nh}) is
\[
\stepone \eqdef \lambda p.\, \Pair(\case\, (\projl\, p)\, \of\, \{-1 \to 1; 0 \to 0; 1 \to -1\}, \projr\, p)\,.
\]
Therefore,  by (\ref{eq:coind}) of Sect.~\ref{sub-pe},
the realizer extracted from the proof of $P \subseteq \C$ is
the following program  $\minus : \ftyp{\stream{\tri}}{\stream{\tri}} $ 
\[
\minus \eqrec (\mon\  \minus) \circ \stepone.   
\]
After some simplification using Lemma~\ref{lem-bot} we have 
\begin{equation}
\label{eq-minus}
\minus \, p \eqrec \Pair(\case\, (\projl\, p)\, \of\, \{-1 \to 1; 0 \to 0; 1 \to -1\}, \minus \ (\projr\, p)).
\end{equation}

\begin{theorem}
\label{thm-c-g'}
$\C\subseteq\G$.
\end{theorem}
\begin{proof}
By coinduction.  Hence we show 
$\forall x (\C(x) \to (-1 \leq x \leq 1) \land \D(x) \land \C(\tent(x)))$.
Since $\forall x (\C(x) \to -1 \leq x \leq 1)$ is immediate, 
we need to show the following two claims.

\noindent
\emph{Claim 1}.  $\forall x\,(\C(x) \to \D(x))$, that is,
$\forall x \in\C\setminus\{0\}\  \B(x)$ where
$\B(x) \eqdef x \le 0 \lor x \ge 0$.
We use $\AIB_{1/2}(\C, \B)$. Therefore, we show
\begin{equation}
  \forall x \in \C \setminus \{0\}\  (\B(x) \lor (|x| \leq 1/2 \land \C(2x) \land (\B(2x) \to \B(x)))).   
\label{eq-cd1}
\end{equation}
Since $\C(x) \eqnu \exists d \in \SD\, (x \in \II_d \wedge \C(2x-d))$,
we have the following cases.

Case $d = -1$.  We have $-1\leq x \leq 0 \wedge \C(2x+1)$ and thus $x \leq 0$.

Case $d = 1$.  We have $0\leq x \leq 1 \wedge \C(2x-1)$ and thus $x \geq 0$.

Case $d = 0$. We have $|x| \leq 1/2 \wedge \C(2x)$. 
In addition, we always have $\B(2x) \to \B(x)$ (realized by $\identity$).
This completes the proof of \emph{Claim 1}.

\medbreak
\noindent
\emph{Claim 2}. $\forall x(\C(x)\to\C(\tent(x)))$.

We set 
$\C'(y) \eqdef \exists x \in \C\  y = \tent(x)$ and show $\C' \subseteq \C$ by 
half-strong coinduction.  Therefore, we show
\begin{equation}
\C'(y) \to  \exists d \in \SD (y \in \II_d \wedge \C'(2y-d)) \vee \C(y).  \label{eq-ct0}
\end{equation}
Assume $\C'(y)$, i.e., $y = \tent(x)$ for an $x$ that satisfies $\C(x)$.

Case $-1 \leq x \leq 0 \wedge \C(2x+1)$. Then $2x+1 = \tent(x) = y$. Hence we have $\C(y)$.

Case $0 \leq x  \leq 1 \wedge \C(2x-1)$. Then $2x-1 = -\tent(x) = -y$ and hence $\C(-y)$. 
Therefore, by Lemma~\ref{lemma-ng}, we have $\C(y)$.

Case $|x| \leq 1/2 \wedge \C(2x)$. Then $y = \tent(x) \geq 0$ and thus $y \in \II_{1}$.
Hence it suffices to show $\C'(2y-1)$. 
We have $2y-1 = 1-4|x| = \tent(2x)$ 
and therefore $\C'(2y-1)$ holds.
This completes the proof of \emph{Claim 2} and hence the proof of
the theorem.
\end{proof}

We extract a 
program from this proof.

\emph{Program from Claim 1}.
The program
$\steptwo : \ftyp{\stream{\tri}}{(\bool + \stream{\tri} \times (\ftyp{\bool}{\bool}))}$
extracted from the proof of (\ref{eq-cd1}) is  
\begin{alignat*}{2}
\steptwo\, p \eqdef \case\, (\projl\, p)\, \of\,
 \{&-1 \to \Left(\Left\, \Nil); \\&1 \to \Left(\Right\, \Nil);\\
&                0 \to \Right(\Pair(\projr\ p, \identity ))\}.
\end{alignat*}
Therefore, by (\ref{eq:realaib}) of Lemma \ref{lem-rea-ai}, 
the extracted realizer of $\C(x) \to \D(x)$ is 
$\sgh : \ftyp{\stream{\tri}}{\bool}$,
\begin{eqnarray*}
\sgh \,p &\eqrec&
\case\,(\steptwo\,p)\,\of\,\{\Left(b) \to b; \Right(q,g) \to g(\sgh \,q)\}\,.
\end{eqnarray*}
By rewriting a nested case expression using Lemma~\ref{lem-bot},
we have
\begin{align}\label{eq-sgh}
\sgh \, p  \eqrec
\case\, (\projl\, p)\, \of\, \{-1 \to \LG; 1 \to \RG; 0 \to \sgh  (\projr\, p)\}\,.
\end{align}
Note that $\sgh(0:0:\ldots)=\bot$. This can be seen by applying Scott induction 
(Axiom~(viii))
to the predicate 
$P \eqdef \lambda b\,\,(b\,(0:0:\ldots)=\bot)$ and 
$a \eqdef\lambda b.\,\lambda p.\,\case\,(\projl\,p)\,\of\,\{-1\to\LG;1\to\RG;0\to b(\projr\,p)\}$.

\emph{Program from Claim 2}.
The program extracted from the proof of (\ref{eq-ct0}) is
$\stepthree : \ftyp{\stream{\tri}}{\tri \times \stream{\tri}+\stream{\tri}}$,
\begin{eqnarray*}
\stepthree \ p &\eqdef& \case\, (\projl\, p)\, \of\,
                        \{-1 \to \Right(\projr\, p); \\
&&\hspace*{2.7cm}   1 \to \Right (\minus  (\projr \, p));\\
&&\hspace*{2.7cm} 0 \to \Left\, (\Pair(1, \projr\, p))\}\,.
\end{eqnarray*}
Therefore, according to equation~(\ref{eq:hscoi}) of 
Sect.~\ref{sub-pe}, 
the program extracted 
from the proof of
$\C(x) \to \C(t(x))$ is $\sgt : \ftyp{\stream{\tri}}{\stream{\tri}}$,
\[
 \sgt \ p \eqrec [(\mon\, \sgt ) + \idty]  (\stepthree\  p)\,.
\]
This definition can be simplified to (using again Lemma~\ref{lem-bot}),
\begin{align}\label{eq-sgt}
\begin{split}
\sgt \, p \eqrec \case (\projl\, p) \of \{&-1 \to \projr\, p; 1 \to  \minus  (\projr\, p);\\
&\ 0 \to \Pair(1,\sgt (\projr\, p))\}\,.
\end{split}
\end{align}
Now, by equation~(\ref{eq:coind}) of Sect.~\ref{sub-pe},
the extracted program $\stog:\ftyp{\stream{\tri}}{\stream{\bool}}$ 
from the proof of $\C \subseteq \G$
is $\stog  \eqrec (\mon\, \stog ) \circ \stepfour$ with
$\stepfour : \ftyp{\stream{\tri}}{\stream{\bool}\times\stream{\tri}}$,  
$\stepfour\ p = \Pair(\sgh \ p, \sgt \ p) $.
This simplifies to
\begin{align}\label{eq-stog}
\stog \, p \eqrec \Pair(\sgh \, p, \stog (\sgt \, p)).
\end{align}
Note that since $\sgh(0:0:\ldots)=\bot$ the first digit of 
$\stog(0:0:\ldots)$ is $\bot$ and therefore 
$\stog(0:0:\ldots)$ evaluates to $\bot : \RG : \LG : \LG : \ldots$.
We will study this evaluation in Example \ref{ex3}, at the end of this paper.

Thus, we have obtained a program that consists of four recursions.
In the rest of this section, 
we transform this program into a program with one recursion.
We use the list notation $a:p$ for 
$\Pair(a, p)$ and write $\head$ for $\projl$ and $\tail$ for $\projr$.

First, by Scott-induction it is easy to see the equivalence of  
(\ref{eq-stog}) to
the following program provided $p$ is restricted to 
total elements of $\stream{\tri}$,
that is, elements of $\stream{\tri}_t$ 
where 
$\stream{\tri}_t(a) \eqnu \head\,a\in\{-1,0,1\} \land \stream{\tri}_t(\tail\,a)$.
\begin{itemize}
\item[] $\stog\, p\ \eqrec\ \case\ (\head\, p)\ \of\ \{$\\
\begin{tabular}{rcl}
  $-1$ & $\to$ & $\LG : \stog\,(\tail\, p)\,;$\\
   $1$ & $\to$ & $\RG : \stog(\minus\,(\tail\, p))\,;$\\
   $0$ & $\to$ & $\sgh\, (\tail\, p) : \stog\,(1 : \sgt\, (\tail\, p))$
\end{tabular}\\
\hbox{}\quad $\}$
\end{itemize}
Note that the two programs are not equal for $p = \bot$ since
  $\stog \, \bot$ is equal to $\bot : \stog \, \bot$
with the old definition (\ref{eq-stog}) of $\stog$, whereas
$\stog\, \bot = \bot$ with the new definition of $\stog$.
However, since all realizers of $\C$ are total (easy proof by coinduction),
both programs realize ${\C \subseteq \G}$.
Therefore we use the same name  $\stog$ for both.

We now show that the new definition of $\stog$ can be simplified.

By strong coinduction (Sect.~\ref{sec-ifp}) one can easily prove 
$G(-x)\to G(x)$. The extracted program 
$\nh: \ftyp{\stream{\bool}}{\stream{\bool}}$ 
inverts the first digit of a Gray code.
\begin{eqnarray*}
\inv \, a  &=& \case\ a\ \of\ \{\,\LG \to \RG\,;\, \RG \to \LG\, \}\\
\nh\, q     &=& (\inv\, (\head\, q)) : (\tail\, q)
\end{eqnarray*}

One can also show, using Scott-induction, 
that
$\sgh(\minus\, p) = \inv\,(\sgh\, p)$
and
$\sgt(\minus\, p) = \sgt\, p$.
Therefore, 
for total $p$,
\begin{eqnarray*}
\stog \,(\minus\,p)  &=& \sgh\,(\minus\, p) : \stog\,(\sgt\,(\minus\, p))\\ 
                     &=& \inv\,(\sgh\, p) : \stog\,(\sgt\, p)\\
                     &=& \nh\,(\stog\, p)\,.
\end{eqnarray*}
With this equation, we can simplify $\stog$ as follows.
\begin{itemize}
\item[] $\stog\, p\ =\ \case\ (\head\, p)\ \of\ \{$\\
\begin{tabular}{rcl}
  $-1$ & $\to$ & $\LG : \stog\,(\tail\, p)\,;$\\
   $1$ & $\to$ & $\RG : \nh\,(\stog\,(\tail\, p))\,;$\\
   $0$ & $\to$ & $\sgh\, (\tail\, p) : \stog\,(1 : \sgt\, (\tail\, p))$
\end{tabular}\\
\hbox{}\quad $\}$
\end{itemize}
The last case further simplifies to 
$0 \to \sgh\, (\tail\, p) : \RG : \nh\, (\stog\, (\sgt\, (\tail\, p)))$
by expanding $\stog$.
Since $\stog\, p = \sgh\, p : \stog\, (\sgt\, p)$, 
one can further rewrite the definition of $\stog$
using the let notation $\lett\ q = M\ \inn\ N$
for $(\lambda q.\,N)\,M$.
\begin{itemize}
\item[] $\stog\, p\ =\ \case\ (\head\, p)\ \of\ \{$\\
\begin{tabular}{rcl}
  $-1$ & $\to$ & $\LG : \stog\,(\tail\, p)\,;$\\
   $1$ & $\to$ & $\RG : \nh\,(\stog\,(\tail\, p))\,;$\\
   $0$ & $\to$ & $\lett\ q = \stog\, (\tail\, p)\ \inn\  
                 (\head\, q) : \RG : \nh\, (\tail\, q)$
\end{tabular}\\
\hbox{}\quad $\}$
\end{itemize}
The above equation holds for total $p$. 
Viewing it as recursive definition (replacing `$=$' by '$\eqrec$')
on obtains a program which coincides with the previous one on total 
arguments (proof by Scott-induction) and hence realizes $\C \subseteq\G$.
It is precisely the Haskell program of signed digit to Gray code conversion
in \cite{Tsuiki02} 
if we view $:, \head$ and  $\tail$  as ordinary list operations.


\section{Operational semantics}
\label{sec-opsem}
The Soundness Theorem (Thm.~\ref{thm-soundness}) shows that  
from an $\IFP$-proof of a formula one can extract 
a program realizing it, provably in $\RIFP$.
Because the program axioms
of $\RIFP$ are correct w.r.t.~the domain-theoretic 
semantics, this theorem shows that the denotational 
semantics of a program extracted from an $\IFP$ proof is a correct realizer 
of the formula.
However, so far we have no means to \emph{run} the extracted programs  
in order to \emph{compute} data that realize the formula.
In this section we address this issue by defining an operational semantics 
and showing that it fits the denotational semantics
through two Computational Adequacy Theorems 
(Thms.~\ref{thm-adequacy},~\ref{thm-adequacytwo}).
The first is essentially an untyped version of Plotkin's Adequacy Theorem 
for the simply typed language PCF~\cite{Plotkin77}. Its proof uses compact elements 
of the untyped domain model as a replacement for types, a technique introduced by 
Coquand and Spiwack~\cite{CoquandSpiwack06},
and follows roughly the lines of~\cite{Berger10}.
The second Adequacy Theorem
concerns the computation of infinite data. 
A related result for an extension of PCF by real numbers was obtained by
Escardo~\cite{Escardo96}. While Escardo works in a typed setting and concerns
incremental computation on the interval domain,
our result is untyped and computes arbitrary infinite data built
from constructors.
There exists a rich literature on computational adequacy covering, for example, 
typed lambda calculi with various effects \cite{PlotkinPower01,Lairdetal13}, 
denotational semantics based on games or 
categories~\cite{CrolePitts92,Simpson04}, and axiomatic 
approaches~\cite{CamposLevy18,FiorePlotkin94}.

{{In the following we work with} our untyped programming language 
that {includes} programs not typable with our type system, 
and consider types only in Section \ref{sub-data}.}
{This shows that the operational properties of our programs are
independent of the type system.}

\subsection{Inductive and coinductive definitions of data}
\label{sub-ind-coind-data}
First we make precise what we mean by data.
Recall from Sect.~\ref{sub-domain} that programs are interpreted in 
the domain $D$ defined by the recursive domain equation
\[
D = (\Nil + \Left(D) + \Right(D) + \Pair(D\times D) + \Fun(D\to D))_\bot \,.
\]
We consider the sub-domain $E$ of $D$ built from constructors only 
\[
E = (\Nil + \Left(E) + \Right(E) + \Pair(E\times E))_\bot
\]
and call its elements \emph{data}. 
We also define various predicates on $D$ as least or greatest 
fixed points of the following operators $\otdata$  and $\odata$  of arity $(\delta)$.
The definitions and proofs below take place in informal mathematics
although we take advantage of the notations and proof rules 
provided by the formal system $\IFP$ regarding inductive and 
coinductive definitions. 
\begin{align*} 
\otdata(X)(a) \eqdef  & 
\bigvee_{\stackrel{\hbox{$C$}}{\hbox{constructor}}} 
\left(
\exists a_1,\ldots,a_k \ \ a = C(a_1,\ldots,a_k) \land \bigwedge_{i\leq k} X(a_i)  
\right)
\end{align*}
and its variant $\odata$ obtained by adding $\bot$ as an option
\begin{align*}
\odata(X)(a) \eqdef  &\  a = \bot \lor \otdata(X)(a)\,.
\end{align*}
We have 
\begin{align*}
\data = &\ \nu(\odata) &\qquad(\hbox{arbitrary data})
\end{align*}
and we define 
\begin{align*}
\fdata \eqdef &\ \mu(\odata) &\qquad(\hbox{finite data})\\
\tdata \eqdef &\ \nu(\otdata) &\qquad(\hbox{total data})\\
\ftdata \eqdef &\ \mu(\otdata) &\qquad(\hbox{finite total data})
\end{align*}
It is easy to see that $\fdata$ consists of 
the compact data, $\tdata$ of the data containing no $\bot$, 
and $\ftdata = \fdata \cap \tdata$, hence our choice of names.

Using binary versions of the operators $\otdata$ and $\odata$,
\begin{align*}
\oteq(X)(a, b) \eqdef & \bigvee_C\left(
\begin{array}{l} \exists a_1,\ldots,a_k,b_1,\ldots,b_k\ a = C(a_1,\ldots,a_k) \land  \\\ \ \ \ \hspace*{1.4cm}
  b = C(b_1,\ldots,b_k) \land \bigwedge_{i \leq k} X(a_i, b_i)
\end{array}
\right)\\
\oless(X)(a,b) \eqdef &\  a = \bot \lor \oteq(X)(a,b)\\
\oeq(X)(a,b) \eqdef &\ a = b = \bot \lor \oteq(X)(a,b) 
\end{align*}
we define the relations
\begin{align*}
a \dlee b \eqdef &\ \nu(\oless)(a,b) &\qquad(\hbox{domain ordering on }E)\\
\appr{a}{b} \eqdef &\ \mu(\oless)(a,b) &\qquad(\hbox{finite approximation})\\
\peq{a}{b} \eqdef &\ \nu(\oeq)(a,b) &\qquad(\hbox{bisimilarity})\\
\teq{a}{b} \eqdef &\ \nu(\oteq)(a,b) &\qquad(\hbox{total bisimilarity})
\end{align*}
Note that $\dlee$ coincides with the domain order $\dle$ 
on $E$ but not with that on $D$.  
$a \dlee b$ implies $a \in \data$, by coinduction,
 therefore $\dlee$ is not reflexive on $D \setminus E$.
Clearly, $\appr{a}{b}$ holds iff $a \dlee b$ and $\fdata(a)$,
and $\teq{a}{b}$ holds iff $\peq{a}{b}$ and $a,b\in\tdata$. 
If we replace in the definition of $\peq{a}{b}$
the largest fixed point $\nu$ by the least fixed point $\mu$, 
we obtain the relation $\mu(\oeq)(a,b)$ which clearly implies that
$a$ and $b$ are equal 
elements of $\fdata$
(easy inductive argument).
However, 
\begin{equation}
\label{eq-algebraic}
\forall a, b\,(\peq{a}{b} \to a=b)
\end{equation}
is a non-trivial assertion expressing that
the elements of $E$ are completely determined 
by their constructors, which we use in this section.
From~(\ref{eq-algebraic}) one can derive the equivalence 
$(a = b 
\, \land\, a, b \in \data)
\leftrightarrow (a \dlee b \land b \dlee a)$ and 
the maximality of the elements in $\tdata$, 
$(a \dlee b \land \tdata(a)) \to a = b$. 
We prove the following lemma to give typical examples of 
inductive and coinductive proofs on data.

\begin{lemma}\label{lemma:fdata}\quad
  \begin{itemize}
  \item[(a)] $\appr{a}{b}$ iff $\fdata(a) \land a \dlee b$.
  \item[(b)] $a \dlee  b$ iff 
             $\data(a) \land \forall d (\appr{d}{a}  \to \appr{d}{b})$.
  \end{itemize}
  \begin{proof}
    (a)  Left to right is by induction on $\appr{a}{b}$. 
Right to left is induction on $\fdata(a)$ to prove that
    $\fdata(a) \to \forall b (a \dlee b \to \appr{a}{b})$.
We show $\forall a\, (\Phi_\bot(P)(a) \to P(a))$ for 
$P(a) = \forall b (a \dlee b \to \appr{a}{b})$.
     Suppose that $\Phi_\bot(P)(a)$. If $a = \bot$, then $P(\bot)$.
     If $a = C(a_1,\ldots,a_k) \land \bigwedge_{i\leq k} P(a_i)$ and $a \dlee b$,
     let $b = C(b_1,\ldots,b_k)$. We have $a_i \dlee b_i$ and thus 
$\appr{a_i}{b_i}$ by $P(a_i)$.

(b) Left to right is immediate 
by (a). 
Right to left is by coinduction on $a \dlee b$.
 Let $P(a, b) \eqdef \data(a) \land\forall d (\appr{d}{a} \to \appr{d}{b})$.  
    We need to show  $\forall a, b\,(P(a, b) \to \Phi^2_\bot(P)(a, b))$.
 Because $\data(a)$, $a = \bot$ or $a$ has the form 
$C(a_1,\ldots,a_k)$  with $a_1,\ldots,a_k \in \data$.
    If $a = \bot$, then $\Phi^2_\bot(P)(a, b)$ holds.  
If $a$ has the form $C(a_1,\ldots,a_k)$,  we have 
$\appr{C(\bot^k)}{a}$ and thus $\appr{C(\bot^k)}{b}$ 
by $P(a, b)$. Therefore,     $b = C(b_1,\ldots,b_k)$ for some $b_i$.  
We need to show that $P(a_i, b_i)$.  If $\appr{d}{a_i}$,
    then $\appr{C(\bot^i,d, \bot^{k-i-1})}{ a}$.  Hence,  
$\appr{C(\bot^i,d, \bot^{k-i-1})}{ b}$, and thus $\appr{d}{ b_i}$. 

  \end{proof}
\end{lemma}

\subsection{Inductively and coinductively defined  reduction relations}
\label{subsec-ope}

We define four reduction relations between closed programs and data through
induction and coinduction.  
These 
relations are related to computational procedures
in Sect.~\ref{section:computationinfinite}.
In order to treat programs as syntactic objects, we introduce a new sort $\pi$ of 
programs and use $M,N,K,\ldots$ for variables of sort $\pi$. 
When a program is considered as an element of $\pi$, 
we use $x, y,\ldots$ as names for program variables while we use $a, b,\ldots$ 
to denote elements of $D$.

A \emph{value} is a closed program $M$ that
begins with a constructor or has the form $\lambda x.\,M$.
Following \cite{Berger10}, we first  
define inductively a \emph{bigstep reduction relation} $M \bs V$ between 
closed programs $M$ and values $V$ as follows:
\begin{itemize}
\item[(i)] $V \bs V$ \\[-0.1em]
\item[(ii)] \AxiomC{$M \bs C(\vec{M})$\ \ \ $N[\vec{M}/\vec{y}]\bs V$}   
             \UnaryInfC{$\case\,M\,\of\, \{\ldots;C(\vec{y})\to N;\ldots\}\bs V$}
            \DisplayProof \\[0.5em]
\item[(iii)] \AxiomC{$M\bs \lambda x.\,M'$}
             \AxiomC{$M'[N/x]\bs V$}
            \BinaryInfC{$M\,N\bs V$}
            \DisplayProof \\[0.5em] 
\item[(iv)]  \AxiomC{$M\,(\rec\,M)\bs V$}
            \UnaryInfC{$\rec\,M\bs V$}
            \DisplayProof 
\end{itemize}
\begin{lemma}\label{unique}
  For a closed program $M$, there is at most one value $V$ such that $M \bs V$.
\end{lemma}
\begin{proof}
 There is at most one $\bs$ reduction rule applicable to a closed program.
\end{proof}

Since bigstep reduction stops at constructors (due to rule (i)),
in order to obtain a data, we need to continue computation 
under constructors.   
We define four reduction relations
$M \muprint a$, $M \mubprint a$, $M \nuprint a$, $M \nubprint a$, all of 
arity $(\pi, \delta)$,
as least and greatest fixed points of 
the operators 
\begin{align*}
&\ocl(X)(M, a) \eqdef
\ \bigvee_C \left(\begin{array}{l}
\exists M_1,\ldots,M_k, a_1, \ldots,a_k\ (M \bs C(M_1,\ldots,M_k)\\  
\hspace*{1.8cm} \land\, 
   a = C(a_1,\ldots, a_k) \land \bigwedge_{i \leq k} X(M_i, a_i))
                      \end{array}
  \right)\\
&\ocl_\bot(X)(M, a) \eqdef a = \bot \lor \ocl(X)(M, a)\,. 
\end{align*}
Here again, $C$ ranges over constructors.   Now we define
\begin{align*}
\muprint\ \eqdef\ & \mu(\ocl)\\
\nuprint\ \eqdef\ & \nu(\ocl)\\
\mubprint\ \eqdef\ & \mu(\ocl_\bot)\\
\nubprint\ \eqdef\ & \nu(\ocl_\bot)\,.
\end{align*}

Note that the definition of $M \muprint a$ is equivalent to an inductive
definition by the following 
reduction rules.
\begin{center}
\AxiomC{$M \bs \Nil$}
             \UnaryInfC{$M \muprint \Nil$}
            \DisplayProof \ \ \ \ \ 
\AxiomC{$M \bs \Pair(M_1,M_2)$\ \ $M_1 \muprint a_1$\ \  $M_2 \muprint a_2$}
             \UnaryInfC{$M \muprint \Pair(a_1, a_2)$}
            \DisplayProof \\[0.5em]

\AxiomC{$M \bs \Left(M)$\ \ $M \muprint a$}
             \UnaryInfC{$M \muprint \Left(a)$}
            \DisplayProof \ \ \ \ \ 
\AxiomC{$M \bs \Right(M)$\ \ $M \muprint a$}
             \UnaryInfC{$M \muprint \Right(a)$}
            \DisplayProof 
          \end{center}
$\nuprint$ can be defined by replacing in the rules above
$\mu$ with $\nu$ and interpreting the rules  coinductively,
that is, permitting infinite derivations. 
$\mubprint$ and $\nubprint$ are obtained by adding the axioms
$M \mubprint \bot$ and $M\nubprint \bot$ respectively.

$M \muprint a$ is the finite reduction to a finite total data and  
$M \nuprint a$ is the (possibly) infinite reduction to a (possibly) infinite total data.
$M \mubprint a$ and
$M \nubprint a$ are reductions that may leave some part unreduced by assigning $\bot$,
and are used to obtain observations of infinite data through finite approximations.
  For example, for
  $M = \rec(\lambda x. \Pair(\Nil, x))$,
  no $a \in D$ satisfies $M \muprint a$ but
\begin{align*}
M &\nuprint\, \Nil:\Nil:\Nil:\ldots\\
M &\nubprint  \bot:\Nil:\Nil:\ldots\\
M &\mubprint  \bot:\Nil:\bot \qquad (=  \Pair(\Pair(\bot,\Nil),\bot)).
\end{align*}

\begin{lemma}\label{lemma:nutomu}\quad
  \begin{itemize}
    \item[(a)] $M \muprint a $ iff $M \mubprint a \land \ftdata(a)$.
    \item[(b)] $M \nuprint a $ iff $M \nubprint a \land \tdata(a)$.
    \item[(c)] $M \mubprint a$ iff $M \nubprint a \land \fdata(a)$.
\item[(d)]  $M \nubprint a$ iff $\forall d\, (\appr{d}{a} \to M \mubprint d) 
\land \data(a)$.
    \end{itemize}
  \end{lemma}
  \begin{proof}
    (a) By induction on $\muprint$ and $\mubprint$.

    (b) By coinduction on $\nuprint$ and $\nubprint$.
    
    (c) Left to right is immediate induction on $\mubprint$.
    Right to left is by induction on $\fdata(a)$.  

    (d)  Right to left by coinduction on $\nubprint$.  For $P(M, a) \eqdef \forall d\, (\appr{d}{a} \to M \mubprint d)
\, \land\, E(a)$,
we prove $P(M, a) \to \ocl_\bot(P)(M, a)$.  Suppose that $P(M, a)$.
Since $a \in \data$, $a = \bot$ or $a$ has the form $C(a_1,\ldots,a_k)$ for $a_i \in E$.
    If $a = \bot$, then we have $\ocl_\bot(P)(M, a)$.  If
    $a = C(a_1,\ldots,a_k)$,  then $\appr{C(\bot^k)}{a}$ and therefore
    $M \mubprint C(\bot^k)$.  Hence,
    $M \bs C(M_1,\ldots,M_k)$
    for some $M_1,\ldots,M_k$.
    We need to show $P(M_i, a_i)$ for each $i \leq k$.  
    Suppose that $\appr{d'}{a_i}$ and let 
    $d = C(\bot^{i-1}, d', \bot^{k-i})$.
    Since $\appr{d}{ a}$, we have $M \mubprint d$.  Therefore, 
    $M_i \mubprint d'$.

  Left to right:  Suppose $M \nubprint a$.  We have $\data(a)$ by coinduction on $\data$.
We show that $\appr{d}{ a}$ implies $M \mubprint d$.
$\appr{d}{ a}$ implies $\fdata(d)$ by Lemma~\ref{lemma:fdata}~(a).
  On the other hand, $M \nubprint a$ and  $\appr{d}{ a}$ imply
  $M \nubprint d$  by coinduction.  
  Therefore, by part~(c), 
  $M \mubprint d$.

\end{proof}
    
\medbreak

\subsection{Computational Adequacy Theorem}
\label{sub-adequacy}
Now we prove our first result linking the denotational with the operational
semantics.
\begin{theorem}[Computational Adequacy I] %
\label{thm-adequacy}
Let $M$ be a closed program.
\begin{enumerate}
\item[(a)] $M \muprint a$ iff $a = \val{M} \land \ftdata(a)$. 
\item[(b)] $M \mubprint a$ iff $a \dlee \val{M} \land \fdata(a)$.
\item[(c)] $M \nuprint a$ iff $a = \val{M} \land \tdata(a)$.
\item[(d)]          $M \nubprint a$ iff $a \dlee \val{M}$. 
\end{enumerate}
\end{theorem}
Note in (d)  that  $a \dlee \val{M}$ implies $\data(a)$.
The proof of the theorem will be given through the following 
Lemmas~\ref{lemma:totalpartial}-\ref{lem-approximation2}.
Computational adequacy usually means (a), and (c) is its generalization to infinite total data.  
As we will see in Lemma~\ref{lemma:totalpartial}, (b) and (d) are proved as lemmas for (a) and (c).
They are also foundations for the 
second Adequacy Theorem (Thm.~\ref{thm-adequacytwo}).
\begin{lemma}
\label{lemma:totalpartial}
In Thm.~\ref{thm-adequacy},
part (b) implies part (a),
and
part (d) implies part (c).
  \end{lemma}
  \begin{proof}
    {[(b) implies (a)]}: 
$M \muprint a$ implies $\ftdata(a)$ by Lemma~\ref{lemma:nutomu}~(a).
In addition, if $\ftdata(a)$ holds, then  $M \muprint a$ and $M \mubprint a$ 
are equivalent by Lemma~\ref{lemma:nutomu}~(a), and
    $a \dlee b$ and $a = b$ are equivalent as we mentioned before Lemma \ref{lemma:fdata}.

    [(d) implies (c)]:  Similar.  
Note that, by (\ref{eq-algebraic}), $=$ on $E$ is the bisimulation relation.
  \end{proof}

Due to this lemma, we only need to prove (b) and (d).  The `only if' 
parts of (b) and (d) are obtained by the following lemma.

\begin{lemma}[Correctness]\quad  
\label{lem-correctness}
\begin{itemize}
\item[(a)] If $M \bigstep V$, then $\val{M}=\val{V}$.
\item[(b)] If $M \mubprint a$, then $\appr{a}{\val{M}}$.
  
\item[(c)] If $M \nubprint a$, then $a \dlee \val{M}$.
\end{itemize}
\end{lemma}
\begin{proof}
(a) is proven by induction along the definition of $M\bigstep V$.

(b) 
We define $P(M, a) \eqdef \appr{a}{ \val{M}}$ and prove
$M \mubprint a \to P(M, a)$ by induction.
Therefore, we prove $\ocl_\bot(P)(M, a) \to P(M, a)$.
Suppose that $\ocl_\bot(P)(M, a)$.  If $a = \bot$, then we have $P(M, a)$.
If $\ocl(P)(M, a)$, then $M \bs C(M_1,\ldots,M_k) $, $a = C(a_1,\ldots,a_k)$, and $P(M_i , a_i)$
for every $i \leq k$.
Hence, by (a), 
$\val{M} = \val{C(M_1,\ldots,M_k)} =
C(\val{M_1}, \ldots, \val{M_k})$.
Since $P(M_i, a_i)$, we have $\appr{a_i}{ \val{M_i}}$
and therefore 
$\appr{a}{ \val{M})}$.

(c) By Lemma~\ref{lemma:fdata}~(b), we need to show that
$M \nubprint a$ and $\appr{d}{ a}$ implies $\appr{d}{ \val{M}}$.
First,  we can easily show that $M \nubprint a$ and $\appr{d}{a}$ implies 
$M \nubprint d$.  Since $M \nubprint d$ and $\fdata(d)$, 
we have $M \mubprint d$ by Lemma~\ref{lemma:nutomu}~(a).  Therefore, 
$\appr{d}{ \val{M}}$ by (b).
\end{proof}

We prove the `if' part of Thm.~\ref{thm-adequacy}~(b) following 
\cite{Berger10}, which uses ideas from~\cite{Plotkin77} 
and~\cite{CoquandSpiwack06}.
Let $D_0$ be the set of compact elements of $D$. 
To every $a\in D_0$ we assign a set of closed programs
$\cl{a}$  by induction on 
$\rk(a)$ (Sect.~\ref {sub-domain}).   
\begin{eqnarray*}
\cl{\bot} &=& \hbox{the set of all closed programs}\\
\cl{C(a_1,\ldots,a_k)} &=& 
 \{M \mid \exists M_1,\ldots,M_k,\,M\bigstep C(M_1,\ldots,M_k) \land\\
 &&\quad\quad\quad\quad\bigwedge_{i \leq k}M_i\in\cl{a_i}) \}\\
\cl{\Fun(f)} &=& \{M\mid \exists x, M',\,(M\bigstep\lambda x.\,M' \land\\
&&\quad\quad\quad\quad\forall b \in D_0\,(\rk(b) <\rk(\Fun(f))\to\\
&& \quad\quad\quad\quad\quad\forall N\in\cl{b}\,(M'[N/x]\in\cl{f(b)})))\} 
\end{eqnarray*}

Note that for $a \in D_0\cap E$ ($=\fdata(a)$), $M \in \cl{a}$  is 
equivalent to $M \mubprint a$.

\begin{lemma}
\label{lem-mon-cl}
For $a,b\in D_0$, if $a\dle b$, then $\cl{a}\supseteq\cl{b}$.
\end{lemma}
\begin{proof}
As the proof of Lemma~12 in \cite{Berger10}.
\end{proof}

 \begin{lemma}
 \label{lem-reducibility}
Suppose that $a \in D_0 \setminus \{\bot\}$.
 $M\in\cl{a}$ iff $M\bigstep V$ for some $V\in\cl{a}$.
\end{lemma}
\begin{proof}
Immediate from the definition of $\cl{a}$.
\end{proof}

\begin{lemma} 
\label{lem-rec}
If $M \in \cl{\Fun(f)}$, then $\rec\, M \in \cl{f^n(\bot)}$ for every $n \in \NN$.
\end{lemma}
\begin{proof}
  Induction on $n$.
  It is trivial for $n = 0$ because $\cl{\bot}$ contains every closed program.
  Suppose that $\rec\,{M}\in \cl{f^n(\bot)}$.
  According $\rk 2$,  for $b = f^n(\bot)$,  $f(b) = f(b_0)$ for some compact $b_0 \dle b$ with
  $\rk(\Fun(f)) > \rk(b_0)$.
  Since $\rec\,M \in \cl{b}$, we have $\rec\,M \in \cl{b_0}$ by Lemma \ref{lem-mon-cl}.
  Since $M \in \cl{\Fun(f)}$, $M \bigstep \lambda x. K$ for some $x$ and $K$ and $\forall c \in D_0 (\rk(c) < \rk(\Fun(f)) \to \forall N \in \cl{c} (K[N/x] \in \cl{f(c)})$.
  We apply this to the case $c = b_0$ and $N = \rec\ M$ and get
  $K[\rec\ M/x] \in \cl{f(b_0)} = \cl{f^{n+1}(\bot)}$.
  Therefore, 
  $K[\rec\ M/x] \bigstep V$ and $V \in \cl{f^{n+1}(\bot)}$.
  Thus, we also have $\rec\ M \bigstep V$ and therefore 
  $\rec\ M \in \cl{f^{n+1}(\bot)}$, by Lemma~\ref{lem-reducibility}.
\end{proof}

\begin{lemma}[Approximation] 
\label{lem-approximation}
For a closed program $M$ and  $a\in D_0$, if $a\dle \val{M}$, then 
$M\in\cl{a}$.
\end{lemma}
\begin{proof}
We show a more general statement about arbitrary programs involving
substitutions and environments to take care of free variables.
  A substitution is a finite mapping from variables to the set of closed programs.
  An environment is a finite mapping from variables to $D$.
  For a substitution $\theta$ and an environment $\eta$, 
  we write $\theta\in\cl{\eta}$ if $\eta(x)$ is compact and $\theta(x) \in \cl{\eta(x)}$ 
for each $x \in \dom(\theta)$.
We prove by induction on $M$:
\begin{quote}
For an environment $\eta$, a substitution $\theta$ such that
$\theta\in\cl{\eta}$, a program $M$ such that $FV(M) \subseteq \dom(\theta)$ 
and $a \in D_0$, if $a\dle \val{M}\eta$ then $M\theta\in\cl{a}$.  
\end{quote}

Since the statement is clear for $a = \bot$, we assume $a \ne \bot$.
We may also assume $M\neq\botexp$ since otherwise the condition 
$a \dle\valu{M}{\eta}$ is not satisfied. 
The cases that $M$ is  $x$,  $C(N_1,\ldots,N_k)$, 
$\case\,M'\,\of\,\{\ldots;C(\vec{y}) \to K;\ldots\}$,
$\lambda x.\,M'$,
$M'\ N$ are similar to the corresponding cases of 
Lemma~15 in \cite{Berger10}. 
We only consider the case $M = \rec\,N$.
Suppose that $a \dle \valu{M}\eta$.
Since $a\neq\bot$, $\valu{N}\eta = \Fun(g)$ for some continuous function 
$g:D \to D$ such that $\valu{M}\eta$ is the least fixed point of $g$.
Therefore, $a \dle  g^n(\bot)$ for some $n$.
By continuity,  there is a compact $f\in D \to D$ such that 
$f \dle g$ and  $a \dle  f^n(\bot)$.
Since $\Fun(f) \dle \valu{N}\eta$, by induction hypothesis, 
$N\theta \in \cl{\Fun(f)}$.
By Lemma \ref{lem-rec}, 
$\rec\,(N\theta) \in \cl{f^n(\bot)}$.
By Lemma \ref{lem-mon-cl}, $\cl{a} \supseteq \cl{f^n(\bot)}$.
Therefore, $M\theta = \rec\,(N\theta) \in \cl{a}$.
\end{proof}

\begin{proof}[Proof of the if part of Thm.~\ref{thm-adequacy}~(b)]
Suppose that $d \dlee \val{M}$ for a finite data $d$.
Then, $M\in\cl{d}$ by the Approximation Lemma.
Therefore, by the remark after the definition of $\cl{a}$, we have
$M \mubprint d$.
\end{proof}

\begin{lemma}
\label{lem-approximation2}
If $\valu{M}$ has the form
$C(a_1, \ldots, a_k)$, then $M\bs C(M_1,\ldots,M_k)$ for some $M_1,\ldots,M_k$.
\end{lemma}
\begin{proof}
  Let $a = C(\bot, \ldots, \bot)$.
  If $\valu{M}$ has the form $C(a_1, \ldots, a_k)$,
  then $a \dlee \valu{M}$.
By applying Thm.~\ref{thm-adequacy}~(b), we obtain $M \mubprint a$.
  Thus, $M\bs C(M_1,\ldots,M_k)$ for some $M_1,\ldots,M_k$.
\end{proof}

\begin{proof}[Completing the proof of the first Adequacy Theorem]
Finally, we prove the  
`if' part of (d) of Thm.~\ref{thm-adequacy}.
We prove 
by coinduction that 
$a \dlee \val{M}$ 
implies $M \nubprint a$.
Therefore, for  $a \in D$ and a closed program $M$, 
we show 
\begin{align*}
a \dlee \val{M} 
\to\  &a = \bot\ \lor \\
&  \bigvee_C \left(\begin{array}{l}
\exists M_1,\ldots,M_k, a_1, \ldots,a_k\ (M \bs C(M_1,\ldots,M_k)\\  
\hspace*{2cm} \land a = C(a_1,\ldots, a_k) \land \bigwedge_{i \leq k} a_i \dlee \val{M_i})
                       \end{array}\right)\,.
\end{align*}
Suppose that $a \dlee \val{M}$.
Since this implies $a \in \data$, it follows that $a = \bot$ or $a$ has the form $C(a_1,\ldots,a_k)$ for $a_i \in E$.
If $a = \bot$ we are done.
If $a = C(a_1,\ldots, a_k)$,  then $\val{M}$ also has the form $C(a_1',\ldots, a_k')$  for some $a_i' \sqsupseteq_E a_i$.
Therefore,  we can apply Lemma~\ref{lem-approximation2} and obtain
$M\bs C(M_1,\ldots,M_k)$ for some $M_1,\ldots,M_k$.
By Lemma~\ref{lem-correctness}~(a), we have $\val{M} = 
\val{C(M_1,\ldots,M_k)}
= C(\val{M_1},\ldots,\val{M_k}).$
Therefore, $a_i \dlee a_i' = \val{M_i}$.
\end{proof}

\subsection{Computation of infinite data}
\label{section:computationinfinite}

Thm.~\ref{thm-adequacy}~(c) and~(d) characterize the denotational 
semantics of a program $M$ in terms of the relations $M \nuprint a$ and
$M \nubprint a$ which have a more proof-theoretic rather than operational
character since they are defined by (possibly infinite) derivations.
In this section we define a notion of possibly infinite step-by-step computation
that continues under data constructor and prove
our second Adequacy Theorem (Thm.~\ref{thm-adequacytwo})
which provides a truly operational characterization of the denotational
semantics of a program.

As one can see from Thm.~\ref{thm-adequacy}~(d), the reduction 
relation $M \nubprint a$ is not functional and
a program $M$ is related to a set of data whose 
upper bound is the denotational semantics
of $M$.  To obtain a more precise operational notion,
we use the following inductively defined smallstep leftmost-outermost reduction 
relation $\ssp$ on closed programs that corresponds to bigstep reduction. 
\begin{itemize}
\item[(i)] $\case\,C(\vec M) \,\of\, \{\ldots;C(\vec y)\to N;\ldots\}\ssp
                                                            N[\vec M/\vec y]$
\item[(ii)] $(\lambda x.\,M)\ N \ssp M[N/x]$
\item[(iii)] $\rec\,M \ssp M\,(\rec\,M)$
\item[(iv)] \AxiomC{$M \ssp M'$}
             \UnaryInfC{$\case\,M\,\of\, \{\vec{Cl}\}\ssp 
               \case\,M'\,\of\, \{\vec{Cl}\}$}
            \DisplayProof 
\item[(v)] \AxiomC{$M \ssp M'$}
            \UnaryInfC{$M\,N \ssp M'\,N$}
            \DisplayProof 
\end{itemize}
Since we are only concerned with reducing closed terms 
the substitutions in (i) and (ii) do not need $\alpha$-conversions.
\begin{lemma}\label{lemma:bsssp}
If $M \bs V$, then $M \ssp^* V$. 
\end{lemma}
\begin{proof}
The proof is by induction on the definition of $M \bs V$.

If $M=V$, then the assertion is trivial.

If $M = \case\,M'\,\of\,\{\ldots;C(\vec y)\to N;\ldots\} $,
then $M' \bs C(\vec M) $ and
$N[\vec M/\vec y])\bs V$.
By the induction hypothesis, $M'  \ssp^* C(\vec M) $  and
$N[\vec M/\vec y] \ssp^* V$.
We have
\begin{align*}
   M &= \case\,M'\,\of\,\{\ldots;C(\vec y)\to N;\ldots\} \\
&\ssp^* \case\,C(\vec M)\,\of\,\{\ldots;C(\vec y)\to N;\ldots\} \\
&\ssp N[\vec M/\vec y] \ssp^* V\,.
\end{align*}

If $M= M_1\,N $, then $M_1 \bs \lambda x.\,M' $ and
$M'[N/x]\bs V$. By the induction hypothesis,
$M_1  \ssp^* \lambda x.\,M' $ and 
$M'[N/x]\ssp^* V$.
Therefore,
$M = M_1\, N  \ssp^* (\lambda x.\,M')\, N  \ssp M'[N/x]\ssp^* V$.

If $M= \rec\,M' $, then $M'\,(\rec\,M') \bs V$. 
We have
\[ M = \rec\,M'  \ssp M' \,(\rec\,M' ) \ssp^* V,\]
by the induction hypothesis.
\end{proof}

In order to approximate the denotational semantics operationally,
we need to continue computation under constructors. 
Since a constructor may have more than one argument
and some computations of arguments may diverge,
we need to compute all the arguments in parallel.
For this purpose, we extend the smallstep reduction $\ssp$ 
to a relation $\newprintp$ by the following inductive rules:
\begin{center}
\AxiomC{$M \ssp M' $} 
\UnaryInfC{$M \newprintp M'$}
\DisplayProof 
\quad
\AxiomC{$M_i \newprintp M_i'$ $(i = 1,\ldots, k)$} 
\UnaryInfC{$C(M_1,\ldots,M_k) \newprintp C(M_1',\ldots,,M_k')$}
\DisplayProof\\[1em] 
$M \newprintp M$ \hspace{.5cm} otherwise.
\end{center}
Clearly there is exactly one applicable rule for each closed program $M$.
We denote by $M^{(n)}$ the unique program
$M'$ such that  
$M (\newprintp)^n M'$.

For a closed program $M$, we define 
${M}_\bot \in E$ as follows.
\begin{align*}
{C(M_1,\ldots,M_k)}_\bot &= C({M_1}_\bot,\ldots, {M_k}_\bot)\\
{M}_\bot &= \bot  \hspace*{.5cm} \mbox{if $M$ is not a constructor term}
\end{align*}

\begin{lemma}[Accumulation]
\label{lem:order}
  If $M \newprintp M'$, then ${M}_\bot \dlee {M'}_\bot$.  Therefore,
  ${M^{(n)}}_\bot \dlee   {M^{(m)}}_\bot$ for $n \leq m$.
\end{lemma}
\begin{proof}
Immediate by the definition of $\newprintp$.
\end{proof}

For a closed program  $M$, ${M^{(n)}}_\bot$ can be viewed as the finite 
approximation of the value of $M$ obtained after $n$ consecutive parallel 
computation steps.
The following lemma shows that this computation is complete, 
that is, \emph{every} 
finite approximation 
is obtained eventually. 
\begin{lemma}[Adequacy for finite values]
\label{lem:equiv}
If $M \mubprint a $, then 
$\exists n\, a \dlee {{M^{(n)}}}_\bot$.
\end{lemma}
\begin{proof}
Let $P(M, a) \eqdef \exists n\, a \dlee {M^{(n)}}_\bot$.
We prove by induction that 
$M \mubprint a$ implies  $P(M, a)$.
That is,  
we show 
\[
\ocl_\bot(P)(M, a) \to P(M, a).
\]
If $a = \bot$, then we have $P(M, a)$.
If we have 
\[
(M \bs C(M_1,\ldots,M_k)) \land a = C(a_1,\ldots, a_k) \land \bigwedge_{i \leq k} (\exists n_i\  a_i \dlee 
{{M_i}^{(n_i)}}_\bot)
\]
for a constructor $C$,
then, for $n$ the maximum of $n_i$ ($i \leq k$),
$a = C(a_1,\ldots, a_k) \dlee C({{M_1}^{(n)}}_\bot,\ldots,{{M_k}^{(n)}}_\bot)
= {C(M_1 ,\ldots,M_k )^{(n)}}_\bot$
by Lemma~\ref{lem:order}.
On the other hand,  by Lemma~\ref{lemma:bsssp}, we have
$M^{(m)} = C(M_1,\ldots,M_k) $ for some $m$. 
Therefore, $a \dlee {M^{(m+n)}}_\bot$.
\end{proof}

Since ${{M^{(n)}}}_\bot$ is an increasing sequence by Lemma~\ref{lem:order}, we can define
\[M^{(\infty)} = \bigsqcup_n\, {M^{(n)}}_\bot.\]
We say that the program $M$ \emph{infinitely computes} 
the data $M^{(\infty)}$.

For $d \in D$ we define the data-part $d_E \in E$ as follows.
\begin{align*}
  \bot_E &= \bot\\
  C(d_1,\ldots,d_k)_E &= C((d_1)_E,\ldots,(d_k)_E)\\
  \Fun(f)_E &= \bot
\end{align*}
Clearly, the function $ d \mapsto d_E$ is a projection of $D$ onto $E$.
\begin{theorem}[Computational Adequacy II]
\label{thm-adequacytwo}
$M^{(\infty)} = \val{M}_E$ for every closed program $M$. 
\end{theorem}
\begin{proof}
  It is easy to show the following.
\begin{itemize}
\item[(a)] If $M \newprintp M'$ then $\val{M}_E = \val{M'}_E$ .
\item[(b)] ${M}_\bot \dlee \val{M}_E$.
\end{itemize}
Therefore, ${{M^{(n)}}}_\bot \dlee \val{M}_E$.
Since this holds for every $n$,  we have
  $M^{(\infty)} \dlee \val{M}_E$.

  By Thm.~\ref{thm-adequacy}~(d), $M \nubprint \val{M}_E$
because $\val{M}_E \dlee \val{M}$. 
  Therefore, by Lemma~\ref{lemma:nutomu}~(d), 
$\forall d\ (\appr{d}{\val{M}_E} \to M \mubprint d)$,
and consequently, by Lemma~\ref{lem:equiv},
$\forall d\ (\appr{d}{ \val{M}_E} \to \exists n\  d \dlee {M^{(n)}}_\bot)$.
Since $d \dlee {M^{(n)}}_\bot \to \appr{d}{M^{(\infty)}}$,
we have  $\forall d\ (\appr{d}{ \val{M}_E} \to \appr{d}{ M^{(\infty)}})$.
  Therefore, $\val{M}_E \dlee M^{(\infty)}$ by Lemma~\ref{lemma:fdata}~(b).

\end{proof}

Note that if $\val{M} \in E$, then we have $\val{M}_E = \val{M}$.
Therefore, the second Adequacy Theorem says 
$M^{(\infty)} = \val{M}$ in this case.

\subsection{Data extraction}
\label{sub-data}
Using types we are able to identify criteria under which an extracted 
program denotes an observable data, i.e.\ an element of $E$. 
\begin{lemma}
\label{lem-typ-E}
If $\rho$ is a type that contains no function type and $\zeta$ is a type environment
such that $\zeta(\alpha)\subseteq E$ for all 
type variables $\alpha$ in the domain of $\zeta$, 
then $\tval{\rho}{\zeta}\subseteq E$.
\end{lemma}
\begin{proof}
Structural induction on $\rho$. The only non-obvious case is $\tfix{\alpha}{\rho}$.
By the definition of $\tval{\tfix{\alpha}{\rho}}{\zeta}$, and
since $E$ is a subdomain of $D$,  it suffices to show
$\tval{\rho}{\zeta[\alpha\mapsto E]}\subseteq E$. But this holds by the induction
hypothesis.
\end{proof}

We call an $\IFP$-formula a \emph{data formula} if it contains no free 
predicate variable and no strictly positive subformula 
of the form $A\to B$ where $A$ and $B$ are non-Harrop.

\begin{theorem}[Data Extraction]
\label{thm-pe}
From a proof in $\IFP$ of a data formula $A$ from Harrop assumptions 
$\Gamma$ one can extract a closed 
program $M$ realizing $A$, provably in $\RIFP$ from $\reah(\Gamma)$. 
Moreover, 
$M$ is a data that can hence be infinitely 
computed, that is, $M^{(\infty)}=\val{M}$.
\end{theorem}
\begin{proof}
By the Soundness Theorem (Thm.~\ref{thm-soundness}) 
we can extract a closed program $M:\tau(A)$ such that
$\RIFP$ proves $\reah(\Gamma)\vdash\ire{M}{A}$.
Clearly, since
$A$ is a data formula, $\tau(A)$ contains no function type.
Therefore, by Lemma~\ref{lem-typ-E}, $M$ denotes a data.
By the second Adequacy Theorem (Thm~\ref{thm-adequacytwo}),
$M^{(\infty)}=\val{M}$. 
\end{proof}

\begin{example}\label{ex3}
In Thm.~\ref{thm-c-g'}, we proved $\C\subseteq\G$ and obtained a program 
$\stog$ as its realizer.  On the other hand, one can prove 
$\C(1)$ by showing 
$\{1\} \subseteq \C$ 
by coinduction. 
From the proof, we can extract the realizer $a \eqrec \Pair(1, a)$ 
(i.e., 
$a = 1\!:\!1\!:\!\ldots$
) of $\C(1)$.
From $\C\subseteq\G$ and  $\C(1)$, we can trivially prove $\G(1)$ 
and from these proofs we can extract a realizer 
$\mathsf{M_1} = \stog\, (1\!:\!1\!:\!\ldots)$ 
of $\G(1)$.
With the small-step reduction rule,  one can compute
$$\mathsf{M_1}\, \newprintptr\, \RG \!:\! N_1
\ \newprintptr\  \RG\!:\! \LG\!:\! N_2 \ \newprintptr\  \RG\!:\! \LG \!:\! \LG\!:\! N_3 \ \newprintptr\  
 \ldots$$
for some $N_i (i \geq 1)$.
Taking $(\underline{\ })_\bot$ of these terms, we have an increasing sequence
$$\bot,\ \ \RG\!:\!\bot,\ \  \RG\!:\! \LG\!:\!\bot, \ \ \RG\!:\! \LG\!:\! \LG\!:\!\bot,\ \ \ldots$$
Taking the limit of these terms, one can see that $\mathsf{M_1}$ infinitely computes 
the data 
$\mathsf{M_1}^{(\infty)} = \RG\!:\!\LG\!:\!\LG\!:\!\ldots$, 
which is a realizer of $\G(1)$ by Thm.~\ref{thm-pe}.

While for $\C(1)$ there was only one canonical proof and one realizer, we
now look at $\C(1/2)$ which has more than one canonical 
proof and realizer and will
give rise to three Gray codes, one with an undefined digit.
By the coclosure axiom, $\C(1/2)$ unfolds to 
$\exists d\in\SD\,(1/2\in\II_d\land\C(2\cdot 1/2 - d))$.
Therefore, we can choose $d=0$ and use the above proof of $\C(1)$. 
This yields a realizer 
$0\!:\!1\!:\!1\!:\!\ldots$ 
of $\C(1/2)$, and 
$\mathsf{M_{1/2}} = \stog\, (0\!:\!1\!:\!1\!:\!\ldots)$ 
is a realizer of  $\G(1/2)$.
One can see that
\begin{align*} 
\mathsf{M_{1/2}} \ \newprintptr\  N_1\!:\!\RG\!:\!N_2 \ \newprintptr\  \RG\!:\!\RG\!:\!N_3 \ \newprintptr\  \RG\!:\!\RG\!:\!\RG\!:\!N_4 \ \newprintptr\  
\end{align*}
for some $N_i (i \geq 1)$.
Therefore,  the result of finite-time computation proceeds 
$$\bot,\ \ \bot\!:\!\RG\!:\!\bot,\ \ \RG\!:\!\RG\!:\!\bot,\ \ \RG\!:\!\RG\!:\!\RG\!:\!\bot,\ \ \RG\!:\!\RG\!:\!\RG\!:\!\LG\!:\!\bot, \ \ldots$$
and in the limit, we have 
$\mathsf{M_{1/2}}^{(\infty)} = \RG\!:\!\RG\!:\!\RG\!:\!\LG,\LG\!:\!\ldots$.

Since $1\!:\!0\!:\!0\!:\!\ldots$ 
is another realizer of $\C(1/2)$,
$\mathsf{M'_{1/2}} = \stog(1\!:\!0\!:\!0\!:\!\ldots)$ is also a  realizer of  $\G(1/2)$.
One can see that
$$\mathsf{M'_{1/2}} \ \newprintptr\  \RG\!:\!N_1\ \newprintptr\  \RG\!:\!N_2\!:\!\RG\!:\!N_3\ \newprintptr\  \RG\!:\!N_4\!:\!\RG\!:\!\LG\!:\!N_5 \ \newprintptr\  \ldots$$
for some $N_i (i \geq 1)$.
Therefore,  one can observe the finite approximations
$$\bot,\ \ \RG\!:\!\bot,\ \ \RG\!:\!\bot\!:\!\RG\!:\!\bot,\ \ \RG\!:\!\bot\!:\!\RG\!:\!\LG\!:\!\bot,\ \ldots$$
hence $\mathsf{M_{1/2}'}$ computes the partial infinite data
$\mathsf{M_{1/2}'}^{(\infty)} = \RG\!:\!\bot\!:\!\RG\!:\!\LG\!:\!\LG\!:\!\ldots$.
\end{example}


\section{Conclusion}
\label{sec-conclusion}
We presented $\IFP$, a formal system supporting
program extraction from proofs in abstract mathematics. 
$\IFP$ is plain many-sorted first-order logic extended with two 
extra constructs for strictly positive inductive and coinductive 
definitions that are dual to each other.
Sorts in $\IFP$ represent abstract structures 
specified by (classically true) disjunction free closed axioms.
Hence full classical logic is available.
Computational content is extracted through a realizability 
interpretation that treats quantifiers uniformly in order to permit
the interpretation of sorts as abstract spaces.
The target language of the interpretation is a functional 
programming language in which extracted programs are typable
and therefore easily translatable into Haskell and executed there.
The exact fit of the denotational and operational semantics of the target language is
proven by two computational adequacy theorems. 
The first 
(Thm~\ref{thm-adequacy}) states that 
all compact approximations of the denotational value of a program 
can be computed, the second (Thm~\ref{thm-adequacytwo}) states 
that the full (possibly infinite) denotation value can be computed 
through successive computation steps.
It should be stressed that axioms used in a proof
do not show up as non-executable constants in extracted programs
and therefore do not spoil the computation of programs 
into canonical form.  
Besides the natural numbers as a primary example of a strictly
positive inductive definition we studied wellfounded induction 
and useful variations thereof such as Archimedean induction.

In an extended case study we 
formalized in $\IFP$ the real numbers as an Archimedean real closed field
and introduced various exact real number representations 
(Cauchy and signed digit representation as well as infinite Gray code) 
as the realizability
interpretations of simple coinductive predicates ($\A$, $\C$, and $\G$).
From a proof that $\C$ is a subset of $\G$ we extracted a program converting 
the signed digit representation into infinite Gray code.
There is an experimental Haskell implementation of $\IFP$ and its program 
extraction called Prawf~\cite{BergerPetrovskaTsuiki20} where this is carried out.

This case study highlights some crucial features of $\IFP$:
\begin{itemize}
\item The real numbers are given axiomatically as an abstract structure;
\item signed digit representation and infinite Gray code are obtained
as realizers of coinductive predicates $\C$ and $\G$;
\item Archimedean induction is used to prove that the sign of 
non-zero reals in $\C$ can be decided (first part of the proof
of~Thm.~\ref{thm-c-g'}); 
\item the definition of $\G$ permits partial realizers (which are inevitable 
for infinite Gray code);
\item the second Adequacy Theorem is applied to compute full
infinite Gray code in the limit.
\end{itemize}
This case study not only puts to test the practical usability 
of $\IFP$  but also leads to the study of possible extensions of it.
Having extracted a program realizing the inclusion $\C\subseteq\G$
it is natural to ask 
about the reverse inclusion.
In~\cite{Tsuiki02} a parallel and nondeterministic program converting 
infinite Gray code into signed digit representation is given
which is necessarily parallel and nondeterministic \cite{Tsuiki05}.
Since the programming language of $\RIFP$ doesn't have these features 
such conversion cannot be extracted.
We leave it for further work to develop a suitable extension of
our system improving and extending previous work
in this direction~\cite{BergerMiyamotoSchwichtenbergTsuiki16,BergerCSL16}.
A further interesting line of study will be the extraction of
algorithms that operate on compact sets of real numbers as studied 
in~\cite{BergerSpreen16,Spreen20}.
%


\bibliographystyle{plainnat}
\bibliography{refs}


\end{document}